\documentclass[acmtog,nonacm,dvipsnames]{acmart}

\setcopyright{none}

\graphicspath{{figures_reduced/}}

\usepackage{color}
\usepackage{comment}
\usepackage{makecell}

\usepackage{subfigure}
\usepackage{wrapfig}
\usepackage{multirow}
\usepackage{xspace}
\long\def\symbolfootnote[#1]#2{\begingroup%
\def\thefootnote{\fnsymbol{footnote}}\footnote[#1]{#2}\endgroup}

\usepackage{amsthm}
\newtheorem{definition}{Definition}
\newtheorem{theorem}{Theorem}

\newcommand{\ba}{\mathbf{a}}
\newcommand{\bb}{\mathbf{b}}
\newcommand{\bc}{\mathbf{c}}
\newcommand{\mybf}{\mathbf{f}}
\newcommand{\bp}{\mathbf{p}}
\newcommand{\bq}{\mathbf{q}}
\newcommand{\br}{\mathbf{r}}
\newcommand{\bn}{\mathbf{n}}
\newcommand{\bu}{\mathbf{u}}
\newcommand{\bv}{\mathbf{v}}
\newcommand{\bw}{\mathbf{w}}
\newcommand{\bx}{\mathbf{x}}
\newcommand{\cS}{\mathcal{S}}
\newcommand{\cC}{\mathcal{C}}
\newcommand{\cM}{\mathcal{M}}

\newcommand{\cP}{\mathcal{P}}

\newcommand{\sF}{\mathsf{F}}
\newcommand{\sB}{\mathsf{B}}
\newcommand{\sC}{\mathsf{C}}

\newcommand{\bh}[1]{[BH\S#1]}

\setcopyright{acmcopyright}
\copyrightyear{2022}
\acmYear{2022}
\acmDOI{XXXXXXX.XXXXXXX}

\acmJournal{TOG}
\acmVolume{0}
\acmNumber{0}
\acmArticle{0}
\acmMonth{0}

\begin{document}

\title{ConTesse: Accurate Occluding Contours for Subdivision Surfaces}

\author{Chenxi Liu}
\orcid{0000-0003-3613-1662}
\affiliation{\institution{University of British Columbia}\country{Canada}}

\author{Pierre B\'{e}nard}
\orcid{0000-0002-2846-1955}
\affiliation{\institution{Univ. Bordeaux, CNRS, Bordeaux INP, INRIA, LaBRI, UMR 5800}\country{France}}

\author{Aaron Hertzmann}
\orcid{0000-0001-9667-0292}
\affiliation{\institution{Adobe Research}\country{USA}}

\author{Shayan Hoshyari}
\affiliation{\institution{Adobe}\country{USA}}

\renewcommand{\shortauthors}{Liu, B\'enard, Hertzmann, Hoshyari}

\begin{abstract}
  This paper proposes a method for computing the visible occluding contours of subdivision surfaces. The paper first introduces new theory for contour visibility of smooth surfaces.
  Necessary and sufficient conditions are introduced for when a sampled occluding contour is valid, that is, when it may be assigned consistent visibility. 
  Previous methods do not guarantee these conditions, which 
  helps explain why smooth contour visibility has been such a challenging problem in the past. The paper then proposes an algorithm that, given a subdivision surface, finds sampled contours satisfying these conditions, and then generates a new triangle mesh matching the given occluding contours. The contours of the output triangle mesh may then be rendered with standard non-photorealistic rendering algorithms, using the mesh for visibility computation. 
  The method can be applied to any triangle mesh, by treating it as the
base mesh of a subdivision surface.
\end{abstract}



\newcommand{\tfigwidth}{1.5in}
\newcommand{\tfigheight}{.75in}

\begin{teaserfigure}
\centering
\begin{tabular}{c@{}c@{}c@{}c}
\includegraphics[width=\tfigwidth]{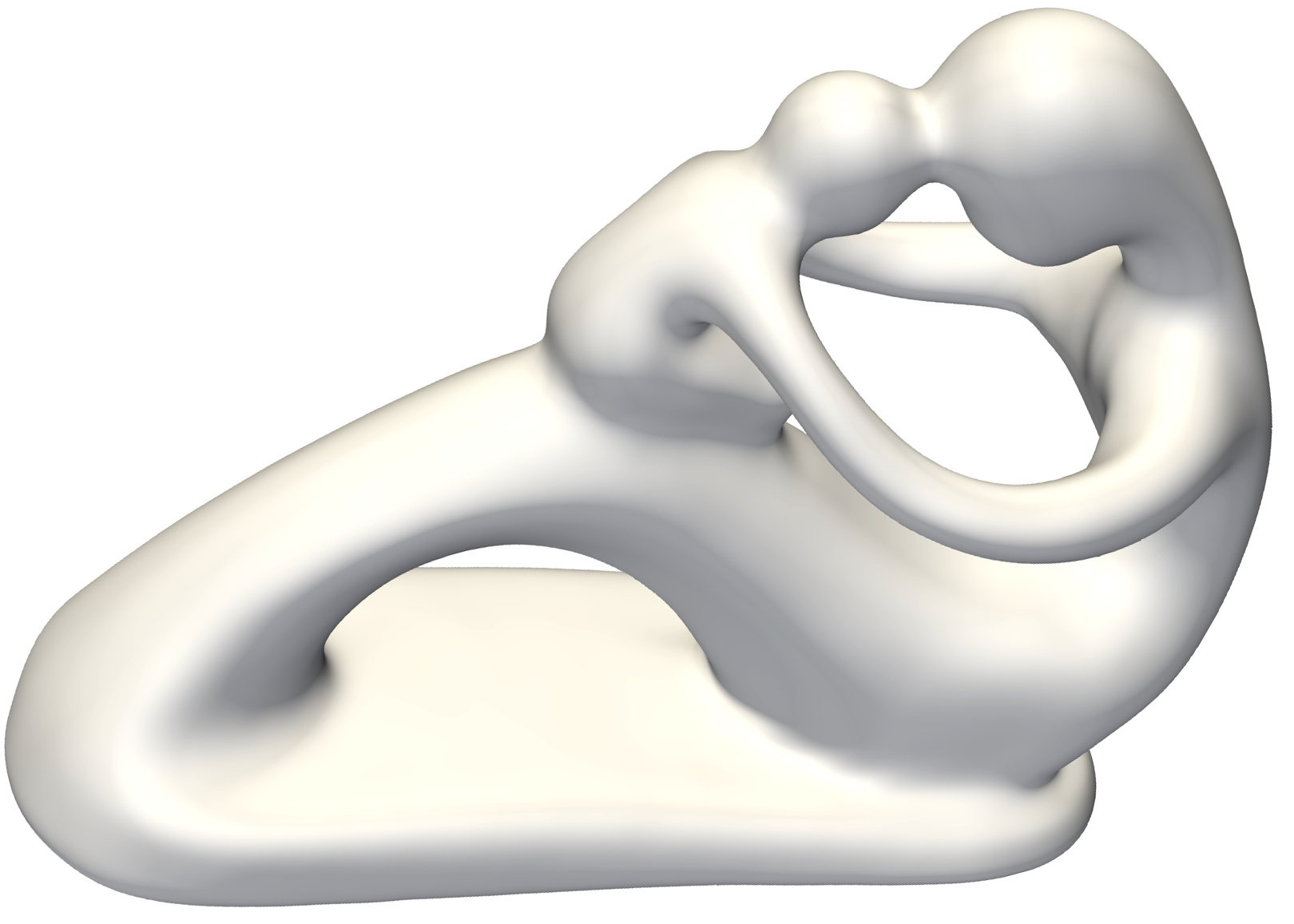} &
\includegraphics[width=\tfigwidth]{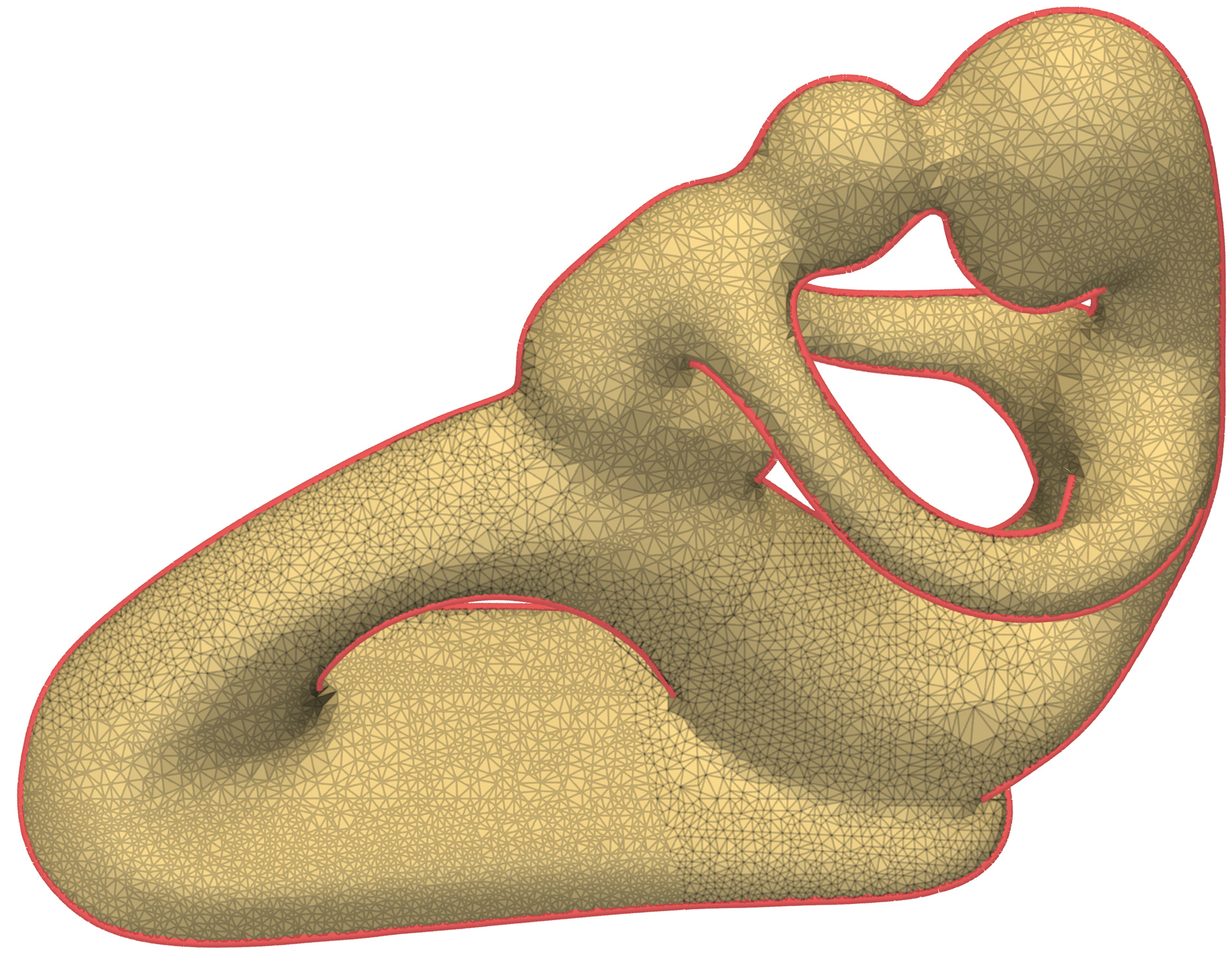}%
\includegraphics[width=\tfigheight]{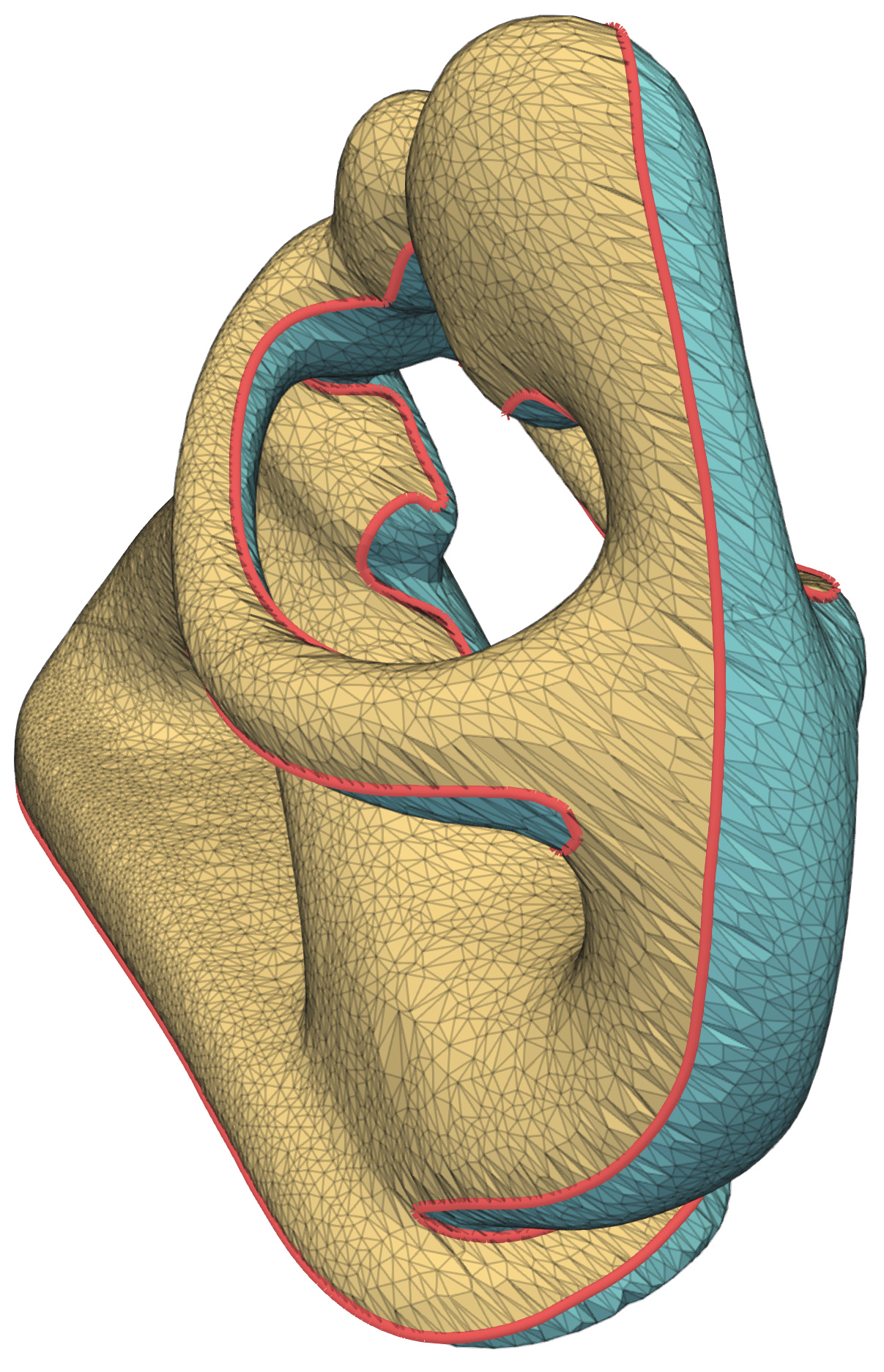}  &
\includegraphics[width=\tfigwidth]{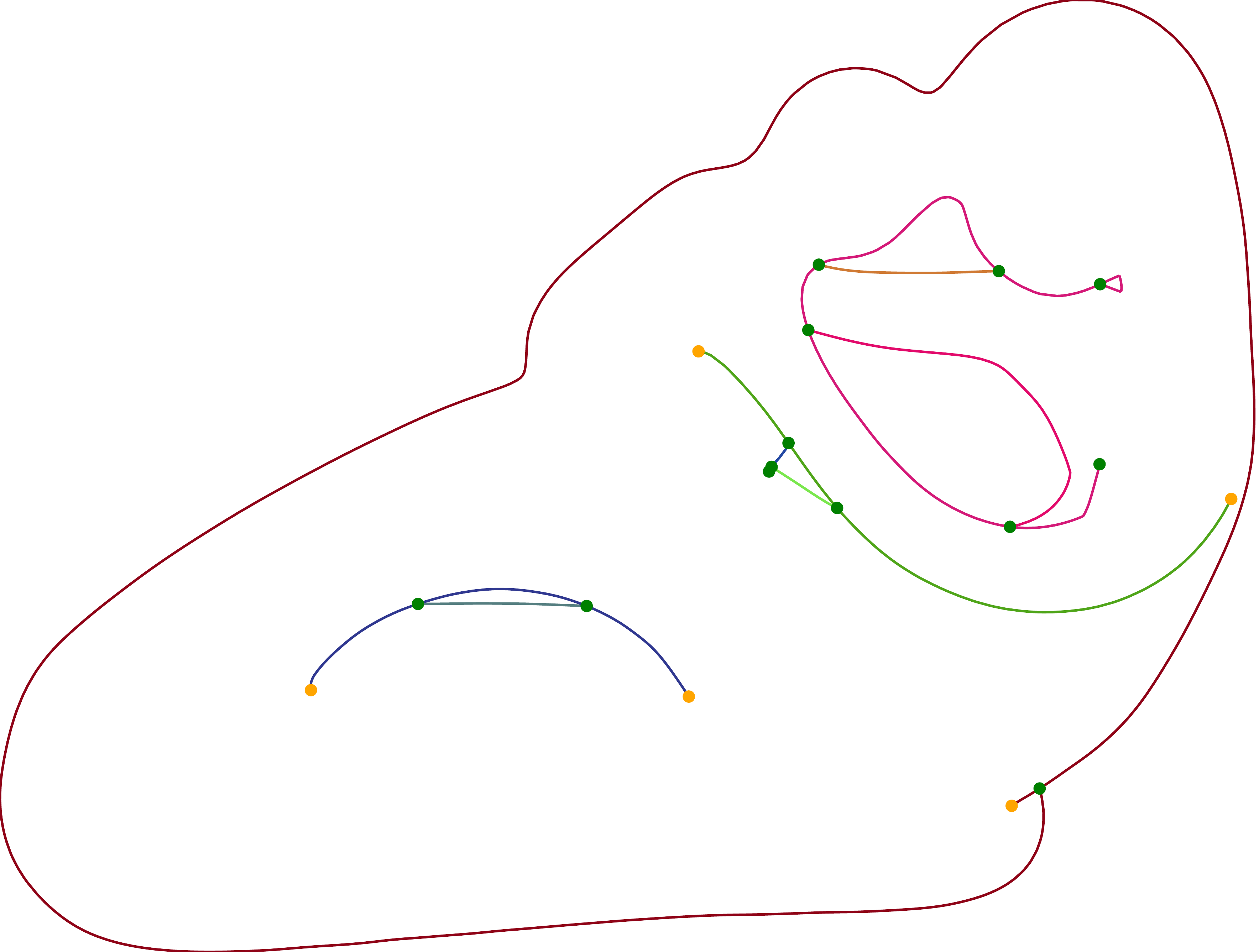} &
\includegraphics[width=\tfigwidth]{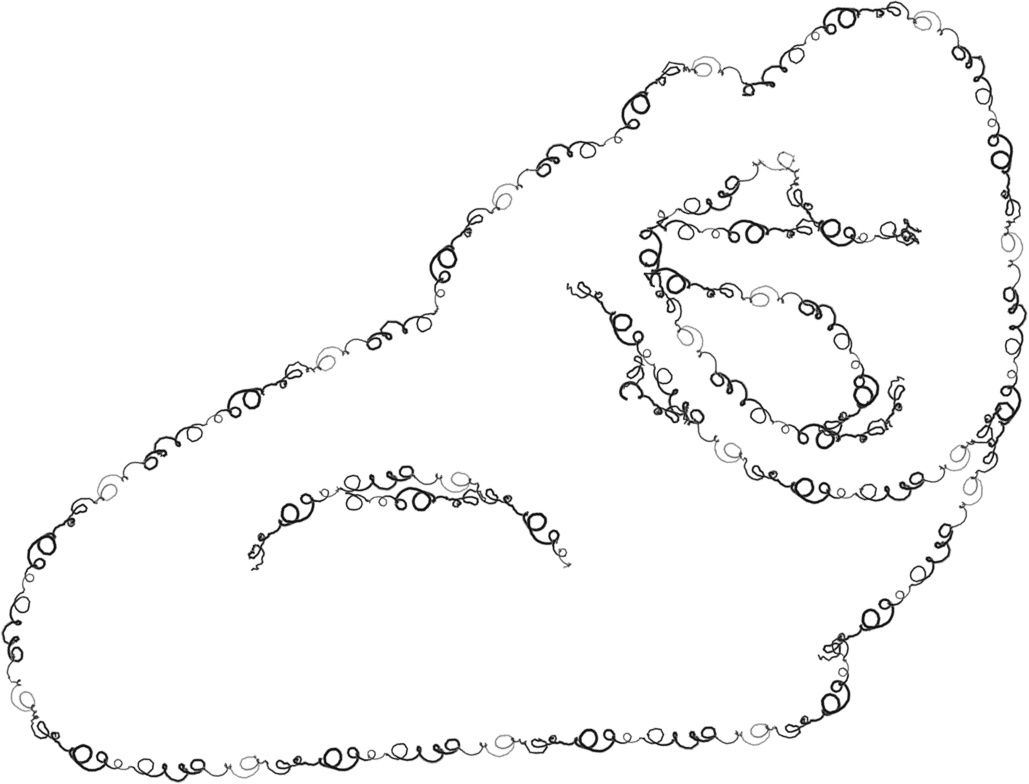} \\
(a) & (b) & (c) & (d)
\end{tabular}
    \caption{Given (a) a smooth 3D surface $\cS$ and a camera viewpoint, our method produces (b) a triangle mesh $\cM$ where the occluding contour of the mesh accurately approximates the occluding contour of the smooth surface. Standard algorithms may then be used to extract (c) the view map of occluding contours, and to (d) stylize them.
    (Fertility courtesy UU from AIM@SHAPE-VISIONAIR Shape Repository).
    }
    \label{fig:teaser}
\end{teaserfigure}

\maketitle

\section{Introduction}

Computing occluding contours is one of the fundamental building blocks of 3D Non-Photorealistic Rendering, because these curves allow us to produce beautiful and stylized artistic images and animations in many styles \cite{Grabli:2010}. 
The occluding contours follow occlusion boundaries in an image, i.e., where one part of a surface occludes another.  These contours accurately capture many of the lines that artists draw \cite{Cole:2008}. 

The problem of computing the occluding contours for smooth surfaces dates back to the earliest days of computer graphics \cite{Weiss:1966:VPI:321328.321330,Appel:1967}.
Yet, despite more than a half-century of effort, this problem remains unsolved for smooth surfaces \cite{BenardHertzmann}. 
No existing method guarantees results that both accurately capture the occluding contours (e.g., no curves missing or merged), while also producing topologically-consistent results (e.g., no gaps in the silhouette).  

The most accurate current method is that of B\'{e}nard et al.~\shortcite{Benard:2014}. 
This method generates a new mesh, so that the contour of the new mesh approximates the contours of the input surface. 
Then, existing methods may be used to compute the occluding contour of the mesh, in a consistent fashion.
However, this method is extremely slow and often produces meshes with a few spurious contours, which could, in theory, lead to some incorrect results.  Moreover, some styles require an accurate planar map \cite{Eisemann:2008,Grabli:2010,Winkenbach:1994}, which would be impossible to guarantee with such errors. 
Importantly, there remain a key theoretical gap: when can we expect that a valid mesh $\cM$ exists?  

This paper introduces new theoretical and practical advances for computing the visible occluding contour for smooth surfaces.
We first derive a new theory that describes when there exists a 3D mesh $\cM$ consistent with a set of sampled occluding contours $\cC$. 
Previous methods produce curves that do not guarantee these conditions, and thus often produce curves without any possible valid visibility. 
These observations provide new insight into why contour visibility has been such a challenging problem, and how to address it.

We then propose \textsf{ConTesse} (contour tessellation), a new method for computing the visible occluding contours of a subdivision surface.  The method can be applied to any triangle mesh, by treating it as the base mesh of a subdivision surface.
Like previous methods, we first sample the occluding contour into piecewise-linear  curves (polylines). Our method then refines these polylines until they satisfy our new validity conditions.
The method then produces a new triangle mesh that fits these sampled contours, based on a new image-space approach for ensuring consistent triangle orientation. Finally, standard algorithms compute and stylize the occluding contours, using the  triangle mesh to determine visibility. We show that the resulting method is substantially faster and less memory-intensive than that of B\'{e}nard et al.~\shortcite{Benard:2014}, and, arguably, conceptually simpler, while producing higher-quality results.  Finally, we discuss possible future improvements to our approach.

\subsection{Problem Definitions and Background}

We now briefly review the problem of determining the visible occluding contour, and why it is difficult.
For an in-depth tutorial on occluding contour algorithms, see \cite{BenardHertzmann}. For brevity, in the rest of this paper, we refer to sections of \cite{BenardHertzmann} with concise notation, e.g., \bh{3} for Section~3.

Suppose we view a 3D surface from a camera position $\bc$, in general position.
For a point $\bp$ on the surface with associated surface normal $\bn$, the orientation of the point is 
\begin{equation}\label{eq:contour}
    g(\bp) = (\bc - \bp) \cdot \bn  \,.
\end{equation}
Points are front-facing when $g(\bp)>0$, and back-facing when $g(\bp)<0$.
Assuming that back-faces are never visible, the boundaries between the front-facing and back-facing regions correspond to occlusion boundaries in the image.  Specifically, for a smooth surface, the \textit{contour generator} is the set of points where $g(\bp)=0$, and, for a triangle mesh, the contour generator is the set of mesh edges that connect front faces to back faces (Figure \ref{fig:ocg}).

The \textit{occluding contour} (or apparent contour) is defined as the image-space projection of the visible portion of the  contour generator (Figure \ref{fig:ocg}). Note that visibility is part of this definition; hidden points are not part of the occluding contour. 
\newcommand{\pswidth}{1.5in}
\begin{figure}
    \centering
    (a)
    \includegraphics[width=0.43\linewidth]{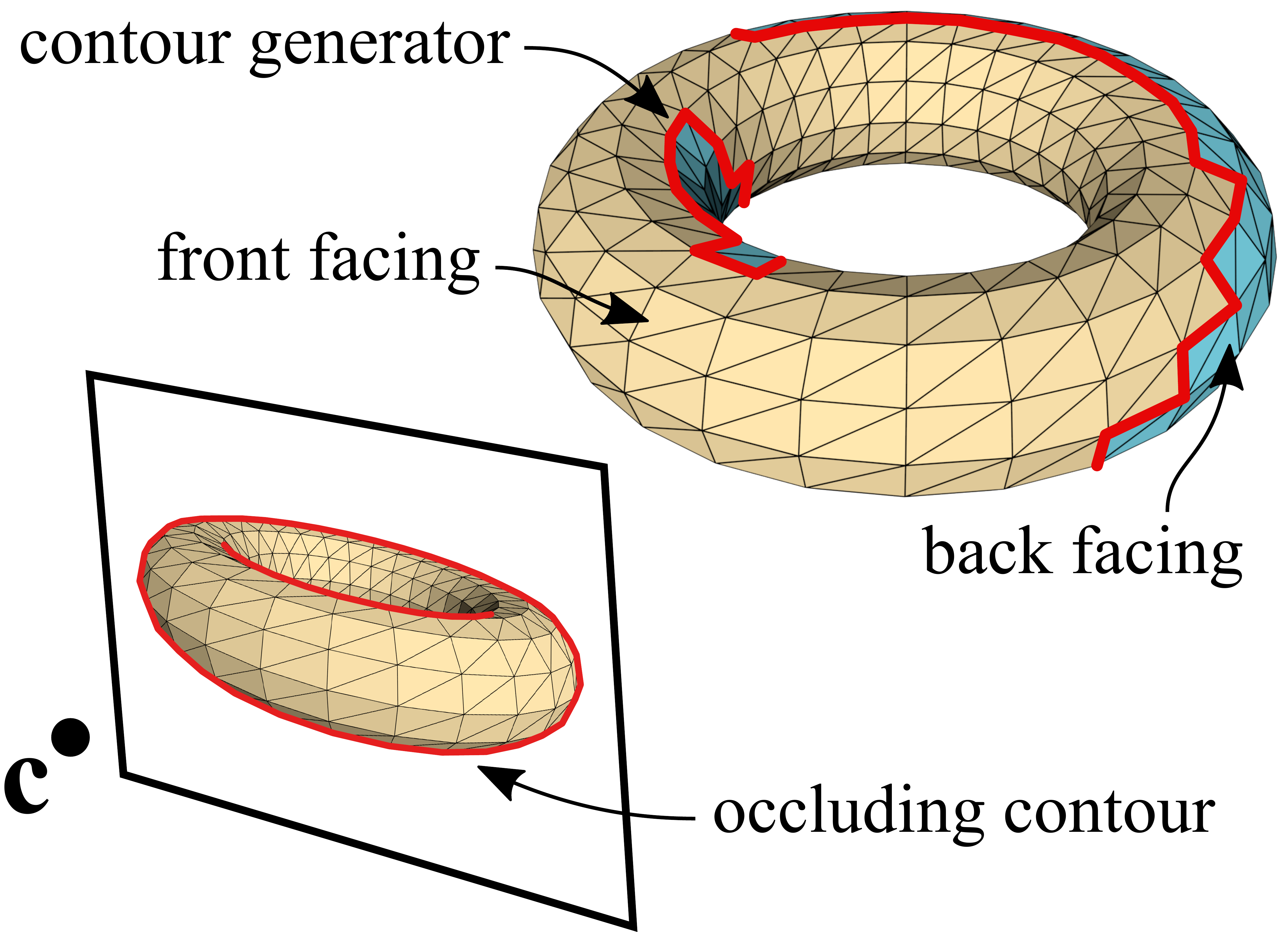}
    ~%
    (b)
    \includegraphics[width=0.43\linewidth]{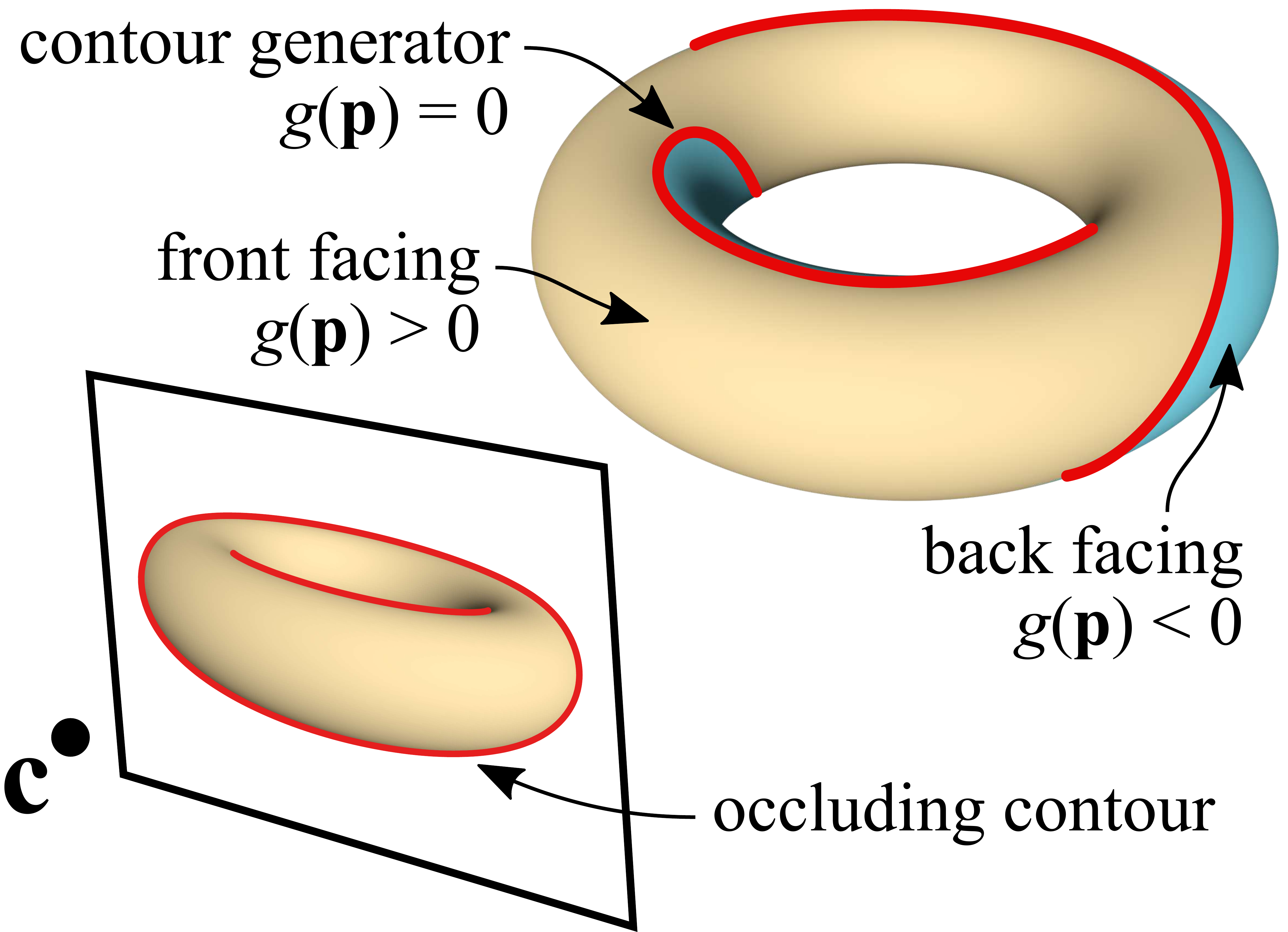}
    \caption{Definition of occluding contour (or apparent contour) and occluding contour generator, for meshes (a) and smooth surfaces (b), from \bh{3.4,7.2}. The occluding contour describes occlusion boundaries in the image.
    Throughout the paper, we use the follow color scheme: 
    contour points are \textcolor{red}{\bf red}, front-facing points are \textcolor{Orange}{\bf yellow}, and back-facing points are \textcolor{Blue}{\bf blue}.
    }
    \label{fig:ocg}
\end{figure}

This paper focuses on the problem of computing the  occluding contour for a smooth surface. At first, this task seems like it ought to be quite simple.  Yet, an extensive body of literature has identified this problem and attempted to solve it by a variety of clever approaches, none of which guarantee correct results. 

The simplest approach is to take a triangle mesh as input, and output the  mesh's occluding contours. However, when the mesh represents a smooth surfaces, the mesh contours produce knotty, incorrect topologies (Figure \ref{fig:problem}(b)). Using heuristics to untangle them \cite{Northrup:2000,Eisemann:2008} produces unreliable results.

\begin{figure}
    \centering
    (a)
\includegraphics[width=1.5in]{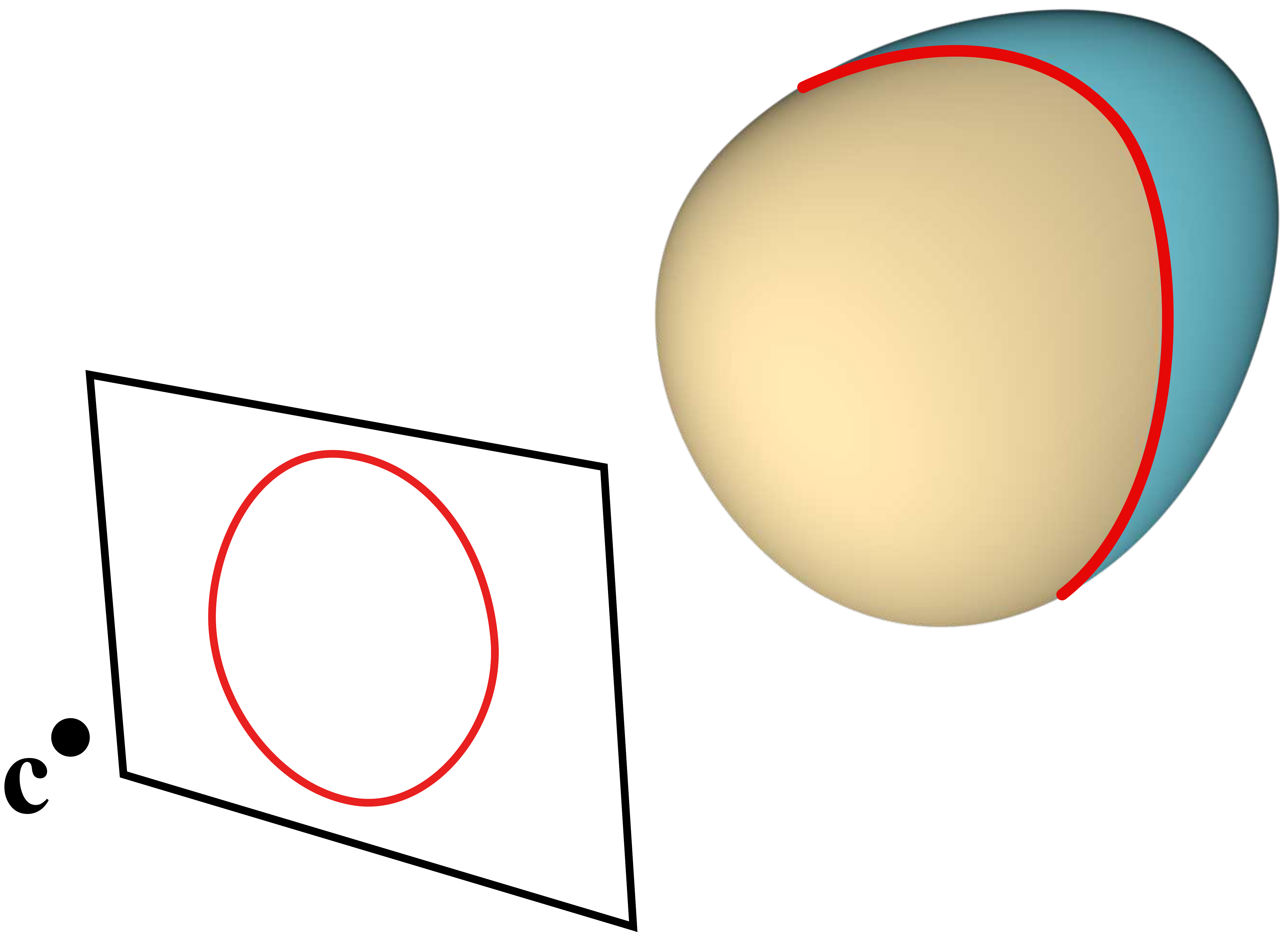}%
~%
(b)
\includegraphics[width=1.5in]{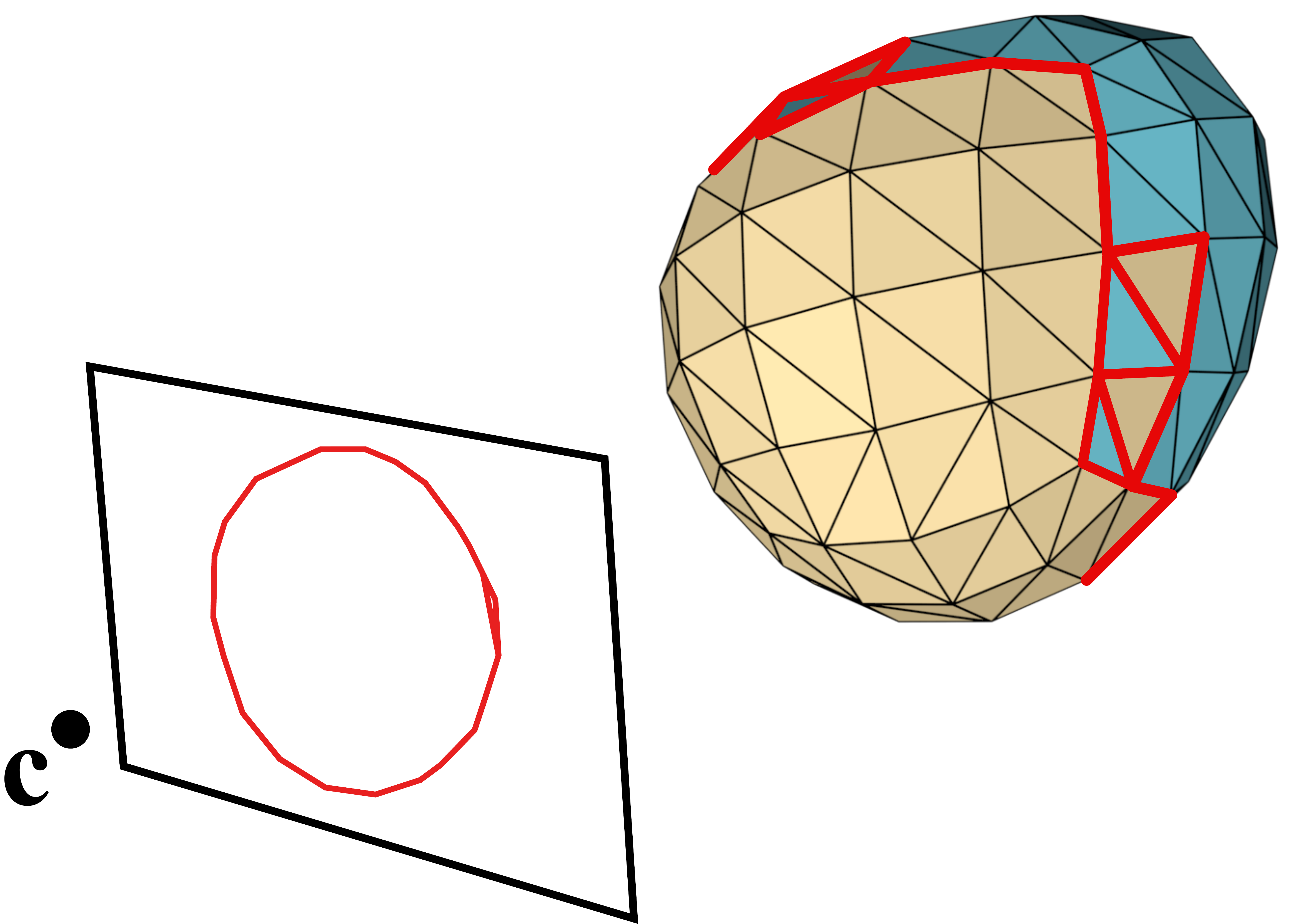} 
    \caption{Convex contour tessellation.  
    (a) The input is a convex smooth surface viewed from a camera position $\bc$. The contour projects to a simple convex curve.
     (b)
    Naively converting the smooth surface to a triangle mesh produces a mesh for which the occluding contour projects to a curve with very different topology. The new topology can be  badly behaved, leading to many visibility errors  \cite{Benard:2014}. 
     }
    \label{fig:problem}
\end{figure}

Many approaches instead sample polyline approximations to the smooth occluding contour generator, which is guaranteed to have simpler topology. These methods then perform visibility tests against a triangle mesh \cite{Appel:1967,Grabli:2010, Markosian:1997,Hertzmann:2000,Karsch:2011,Winkenbach:1996}.  However, these ray tests are not reliable because a triangle mesh has different visibility from the underlying smooth surface  \cite{Eisemann:2008,Benard:2014,Grabli:2010} (Figure \ref{fig:trivial}).
Such small errors in visibility tests can propagate to produce topologically-invalid drawings, such as objects with large gaps in their silhouettes.
Previous methods use voting schemes and other heuristics to fix visibility, but none are exact.
Computing curve visibility with image buffers \cite{Cole:2010,Eisemann:2008,Saito:1990} or ray-tests with smooth geometry \cite{Elber:1990} have similar numerical problems. Our method builds most directly on the method of B\'enard et al.~\shortcite{Benard:2014}, which frequently guarantees correct results, but is very expensive to compute.  
Moreover, as we show in this paper, some polylines simply cannot be assigned valid visibility, which affects all the methods described above. 
\begin{figure}
    \centering
\includegraphics[width=1.5in]{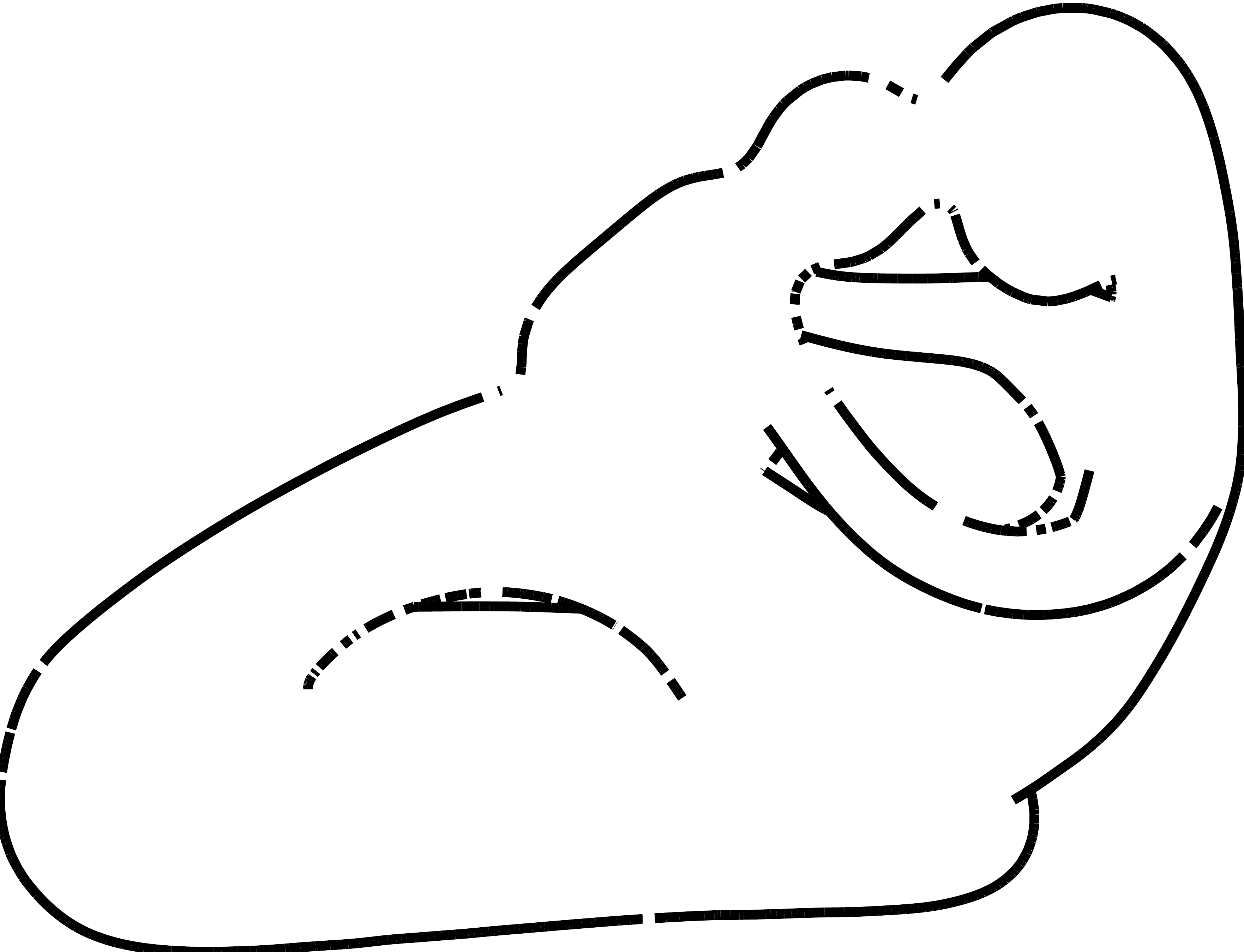}
    \caption{
    Naive visibility. In this example, the occluding contour has been sampled into line segments, and then visibility for each segment was computed by a separate ray test against a dense triangulation of the surface. This naive strategy produces unreliable visibility, such as incorrect gaps in the contours. Most algorithms use a combination of heuristics that can greatly improve over a naive result, but they cannot remove all visibility errors.
(Fertility by UU from AIM@SHAPE-VISIONAIR Shape Repository).    }
    \label{fig:trivial}
\end{figure}

\section{Convex Surface Algorithm}
\label{sec:convex}

In order to provide basic intuitions for our theory and algorithms, we first describe a highly simplified version of the algorithm for convex, closed surfaces. We generalize these ideas to non-convex surfaces in subsequent sections.

The input to the algorithm is a strictly convex, oriented smooth surface $\cS$, viewed from camera position $\bc$.
The occluding contour of this surface  must be a convex curve in the image (Figure \ref{fig:problem}(a)).  

Our goal is to produce a triangle mesh $\cM$ whose  contour generator is a good representation for the true contour generator of $\cS$. Specifically, we want the contour generator of $\cM$ to partition the mesh into two regions, one containing only front-facing triangles, and one with only back-facing triangles. We also want $\cS$ and $\cM$ to be geometrically similar. Once the mesh is computed, the contour visibility can be computed from the mesh using standard exact techniques, such as ray tests.
No existing algorithm provably achieves this goal \bh{7.6}.  For example, directly using the contour generators of a triangle mesh will produce spurious 2D self-intersections \bh{6.1} (Figure \ref{fig:problem}(b)).

The first step is to sample a polyline $\cC$ on the smooth surface, where each vertex on the polygon lies on the contour generator, and the polyline projects to a simple, closed polygon in 2D (Figure \ref{fig:steps}(a)). For example, this can be achieved by root-finding Eq.~\ref{eq:contour} on a mesh representation of the surface, see \cite{Benard:2014}\S 6.2.

Our goal now is to generate a mesh for which $\cC$ is the mesh's contour generator: $\cC$ will partition the mesh into one region containing only front-facing triangles, and another region containing only back-facing triangles, and the front-facing region will be nearer to the camera.
To generate the front-facing region, we proceed as follows (Figure \ref{fig:steps}):   
\begin{enumerate}
    \item Project the polyline $\cC$ to a 2D image plane. %
    \item Tessellate the 2D polygon, using, for example,  Constrained Delaunay Triangulation (CDT) \cite{PaulChew1989}.
    \item Project each new interior vertex back to 3D by casting a ray from the camera through the 2D position, and intersecting it with the front-facing region in $\cS$.
\end{enumerate}
The back-facing region is meshed by the same procedure. Finally, the output mesh $\cM$ is produced by stitching these two regions at $\cC$. 

\begin{figure}
    \centering
        
    \includegraphics{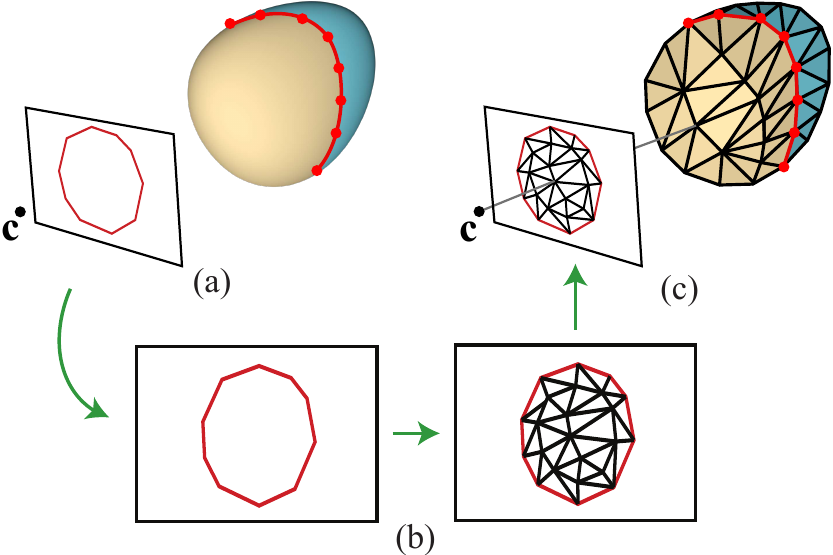} 
        
    \caption{Steps of the convex surface tessellation algorithm. 
    (a) Points are sampled along the smooth surface contour. These points are projected to a 2D polygon in image space.
    (b) The polygon is triangulated in image space.
    (c) 2D triangles are lifted to 3D, by projecting their vertices back onto the original surface, in both the front-facing and back-facing regions. 
    }
    \label{fig:steps}
\end{figure}

The key insight of this algorithm is that we can guarantee front-facing orientations by triangulating in image space and then projecting to 3D.
This is because any valid 2D triangulation will produce only clockwise (or only counter-clockwise) triangles, and, furthermore, projecting a triangle from 2D to 3D preserves orientation (Appendix \ref{app:orient}). Hence, this procedure correctly produces a mesh with a front-facing region and a back-facing region, separated by~$\cC$. Moreover, the output mesh $\cM$ is topologically equivalent to $\cS$, and geometrically similar due to the use of $\cC$ and the ray-casting step. The geometric accuracy can be improved arbitrarily by refining the 3D contour sampling and the 2D triangulation before ray-casting.
Hence, this algorithm solves the contour meshing problem for the convex case.

Once the mesh $\cM$ is computed, it can be rendered and stylized with standard non-photorealistic rendering methods. The main benefit of having the mesh is that it can be used to determine curve visibility with ray tests, for example, in a scene comprising multiple convex objects. Note that computing the back-facing region is unnecessary for computing contour visibility, provided the scene is set up appropriately. We include back-facing regions in this paper solely for completeness and for visualization.

\section{General Contour Regions and Polygons}
\label{sec:contour_regions}

Suppose we sample the contour generators of a  closed smooth surface into polylines $\cC$, where the surface is not necessarily convex. When can those contours be triangulated into a new mesh?  That is, does there exist a new mesh where $\cC$ are its contour generators, and the mesh has the same topology as the smooth surface?
This section describes the types of contours that may occur on smooth surfaces, which allows us to identify which kinds of polygons can and cannot be triangulated. 

These questions are important because the existence of a triangulation implies that there exists a valid visibility labeling for the curves. Conversely, if no triangulation exists, then a plausible visibility labeling may likewise be impossible.

We consider an oriented smooth surface $\cS$, viewed from camera position $\bc$, for which all back-facing points are invisible due to occlusion.  We assume general position (a.k.a.~generic position)  \bh{3.5}. For the discussion in this section, we assume that surfaces are closed and do not contain self-intersections. 
As a consequence, the occluding contour generator of $\cS$ is a set of closed loops that partition the surface into front-facing and back-facing \textit{regions}. 

Hence, for a given polyline sampling of the contour generator of $\cS$, the question of whether or not a valid triangulation exists is equivalent to the question of whether or not all of the regions enclosed by those contours can be meshed with the appropriate orientations, i.e., all front-facing or all back-facing. As discussed in the previous section, this reduces to the question of whether each region can be triangulated in image-space.

\subsection{Types of Regions}

This section categorizes the different types of 3D surface regions, in terms of the types of curves their boundaries project to. Each region must be entirely front-facing or entirely back-facing.  The categorization applies equally to smooth surface regions and polygonal regions. This categorization is nested: each category is more general than the previous ones, and a triangulation algorithm that works for the final category applies to all of these cases. This categorization applies regardless of whether or not the regions are bounded by contours.

\begin{figure*}
    \centering
\begin{tabular}{cccc}
\multicolumn{4}{c}{\Large (a) \textbf{Simple}} \\
    \includegraphics[height=1in]{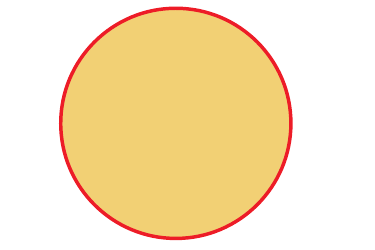} &
    \includegraphics[height=1in]{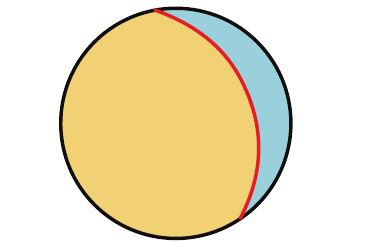} &
    \includegraphics[height=1in]{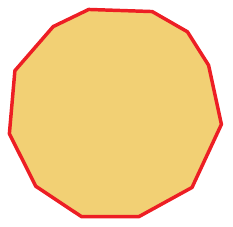} &
    \includegraphics[height=1in]{polygon_types/convex_polygon.pdf} \\
    Camera view &
    Side view &
    2D polygon &
    2D embedding of polygon \\
    of 3D shape &
    &
    from points sampled on contour &
    for visualization
\\
        \hline\\
\multicolumn{4}{c}{(b) \Large \textbf{Self-Overlapping}} \\
    \includegraphics[height=1in]{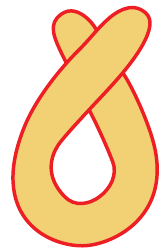} &
    \includegraphics[height=1in]{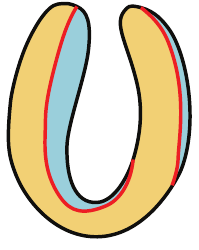} &
    \includegraphics[height=1in]{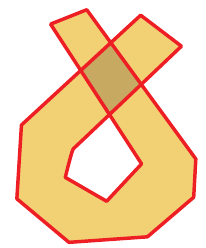} &
    \includegraphics[height=1in]{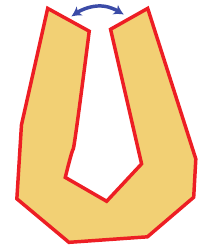} \\
    Camera view &
    Side view &
    2D polygon &
    2D embedding of polygon
\\
        \hline\\
    \multicolumn{4}{c}{\Large (c) \textbf{Weakly Self-Overlapping (WSO)}} \\
\includegraphics[height=1in]{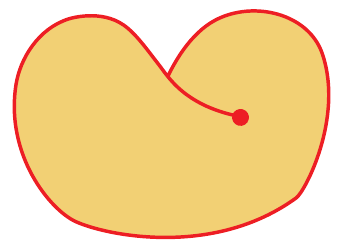} &
    \includegraphics[height=1in]{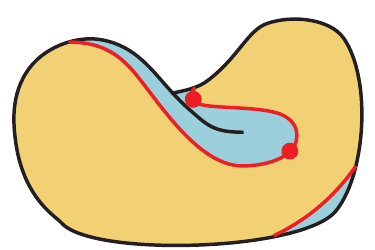}  &
    \includegraphics[height=1in]{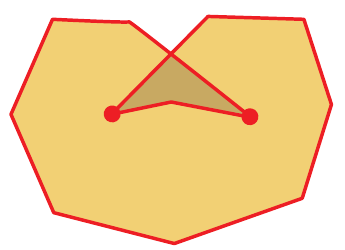} & 
    \includegraphics[height=1in]{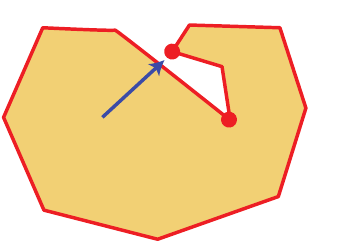} \\ 
    Camera view &
    Side view &
    2D polygon &
    2D embedding of polygon
\\
            \hline\\
    \multicolumn{4}{c}
{\Large (d) \textbf{Weakly Self-Overlapping with Holes (WSOH)}} \\
    \includegraphics[height=1in]{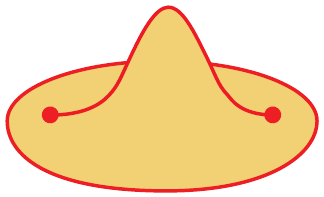} &
    \includegraphics[height=1in]{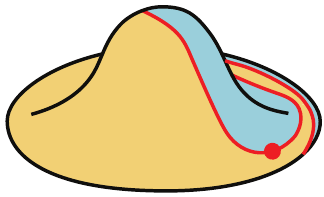} &
    \includegraphics[height=1in]{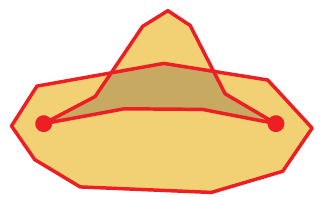} &
    \includegraphics[height=1in]{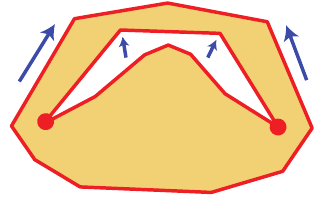} \\
    Camera view &
    Side view &
    2D region &
    2D embedding of region
    \end{tabular}
    \caption{The four categories of contour regions in 3D, and how they project to 2D. Each row shows a smooth surface from the camera view and from a side view (first and second columns) with the contour in \textbf{\textcolor{red}{red}}, the front-facing region in \textcolor{Orange}{\bf yellow}, back-facing regions in \textcolor{Blue}{\bf blue}, and cusps/singularities as red dots.
    A 2D projection of the front-facing region is shown (third column), with vertices sampled from the contour. The final column is meant to aid in understanding the 2D region; vertices are translated to unfold the region.
    \label{fig:region_types}}
\end{figure*}

\paragraph{Simple curves (Figure \ref{fig:region_types}(a)).}
The easiest case is where a region is bounded by a simple 2D curve, i.e., a curve that does not intersect itself in image space.  A valid sampling of a simple curve produces a simple polygon in 2D, which can be triangulated in 2D with methods such as CDT (Figure \ref{fig:steps}).

\paragraph{Self-overlapping polygons (Figure \ref{fig:region_types}(b)).}
When one part of a region overlaps a separate part in image space, the boundary curve is called \textit{self-overlapping}. Self-overlapping polygons can be triangulated using the algorithm of Shor and Van Wyk \shortcite{ShorVanWyk}\S 4.

While self-overlapping polygons can have multiple incompatible triangulations \cite{ShorVanWyk}\S 3, we have not observed incompatible triangulations in practice. If needed, these incompatibilities could be resolved by using ``crossing'' constraints \cite{EppsteinMumford}, i.e., the depth ordering at 2D intersections implied by the 3D locations of the curves.

\paragraph{Weakly self-overlapping (WSO) (Figure \ref{fig:region_types}(c)).}
A curtain fold cusp (\bh{4.3}) in the contour generator creates a singularity in the occluding contour where the surface self-overlaps.  
The corresponding polygon also has a singularity. Singularities are marked with red dots in Figure \ref{fig:region_types}.
Weber and Zorin \shortcite{WeberZorin}\S 3, call a singular polygons that overlaps  ``weakly self-overlapping'' (\textbf{WSO}), and provide an algorithm for triangulating WSO polygons with singular vertices tagged in the input. 
Note that this algorithm can also triangulate simple curves and self-overlapping curves, which are all considered to be WSO.  We give formal definitions of WSO  in Section \ref{sec:theorems}.

\paragraph{Weakly-Self-Overlapping with Holes (WSOH) (Figure \ref{fig:region_types}(d)).}
In the most general case, which we call \textbf{WSOH}, a region may have holes, and possibly handles. Polygonal regions with holes may be triangulated by first introducing a cut to remove holes and handles, and then applying the WSO algorithm. We formally define WSOH in Section \ref{sec:theorems}.

Because this is the most general case, we develop an algorithm that works for WSOH regions, and it automatically handles the simpler cases described above.
Simple curves, self-overlapping curves, and weakly-self-overlapping curves are all considered WSOH.

\subsection{Invalid Polygons}
\label{sec:invalid}

Sometimes, sampling a contour curve produces a polygon that cannot be triangulated in 2D.
For such a polygon, the 3D polyline cannot be triangulated with solely front-facing (or back-facing) triangles.
We call such polygons \textit{invalid}. The example in Figure \ref{fig:invalid}(a) 
\begin{figure*}
    \centering
    (a)
            \includegraphics[width=1in]{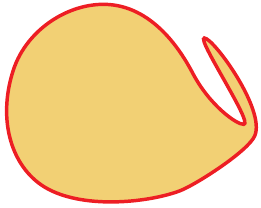}
                \includegraphics[width=1in]{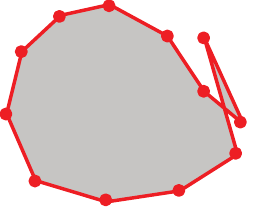}
                \includegraphics[width=1in]{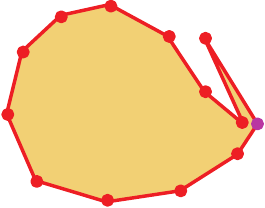}
(b) \includegraphics[width=1.5in]{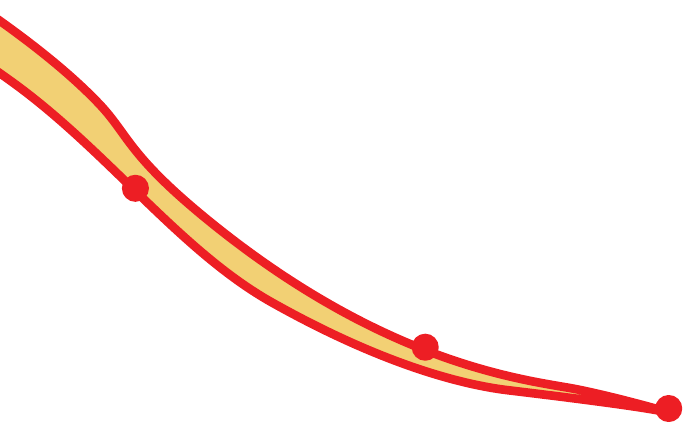} 
\includegraphics[width=1.5in]{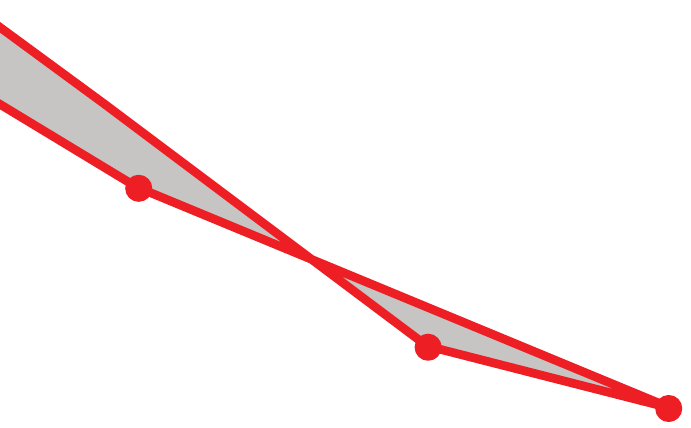}

    \caption{Invalid curves. (a) An example in which sparsely-sampled points around a simple smooth curve produce an invalid 2D polygon (in grey). The polygon self-intersects, and cannot be triangulated without introducing a twist in the resulting 3D surface. In this case, adding a sample point (purple) makes the polygon valid.
    (b) A common case where the polygon is very skinny near a cusp, and undersampling introduces a loop near the cusp.
    }
    \label{fig:invalid}
\end{figure*}
shows a case where undersampling introduces a self-intersection into the curve; we call this structure a \textit{twist}.  
See Shor and Van Wyk \shortcite{ShorVanWyk}\S1 for more discussion and examples of invalid curves.

\newcommand{\Mthree}{$M_{\rm 3D}$\xspace}
\newcommand{\Mtwo}{$M_{\rm 2D}$\xspace}
\newcommand{\hMthree}{$\hat{M}_{\rm 3D}$\xspace}
\newcommand{\hMtwo}{$\hat{M}_{\rm 2D}$\xspace}
\newcommand{\Pthree}{$P_{\rm 3D}$\xspace}
\newcommand{\Ptwo}{$P_{\rm 2D}$\xspace}
\newcommand{\cPthree}{$\mathcal{P}_{\rm 3D}$\xspace}
\newcommand{\cPtwo}{$\mathcal{P}_{\rm 2D}$\xspace}
\newcommand{\Tthree}{$T_{\rm 3D}$\xspace}
\newcommand{\Ttwo}{$T_{\rm 2D}$\xspace}
\newcommand{\Sthree}{$S_{\rm 3D}$\xspace}
\newcommand{\Stwo}{$S_{\rm 2D}$\xspace}

\subsection{Theorems}
\label{sec:theorems}

We now prove theorems that establish the significance of the WSO and WSOH properties for occluding contours. The first two theorems apply to triangle meshes, and the latter two to smooth surfaces.

\subsubsection{Meshes}

We first review the formal definition of WSO for 2D meshes. 

\begin{definition}[\cite{WeberZorin}]
\label{def:wso}
A polygon $P$ is weakly self-overlapping (WSO) if there is a map $f$ from some planar mesh $M$, homeomorphic to a disk, to the plane such that $f(\partial M)=P$, all triangles are mapped with positive orientation, and $\Theta=2\pi$ for each internal vertex in $f(M)$, where $\Theta$ is the sum of triangle angles around a vertex.
$f(M)$ is called a \textit{triangulation} of $P$.
\end{definition}
The mapping is illustrated in Figure \ref{fig:wso_def}, and $\Theta$ in Figure \ref{fig:spiral}(a).
\begin{figure}
    \centering
    (a)\includegraphics[width=3in]{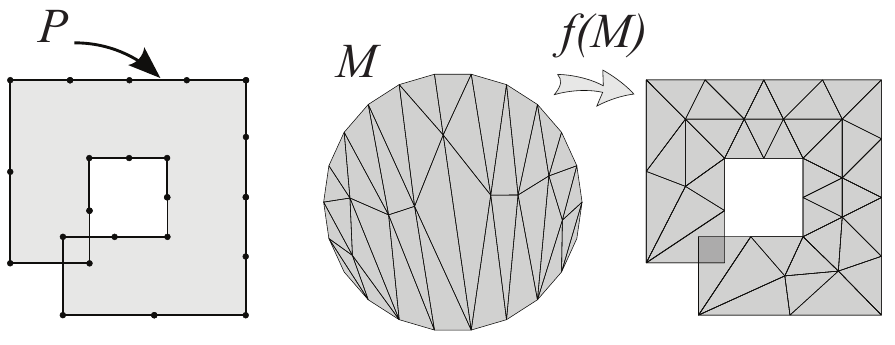}
    (b)\includegraphics[width=3in]{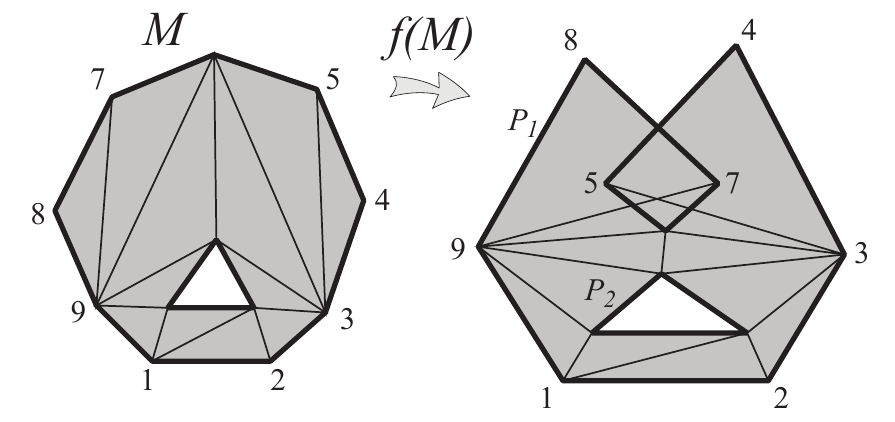}
\caption{(a) Elements of the definition of Weakly Self-Overlapping (WSO), shown here with a self-overlapping polygon $P$. There exists a mesh $M$ with disk topology and a map $f$ such that $P$ is the boundary of $f(M)$, and the mapping $f(M)$ has no flipped triangles or spiral structures.
    (b) Elements of the WSOH definition.  A set of polygons $\cP=\{P_1,P_2\}$ form the boundary of a mapping $f(M)$, where $M$ is a mesh with holes. Some vertices are numbered to show the correspondence.
    (Figure (a) from \cite{WeberZorin}, used courtesy the authors.)
    }
    \label{fig:wso_def}
\end{figure}
\begin{figure}
    \centering
    (a)
    \includegraphics[width=1in]{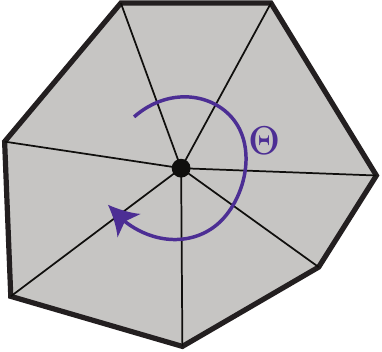}
    (b)
    \includegraphics[width=1in]{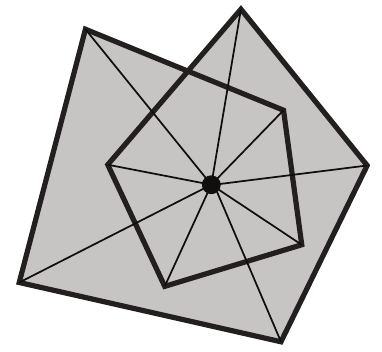}
    \caption{Rotations around a vertex in 2D. Let $\Theta$ be the sum of the angles adjacent to a vertex.
    (a) 
    Normally, at interior vertices of a typical planar triangulation with positive orientation, there are no overlaps or folds, implying $\Theta=2\pi$.
    (b) A vertex where $\Theta=4\pi$. We disallow this case in the WSOH definition.
    \label{fig:spiral} }
\end{figure}

Overlaps ($\Theta > 2\pi$) may occur at boundary vertices; these are singularities. If $\Theta \in [0,2\pi]$ for all boundary vertices as well, then the polygon is self-overlapping.

The condition $\Theta=2\pi$ rules out a peculiar violation of local injectivity. 
In normal situations, the fact that all triangles have positive orientation implies that $\Theta=2\pi$ for each internal vertex (Figure \ref{fig:spiral}(a)). However, Weber and Zorin point out that a spiral structure (Figure \ref{fig:spiral}(b)) produces a non-injective local structure with $\Theta=4\pi$, or, more generally, a positive integer multiple of $2\pi$.  
On a 3D mesh, we call a vertex \textit{spiral} if, for some viewpoint, either the sums of the positive angles or the sums of the negative angles lie outside $[-2\pi, 2\pi]$. This generalized definition will be useful later for ruling out  another hypothetical structure called fusilli cusps \bh{4.7}. We have never observed spiral vertices on any real meshes.

We now generalize the definition of WSO to regions with holes.
\begin{definition} 
Let $\cP$ be a set of $K$ polygons.  %
This set is called weakly self-overlapping with holes (WSOH) if there exists a genus-$(K-1)$ 2D mesh $M$ such that $f(\partial M)=\cP$, all triangles are mapped with positive orientation, and $\Theta=2\pi$ for each internal vertex in $f(M)$. 
$f(M)$ is called a triangulation of $\cP$.
\end{definition}
These quantities are illustrated in Figure \ref{fig:wso_def}.
Note that a single WSO polygon $P$ is a special case of a WSOH set, with $K=1$. The associated mesh has disk topology (genus $0$).

\begin{theorem}
Let \Tthree be a connected triangle mesh in 3D. Let $\bc$ be a camera position with an associated image plane. Assume all vertices in \Tthree have positive depth from the camera and no vertices are spiral.
Let \Ttwo be the projection of the mesh on the image plane.
\Ttwo is the triangulation of a set of WSOH polygons $\cP$ if and only if all triangles in \Tthree are front-facing or all back-facing.  Moreover, \Ttwo is WSO if and only the above conditions hold, and \Tthree is genus 0.
\end{theorem}

\begin{proof}
(if)
Suppose all triangles in \Tthree are front-facing. Then, all triangles in \Ttwo have positive orientation, because projection preserves orientation (Appendix \ref{app:orient}). All interior vertices have $\Theta_i=2\pi$ because of the positive orientation of the adjacent triangles and the non-spiral condition. 
\Ttwo is bounded by $K$ polygons, where $K$ is one more than the genus of \Tthree.

If all triangles in \Tthree are back-facing, then we can produce the WSOH triangulation by reversing the orientations of all triangles.

(only if) Suppose \Ttwo is the triangulation of a WSOH set $\cP$. Then, all triangles must have positive orientation, and, because projection preserves orientation, all faces of \Tthree must be front-facing.
\end{proof}

Hence, suppose we begin with a 3D smooth surface, and sample its contour generators into polylines that bound the front- and back-facing regions of the mesh. It is possible to triangulate the surface with consistent orientations if and only if each region's boundary is WSOH.

WSO is a special case of this theorem: the projection of a front-facing mesh region without holes corresponds to a WSO polygon, and vice versa.

We make one additional conjecture, for which we do not have a proof or counterexample: for a WSOH set of polygons, each of the component polygons is WSO.

\subsubsection{Smooth Surfaces}
\label{app:smooth_theorems}

We now prove the analogous result for smooth surfaces.
We first define WSO and WSOH for smooth curves, adapting ideas from \cite{WeberZorin} and \cite{ShorVanWyk}

\begin{definition}
\label{def:wso_sm}
A closed curve $g: S^1\rightarrow \mathbb{R}^2$ is weakly self-overlapping (WSO) if there is a map $f : D^2\rightarrow \mathbb{R}^2$ from the disk $D^2$ such that $f(\partial D^2)=g$, $f$ maps with positive orientation and is locally injective everywhere in the interior or $D^2$. 
\end{definition}

\begin{definition}
\label{def:wsoh_sm}
A set of $K$ closed curves $g_i: S^1\rightarrow R^2$ is weakly self-overlapping with holes (WSOH) if
there is a genus-$(K-1)$ region $R \subset \mathbb{R}^2$ such that $f(\partial R)=\bigcup g_i$, $f$ maps with positive orientation and is locally injective everywhere in the interior of $R$. 
\end{definition}

As noted by Weber and Zorin, WSOH requires that the map $f$ has Jacobian with positive determinant everywhere; unlike in meshes, we do not need to separately enforce positive orientation and injectivity (e.g., $\Theta=2\pi$).

\newcommand{\sign}{\mathrm{sign}}
\newcommand{\dfdu}{\mybf_u}
\newcommand{\dfdv}{\mybf_v}
\newcommand{\dpdu}{\bp_u}
\newcommand{\dpdv}{\bp_v}
\newcommand{\dzdu}{z_u}
\newcommand{\dzdv}{z_v}

\begin{theorem}
\label{thm:wso_smooth}
Let \Sthree be a connected smooth surface with genus $K-1$.
Let $\bc$ be a camera position with an associated image plane. Assume all points in \Sthree have positive depth. Let \Stwo be the projection of \Sthree on the image plane. \Stwo is a WSOH region if and only if all points in the interior of \Sthree are front-facing or all back-facing.
\end{theorem}

\begin{proof}
Let $(u,v) \in R$ be a parameterization of the interior of \Stwo. In world coordinates, the image plane locations are $\mybf(u,v) = (x,y,1)$, and let $z(u,v)>0$ be the $z$-coordinate for each point, so the corresponding surface \Sthree is $\bp(u,v)=\mybf(u,v)\ z(u,v)=(xz,yz,z)$, with camera position $\bc=(0,0,0)$ (Figure \ref{fig:smooth_tangents}). 
Let $\dfdu \equiv \partial \mybf/\partial u, \dfdv \equiv \partial \mybf/ \partial v$.
The orientation of \Stwo at a point is 
\begin{equation}
\sign ( (\dfdu \times \dfdv) \cdot (\bc-\mybf))=
\sign \det (\dfdu, \dfdv, -\mybf)
\label{eq:two_orient}
\end{equation}
because of the equivalence of scalar vector product to a determinant. (This formula is equal to the sign of the determinant of the Jacobian of the 2D version of the mapping $f(u,v)=(x,y)$.)
\begin{figure}
    \centering
    \includegraphics[width=2.5in]{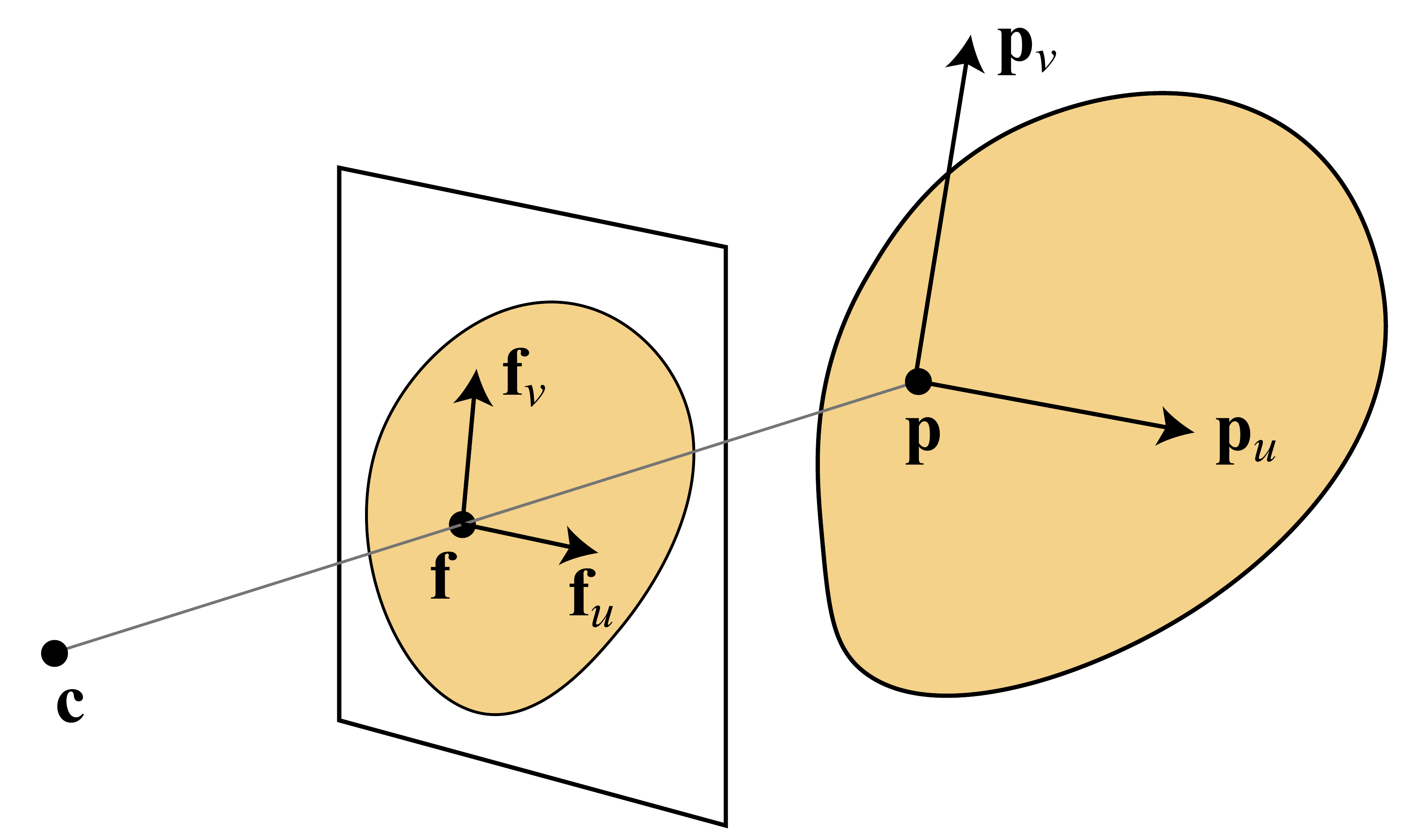}
    \caption{Elements of the proof of Theorem \ref{thm:wso_smooth}. See text for details.}
    \label{fig:smooth_tangents}
\end{figure}

Let $\dpdu \equiv \partial \bp/\partial u, \dpdv \equiv \partial \bp/\partial v$.
The surface normal at a point is the cross-product of the tangent vectors: $\dpdu \times \dpdv$, so the orientation of a surface point is 
\begin{align}
\sign ( (\dpdu \times \dpdv) \cdot (\bc-\bp)) 
&= \sign \det (\dpdu, \dpdv, -\bp) \\
&= \sign \det (\dfdu z + \mybf \dzdu, \dfdv z + \mybf \dzdv, -\mybf z) \\
&= \sign\ z^3 \det (\dfdu, \dfdv, -\mybf) \label{eq:det} \\
&= \sign \det (\dfdu, \dfdv, -\mybf)
\label{eq:three_orient}
\end{align}
Equation \ref{eq:det} follows from properties of the determinant.

Since Equations \ref{eq:two_orient} and \ref{eq:three_orient} are the same, the surface is front-facing everywhere if and only if the parameterization $f$ has positive orientation everywhere.  Hence, a curve being WSO implies that a front-facing surface exists that projects to this curve, and vice versa.
\end{proof}

Hence, all valid regions on smooth surfaces will be WSOH, and we must produce WSOH sets of polygons from these curves in order to be able to triangulate them.

\subsection{Consequences for Contour Visibility Algorithms}
\label{sec:consequences}

These observations give new insight into why the contour visibility problem has proven so troublesome. 

Existing algorithms that compute smooth occluding contours can be grouped into two categories. Planar map algorithms \cite{Eisemann:2008,Karsch:2011,Winkenbach:1996} modify curves until they create a valid planar map; they guarantee consistent visibility but do not make any guarantees about accuracy of the contours. All other existing methods computed sampled representations of the occluding contour, and then compute visibility for these polylines \cite{Hertzmann:2000,Weiss:1966:VPI:321328.321330,Benard:2014}, including methods that use numerical ray tests and curve sampling \cite{Elber:1990}.
For this latter category, we find that naively sampling contours often produces invalid polygons, e.g., Figure \ref{fig:bhk2014}. Further experiments are shown in Section \ref{sec:experiments}. 
\newcommand{\bhkwidth}{1.5in}%
\begin{figure*}
    \centering
    \begin{tabular}{ccc}
    \includegraphics[width=\bhkwidth]{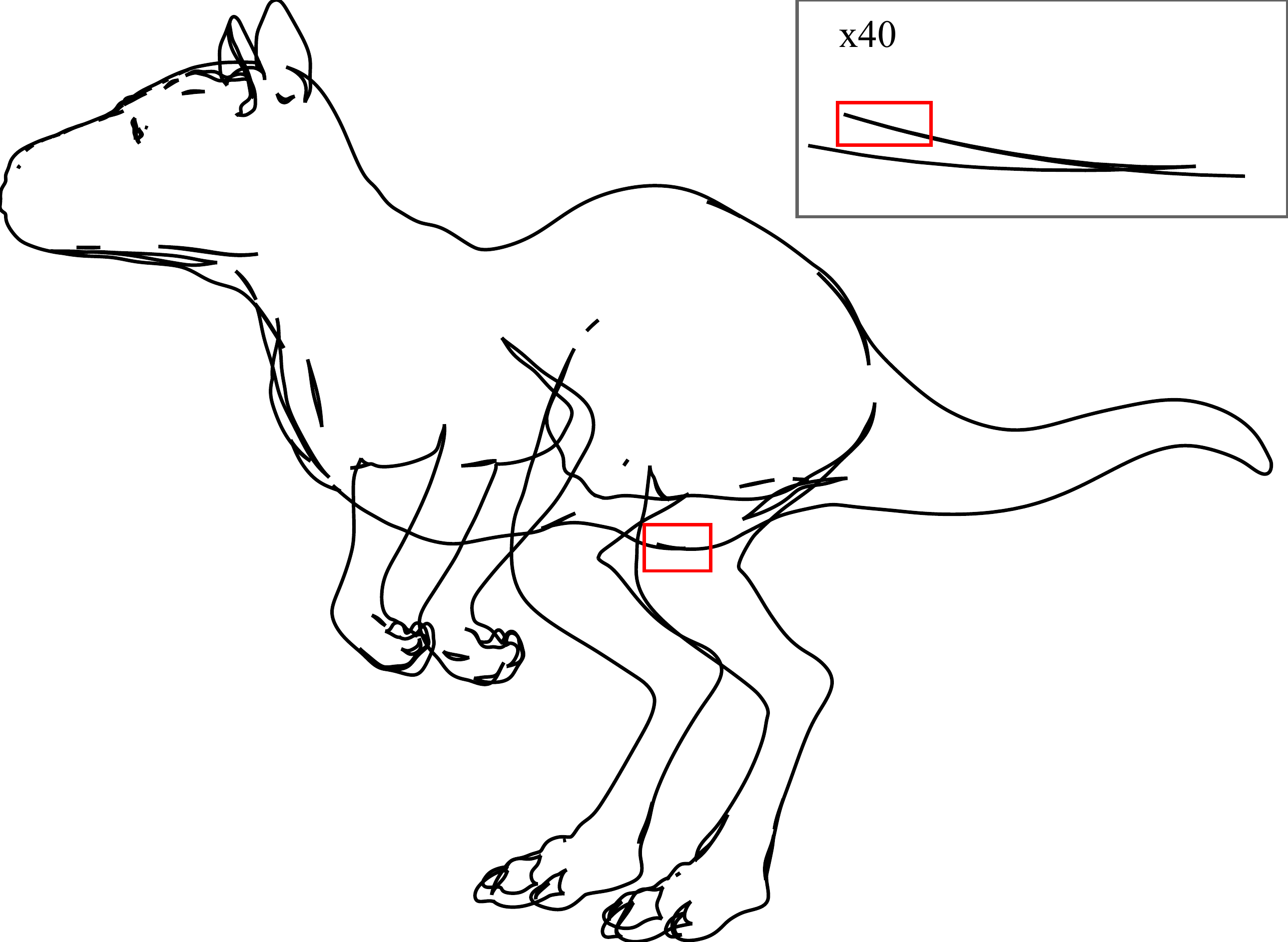}~~
&        \includegraphics[width=2in]{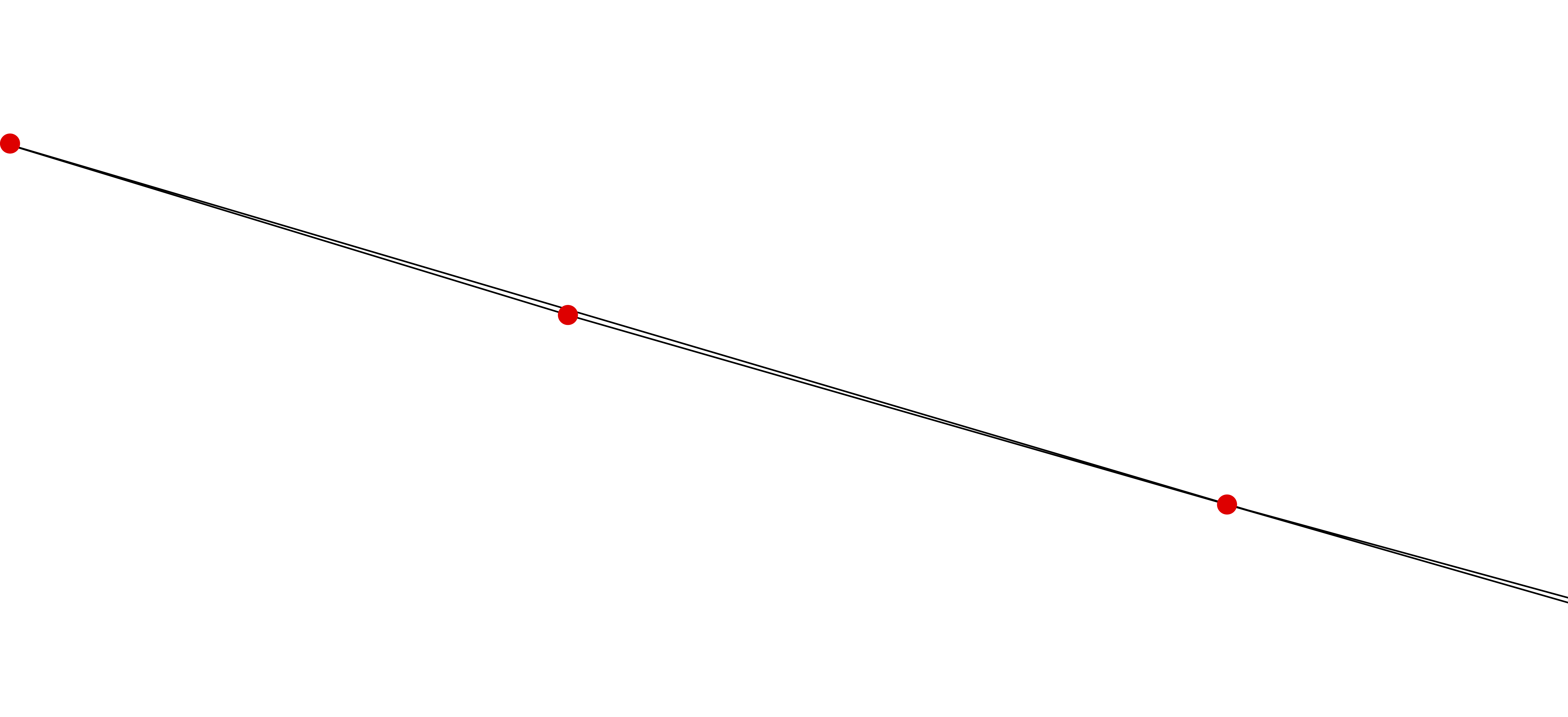}
&        \includegraphics[width=\bhkwidth]{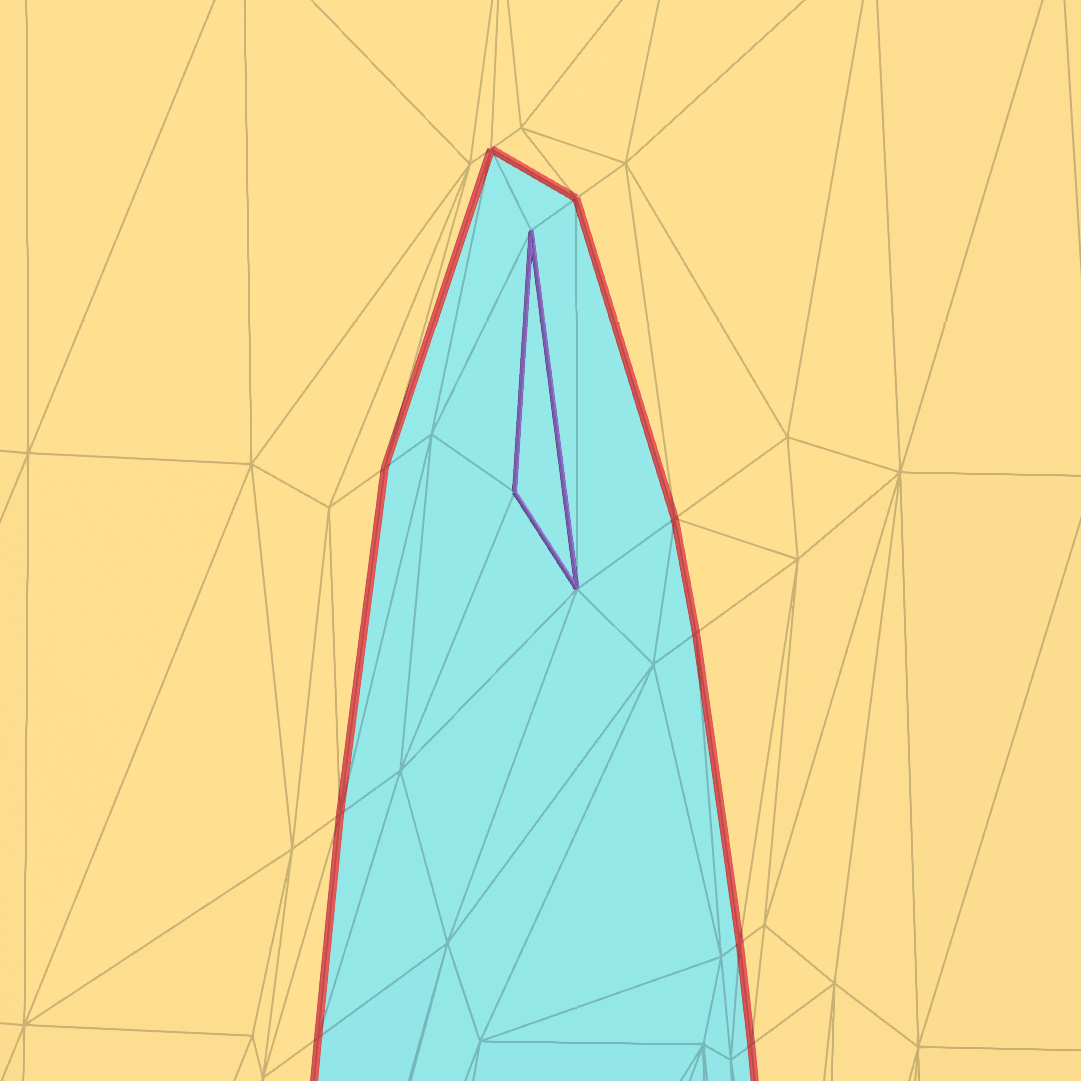} \\
(a) $1\times$/$40\times$ & (b) $1600\times$ & (c) Mesh side view
\end{tabular}
        \caption{Invalid contours with \cite{Benard:2014}.
        (a)
        B\'{e}nard et al.~produce highly refined contours using root-finding. Nonetheless, the contours may fail to be WSO. In the example here, a twist occurs in the sampling near a cusp, similar to the example in Figure \ref{fig:invalid}(b). 
        (b)
        In image-space, this twist is an \textit{extremely} thin structure, nearly invisible even at the $1600\times$ zoom shown here.
        (c) As a result, B\'{e}nard et al.'s triangulation produces an inconsistent triangle in this region, shown in a side view highlighted in purple.  This instance is occluded and thus does not affect visibility, but there is no guarantee that this will always be the case.
        For this example, the output includes 11 inconsistent triangles out of 546,624 that were generated. 
        Our method produces 100\% consistent triangles, generating only 55,476 triangles for this view (Figure \ref{fig:results}), and performing substantially faster (2 minutes for our method, and 10 minutes for B\'{e}nard et al.)
        (Killeroo
\copyright\ headus.com.au)  
    \label{fig:bhk2014}}
\end{figure*}
And no possible algorithm can produce a consistent visibility assignment for invalid contour polygons, because these curves cannot be triangulated. Thus, all existing methods will fail in some cases.

These algorithms often do find correct results, for example, sampled polygons often do happen to be WSOH, and problematic areas are often completely occluded, so that the resulting drawing is valid. But errors inevitably occur as well.

\section{ConTesse Algorithm}

We now describe the \textsf{ConTesse} meshing algorithm in full. This procedure follows the same high-level sequence of steps as in Section \ref{sec:convex}, but these steps are made more involved by non-convexity.  Moreover, we must take steps to ensure that sampled contours are WSOH.
The steps of our algorithm are illustrated in Figure \ref{fig:contesse-overview}.
\begin{figure*}
\centering
\includegraphics[width=7in]{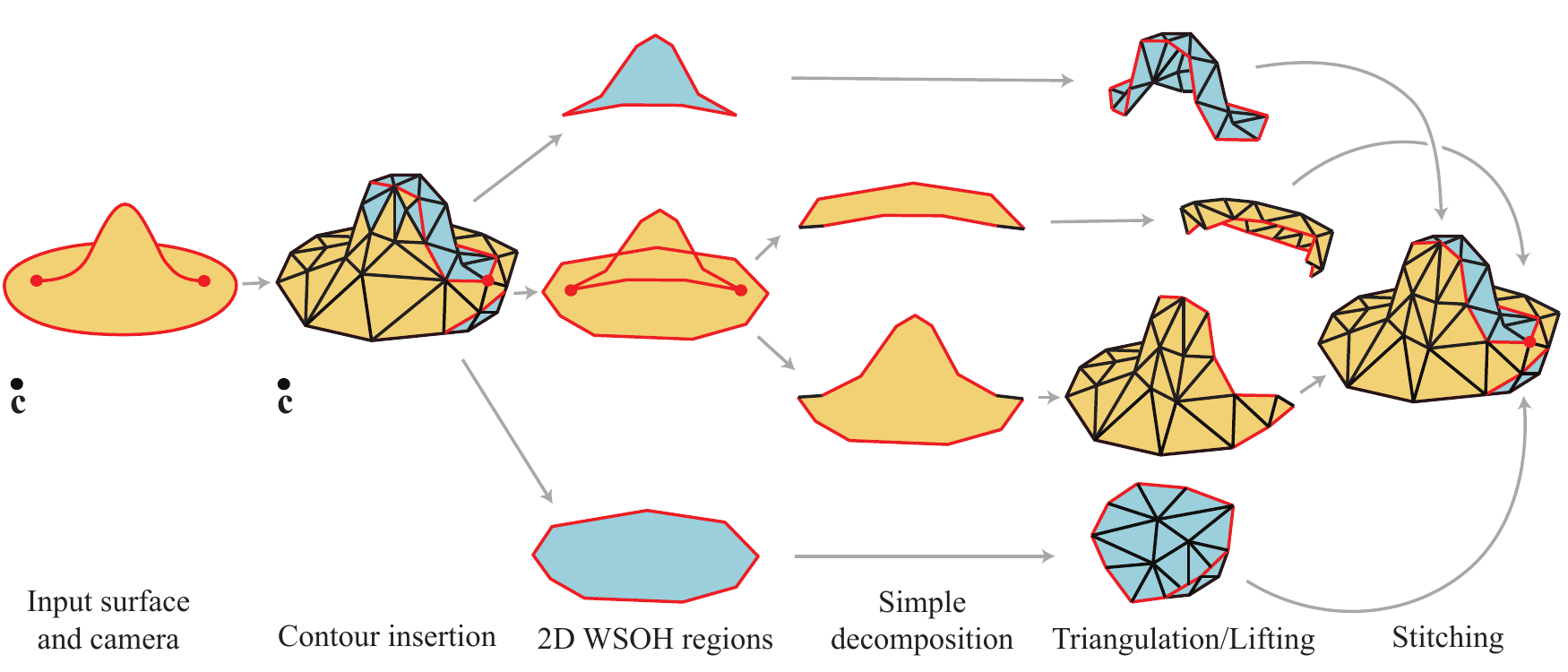}
\caption{Overview of the ConTesse algorithm. The input is a smooth surface viewed from a camera center \textbf{c}. An initial triangle mesh is computed, with contour edges inserted by root-finding. The contours partition the object into WSOH regions. Each region is decomposed into simple polygons, using cuts to remove holes and WSO triangulation to decompose self-overlapping regions. Each simple polygon is then triangulated and lifted to 3D, and these triangulations are stitched to produce the output mesh. This mesh can then be used as input to conventional non-photorealistic rendering algorithms.
\label{fig:contesse-overview}}
\end{figure*}

Related algorithms for reconstructing shape from a network of contours have been developed in computer vision, e.g., \cite{Roberts:1963}, and sketch-based modeling, e.g., \cite{smoothsketch}. Our case is distinct in that we begin with a 3D input surface rather than 2D measurements.  Related problems occur in untangling self-intersecting volumes as well, e.g., \cite{Li:2018:IOS,Sacht:SelfIntersectingVolumes:2013}. 

\subsection{Input and Problem Statement}
\label{sec:problem_statement}

The algorithm takes a mesh as input and camera position $\bc$. We treat the mesh as the base mesh of a subdivision surface, and we require that the surface be orientable, in general position with the camera, with back-faces never visible \bh{3.3,3.5}.  
The surface may have boundaries. We assume the entire scene has positive depth from the camera, i.e., no part of the scene is behind the camera.   We assume no spiral vertices ---including no fusilli cusps \bh{4.7}--- hypothetical structures that we have never observed with real meshes.  We discuss self-intersections in Section \ref{sec:discussion}, which could be handled as a post-process.

The subdivision surface is parameterized by a base mesh $\cP$, so that for any given base mesh point $\bu \in\cP$ there is a corresponding point $\bp(\bu)\in \cS$ on the surface.  Each point is either front-facing $g(\bu) > 0$, back-facing $g(\bu)<0$, or contour generator $g(\bu)=0$, denoted respectively $\sF, \sB$, or $\sC$ for short. 

We aim to produce a new mesh $\cM$ with the following properties:
\begin{enumerate}
    \item The mesh has the same topology as the smooth surface. There exists a smooth bijection that defines the correspondence between points on the surfaces. Mesh vertices have the same 3D locations as their corresponding smooth surface points.
    \item Let $\cC$ be the mesh's contour generator, which partitions the mesh into regions; each region comprises entirely front-facing triangles or entirely back-facing triangles.  
    \item Mesh vertices must correspond to the following types of smooth surface points:
    mesh vertices in $\cC$ correspond to $\sC$ points on the smooth surface; vertices inside front-facing regions correspond to $\sF$ points on the smooth surface; back-facing vertices correspond to $\sB$ points. Every triangle must have at least one non-$\sC$ vertex.
\end{enumerate}
These conditions are equivalent to ``Contour-Consistency'' in~\cite{Benard:2014}\S 4.

\subsection{Contour insertion}
\label{sec:insertion}

The first stage of our algorithm is to create an initial surface triangulation that includes a sampling of the contour generator, that is, new contour vertices $\sC$, at locations corresponding to contour points of $\cS$. In the output, there are no edges containing sign-crossings of $g(\bu)$, and no $\sC\sC\sC$ triangles. We use a version of the method in \cite{Benard:2014}\S 6.1--6.2, with simplified handling of cusps, as follows.

\paragraph{Vertex insertion.}
The first step produces an initial mesh $\cM$ by uniformly subdividing the base mesh $\cP$ a predetermined number of levels (\S 6.1). 
Root-finding on $g(\bu)$ is applied to every $\sF\sF$ and $\sB\sB$ edge, and edges are split whenever roots are found. Specifically, if $\bu_0$ and $\bu_1$ are the parameter locations (preimages) of two adjacent vertices, root-finding densely samples $g(\bu(t)) = g((1-t)\bu_0 + t \bu_1)$ along each edge.
Root-finding and splitting is repeated on any new $\sF\sF$ and $\sB\sB$ edges (up to a maximum of five recursions).  Next, contour insertion is performed on all $\sF\sB$ edges, as in \S 6.2, but with no special handling for cusps. 
Finally, we perform root-finding to find cusps, by repeatedly bisecting any triangles with sign-crossings in both $g$ and radial curvature (\S 6.2). 
If a cusp is detected in the interior of a triangle, a new vertex is inserted at the cusp, and the triangle is split into three triangles. However, if a cusp is detected close to an existing vertex (either in image-space, world-space, or $uv$-space), then the existing vertex is shifted to the cusp location. This shifted vertex produces a $\sC\sC\sC$ triangle, which is resolved by an edge flip.

\paragraph{Singularity labeling.}
Next, the algorithm tags singularities in the image-space contours, where the polygon must locally overlap in image space.
Singularities should correspond to cusps in the contour generator, and so all cusps found by root-finding in the previous step are tagged as singularities.

A vertex can be only singular for one of the two regions it is adjacent to.
For the singular region, the triangles in the one-ring self-overlap in image space (Figure \ref{fig:singularity-direction}(a)). In the tangent plane of a smooth cusp, the self-overlapping side is the convex side of the contour (Figure \ref{fig:singularity-direction}(b)). For the discrete curves we have, we compute  the discrete Laplacian ($\bv + \bx -2\bw$) of the contour loop in 3D, and then determine which of the two regions the Laplacian vector points to by projecting it onto the one-ring (Figure \ref{fig:singularity-direction}(c)). The other region --- that it does not point to --- gets a singularity label at this vertex.

In some cases, the contour polygons have additional singularities missed during the previous step, e.g., see Figure 19 of \cite{Benard:2014}. We detect these as follows. To test a contour vertex $\bw$, the algorithm bidirectionally traces along the contour generator to find two nearby contour vertices $\bv$ and $\bx$, such that each of them is at least $10^{-8}$ from $\bw$ in 3D.
If the image space angle  $\angle \bv\bw\bx$ is less than $\pi/3$, then $\bw$ is marked as a singular vertex.  
\newcommand{\singwdth}{0.9in}

\begin{figure}
    \centering
    (a)
    \includegraphics[width=\singwdth]{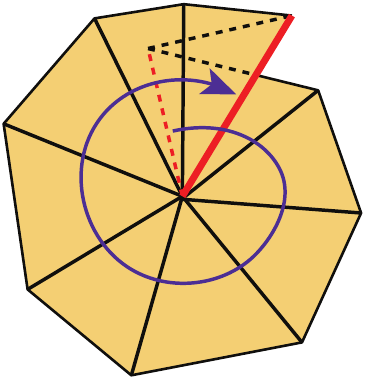}
\hfill
(b)
    \includegraphics[width=\singwdth]{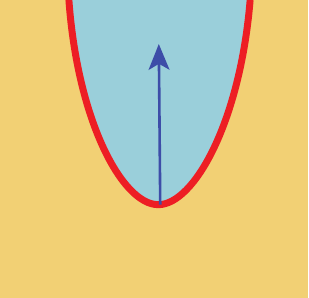}
    \hfill
    (c)
    \includegraphics[width=\singwdth]{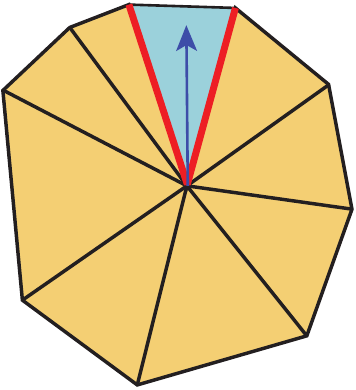}
    \caption{Singularity side determination. (a) At a cusp, one of the two adjacent regions self-overlaps in image-space, i.e., the total rotation angle of the polygons $\Theta>2\pi$.
    (b) In the tangent plane of a cusp on a smooth surface, the curve Laplacian points toward the smaller, non-overlapping region, which is back-facing (blue) in this example.
    (c) Since we are operating with a discrete mesh output, we compute the discrete curve Laplacian, and tag a singularity in the region (yellow) that it does not point to.
    \label{fig:singularity-direction}}
\end{figure}

\paragraph{WSO checking and twist removal.}
While the smooth surface's contour generator must be WSOH in 2D, sometimes a polygon sampled from the contour is not.  This means that the polygon is invalid in 2D and cannot be triangulated consistently (Section \ref{sec:invalid}).  We address these cases with simple heuristics, and, when necessary, additional subdivision levels to increase the contour sampling. Whether a curve is WSO is checked by the algorithm of Weber and Zorin \shortcite{WeberZorin}. All contours should converge to WSO with sufficient sampling density, but we use heuristics to avoid this extra computation when possible. We leave the problem of efficient WSO  sampling as future work. 
If any contours fail to be WSO, then we first employ a series of heuristics to attempt to correct any sampling errors. Because contour insertion produces a reasonably dense sampling of contour generator, invalid portions are typically localized to a few structures that we call \textit{twists}. The heuristics we use to detect and resolve twists are given in Appendix \ref{app:correcting}. 

If any curves are still not WSO after applying the heuristics, then our algorithm subdivides the original mesh to a finer level than before, and repeats all of the steps in this section. This process repeats until all curves are WSO. %
This new sampling is then passed to the next step, below.

\subsection{Region Decomposition}

At the end of the previous stage, we have partitioned the shape into a set of regions that project to WSOH regions in 2D.  In this stage we decompose these regions into simple polygons in 2D, by removing holes and then applying an existing triangulation algorithm to find self-overlaps.  %

\paragraph{Removing holes.}
After insertion, some regions may have holes. 
If the region has holes, we introduce a cut, which is a set of mesh edges added to the region boundary.
Adding cuts to the region boundary produces a new region without holes, where the new boundary of the region traverses each cut twice.

To find a cut, we run the cut-to-disk algorithm of Gu et al.~\shortcite{Gu:2002}\S 3.2 on every surface region. 
The method of Gu et al.~assumes a valid input triangulation, but our input triangulation may pass outside the polygon in image space.  As a result, some possible cuts may pass outside the polygon in image space. In order to avoid bad cuts, we modify the algorithm to avoid cusps, image-space intersections, and triangles facing the wrong direction for their region, where possible. Details are provided in Appendix \ref{app:cuts}.

\paragraph{Initial triangulation.}
After removing holes, the mesh can be partitioned into front-facing regions and back-facing regions. Each region is bounded by a polygon in 3D, comprising the contours and/or boundaries surrounding the region. For each region, the boundary polygon is projected to 2D, and triangulated in 2D using the WSO triangulation algorithm of Weber and Zorin \shortcite{WeberZorin}\S 3.1--3.2, which takes polygons with labeled singularities as input.
Mapping this triangulation to the 3D contour polygon gives an initial valid triangulation in 3D. However, this triangulation only uses contour vertices and so cannot accurately represent surface shape in the interior of the polygon.

\paragraph{Simple decomposition.}
We then decompose the triangulation into simple polygons, with the method of Weber and Zorin \shortcite{WeberZorin}\S 4; see the \emph{Simple Decomposition} step of Figure \ref{fig:contesse-overview}.
Once this step is completed for the whole surface, we have decomposed it into simple 2D regions, each of which is entirely front-facing or entirely back-facing.

\subsection{Triangulation and Lifting}

In this stage, we generate a 3D triangulation for each simple polygon from the previous step.

Our approach is to first identify a set of 3D surface points that lie within the simple polygon, and then triangulate these points. While it may be possible instead to use the ray-casting procedure in Section~\ref{sec:convex}, extra steps would be required to disambiguate rays that intersect the WSO region multiple times.

The procedure for finding these points is as follows.
Initially, each non-contour edge in a 2D polygon connects two contour vertices from the 3D mesh.
We first search for a path on the 3D surface connecting these vertices that projects to the line containing the edge in 2D.
If we find such a path, then we march along it, and periodically produce sample points on the smooth surface.

This process skips samples that do not move in the direction of the endpoint in image space, to avoid folds due to inconsistencies. Specifically, let the endpoints of an edge be $\ba$ and $\bb$; after a sample $\bv_i$ is inserted, the next sample is inserted as $\bv_{i+1}$ only if $(\bv_{i+1}-\bv_i)\cdot(\bb-\ba)>0$. Samples are skipped also if they would be within an image-space distance threshold to an existing vertex. This process is repeated for each edge of the initial triangulation.

This produces a new set of 2D/3D sample points within the polygon. The triangulation is then computed by CDT \cite{PaulChew1989} on these sample points and the bounding polygon in 2D.

In some cases our method fails to find a path between vertices in the original polygon, which produces very long edges in the triangulation; this happens most often when one of the vertices is a cusp. For these edges, we identify a large (five-ring) neighborhood around one of the endpoints, and then apply Delaunay edge-flipping for all edges in the neighborhood, which effectively removes long edges.

This process may produce $\sC\sC\sC$ triangles, i.e., triangles where each vertex lies on a contour edge, which in turn can lead to a degenerate mesh when two adjacent regions have triangles formed of the same three contour vertices. For each $\sC\sC\sC$ triangle, we randomly pick an edge between two vertices that is not a contour edge, and split this edge, and then perturb the new vertex toward  the camera for front-facing patches, or away for back-facing patches.

\subsection{Final output}

The final output mesh is produced by stitching the triangulated regions from the previous step.
The occluding contours of this surface correspond to the occluding contours of the input smooth surface.
This mesh can then be supplied to standard mesh contour detection and stylization algorithms \bh{9}.

Assuming that all regions are WSOH, the output mesh satisfies the goals set out in Section~\ref{sec:problem_statement} by design, and the contour generator of the output mesh is the contour generator sampled from the input surface.  The contour generator's visibility can be computed by applying standard visibility algorithms for mesh contours.

\section{Experiments}
\label{sec:experiments}

We implemented our method using Catmull-Clark subdivision surfaces, with exact limit position and normal evaluation using the algorithm and code from Lacewell and Burley \shortcite{lacewell}. To compute radial curvatures for cusp detection, we use finite differences to estimate surface derivatives, and then compute radial curvature analytically from these estimates.

Our system takes a 3D model and a viewpoint $\bc$ as input.  The system outputs a remeshed version of the input surface.  The system then computes the visible contours of this mesh using standard methods \bh{4}, which are output as a tagged SVG file. These may be further stylized by standard methods; we use topological simplification \bh{9.3} and stroke texturing in our examples \bh{9.2}. 

\newcommand{\figwidth}{1.6in}
\newcommand{\figheight}{1in}
\begin{figure*}
\centering
\begin{tabular}{cccc}
(a) Camera view & (b) Side view & (c) Output curve network & (d) Occluding contours \\
\includegraphics[width=\figwidth]{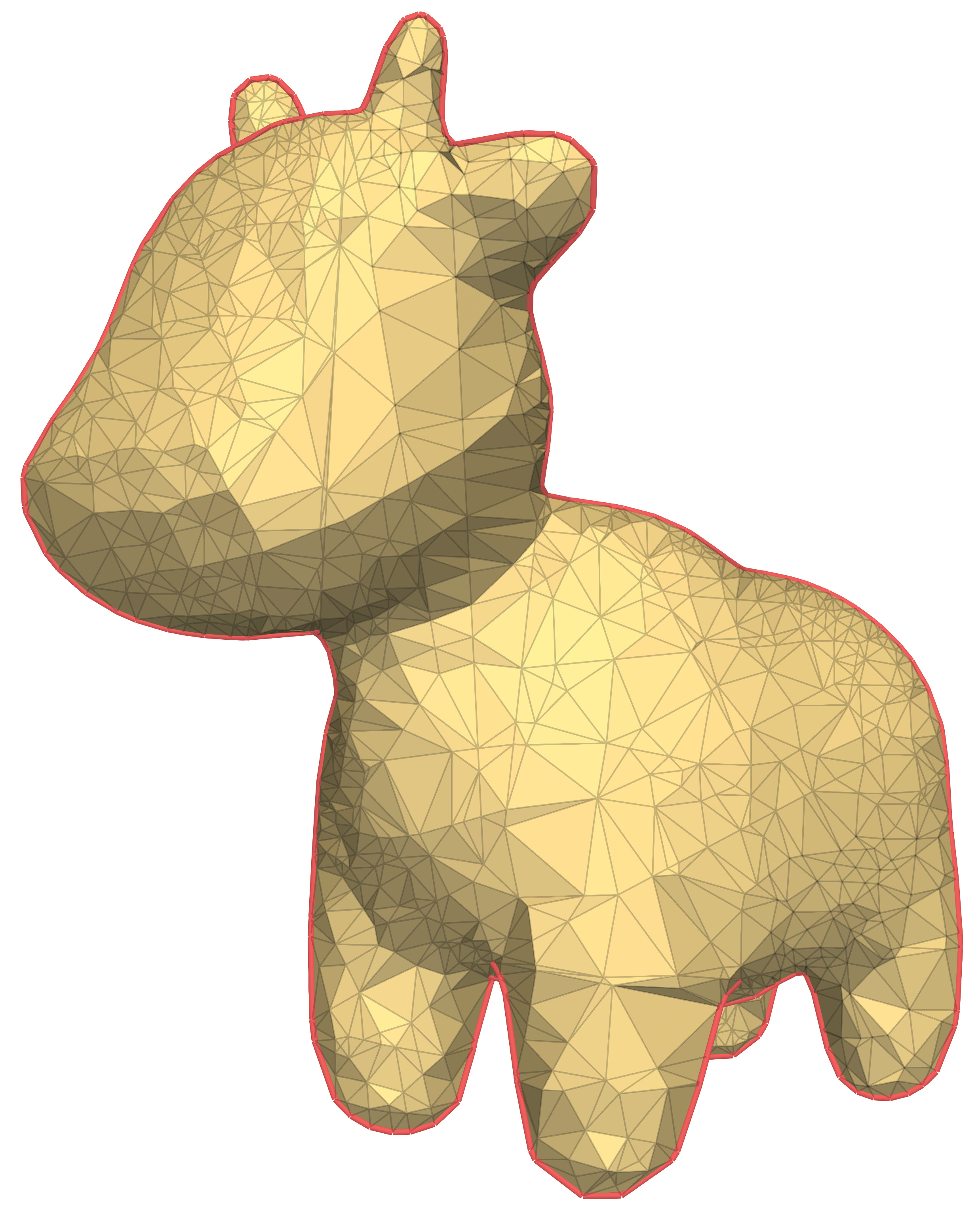} &
\includegraphics[width=\figwidth]{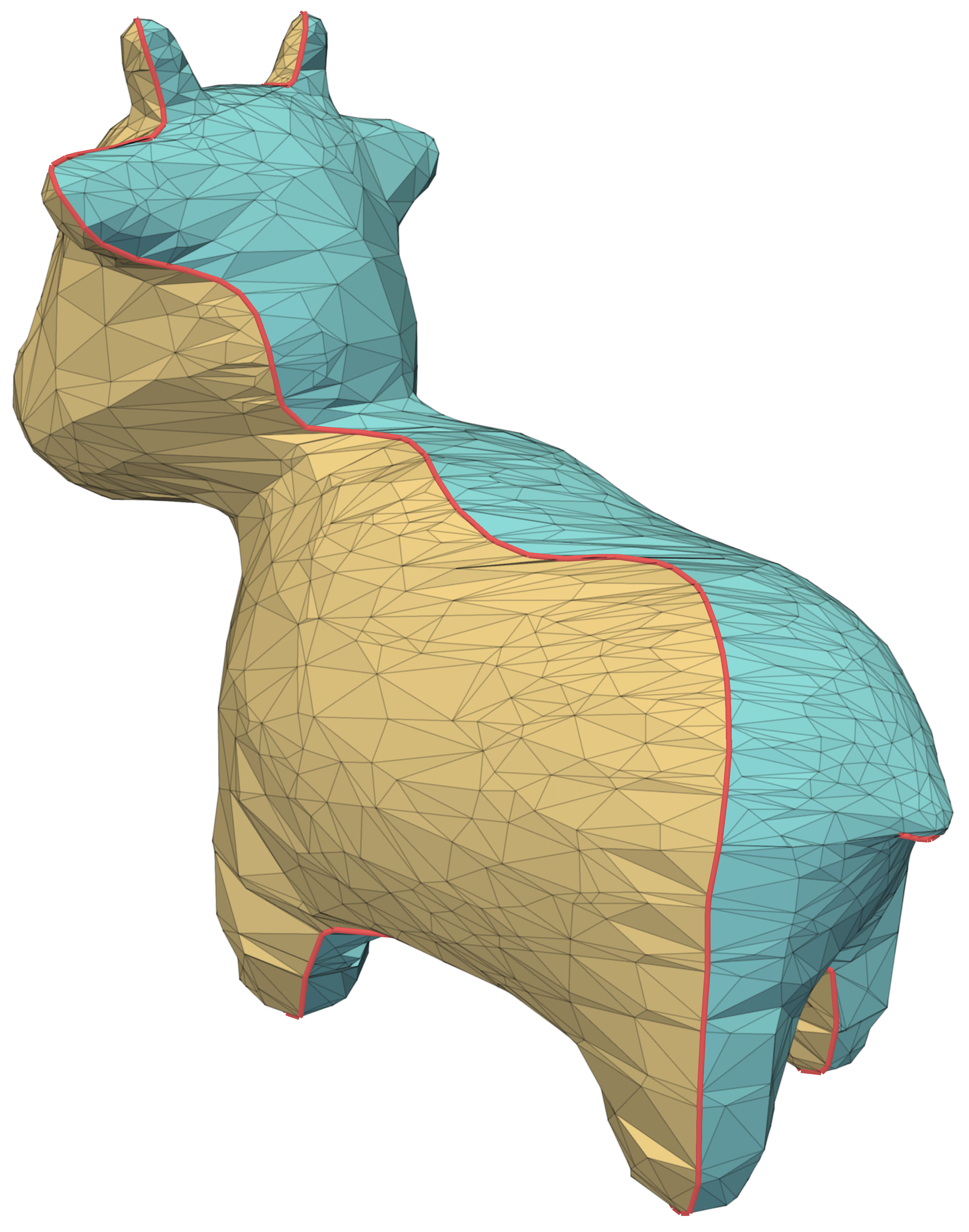} & 
\includegraphics[width=\figwidth]{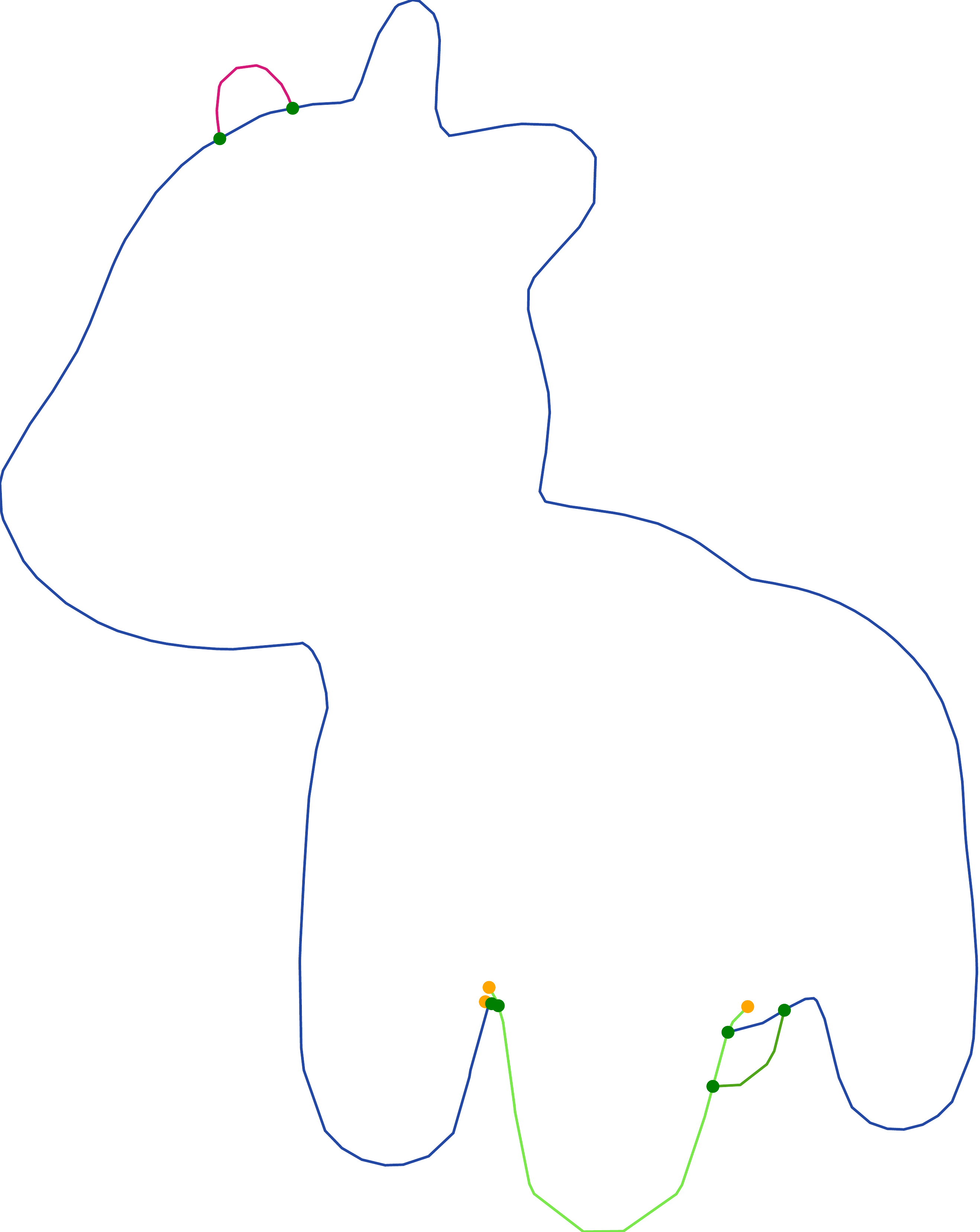} &
\includegraphics[width=\figwidth]{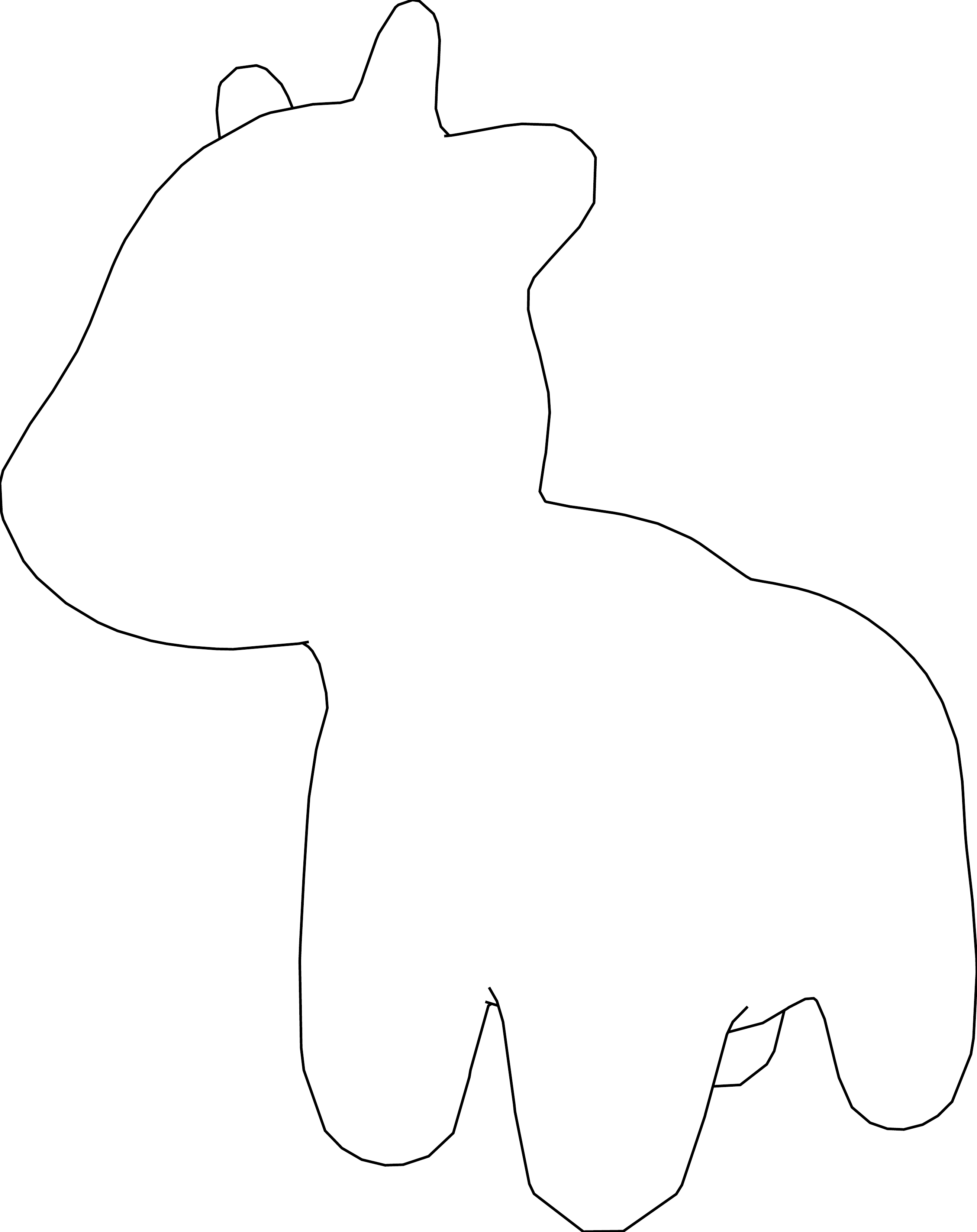} \\

\includegraphics[width=\figwidth]{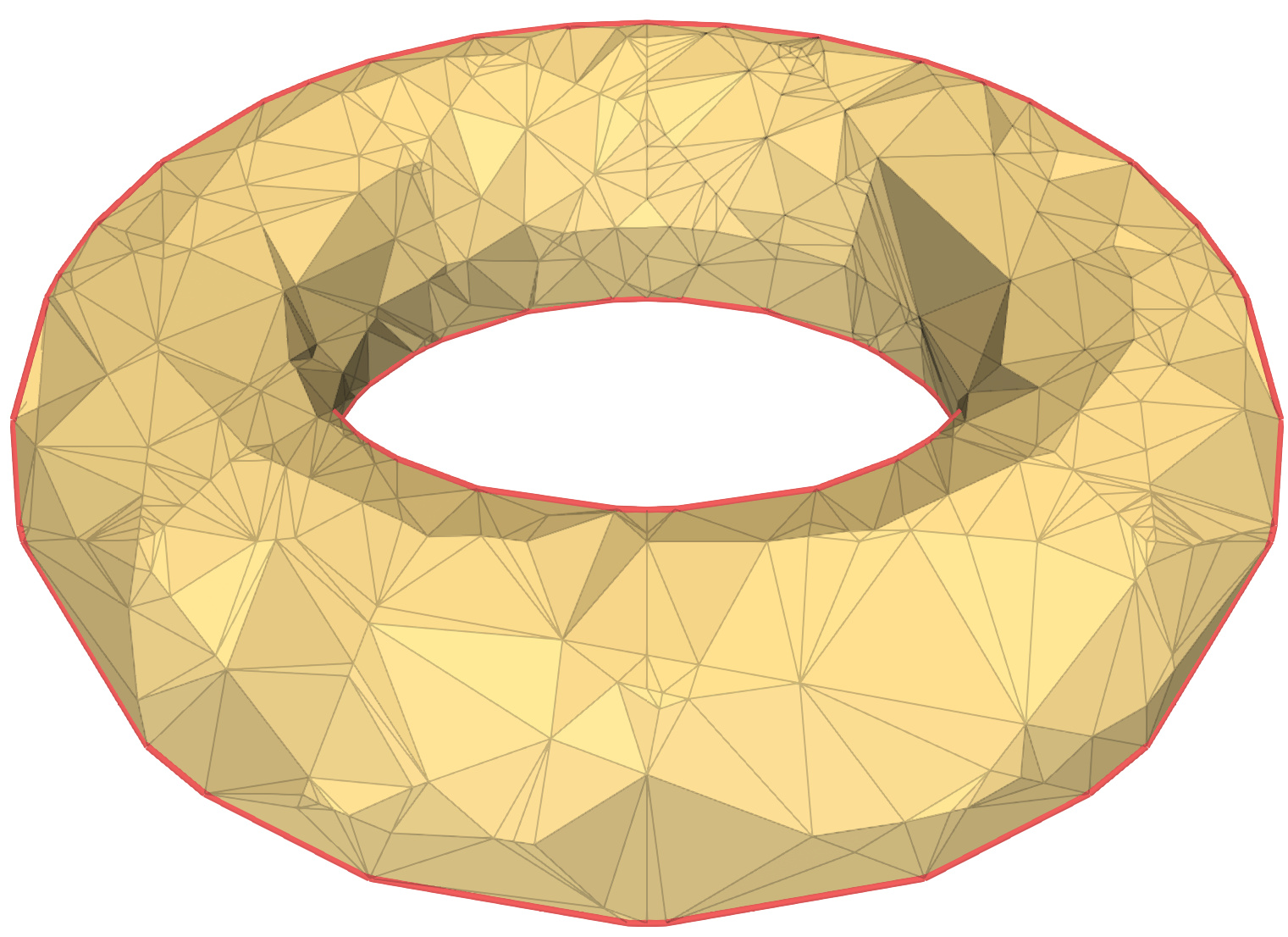} &
\includegraphics[width=\figwidth]{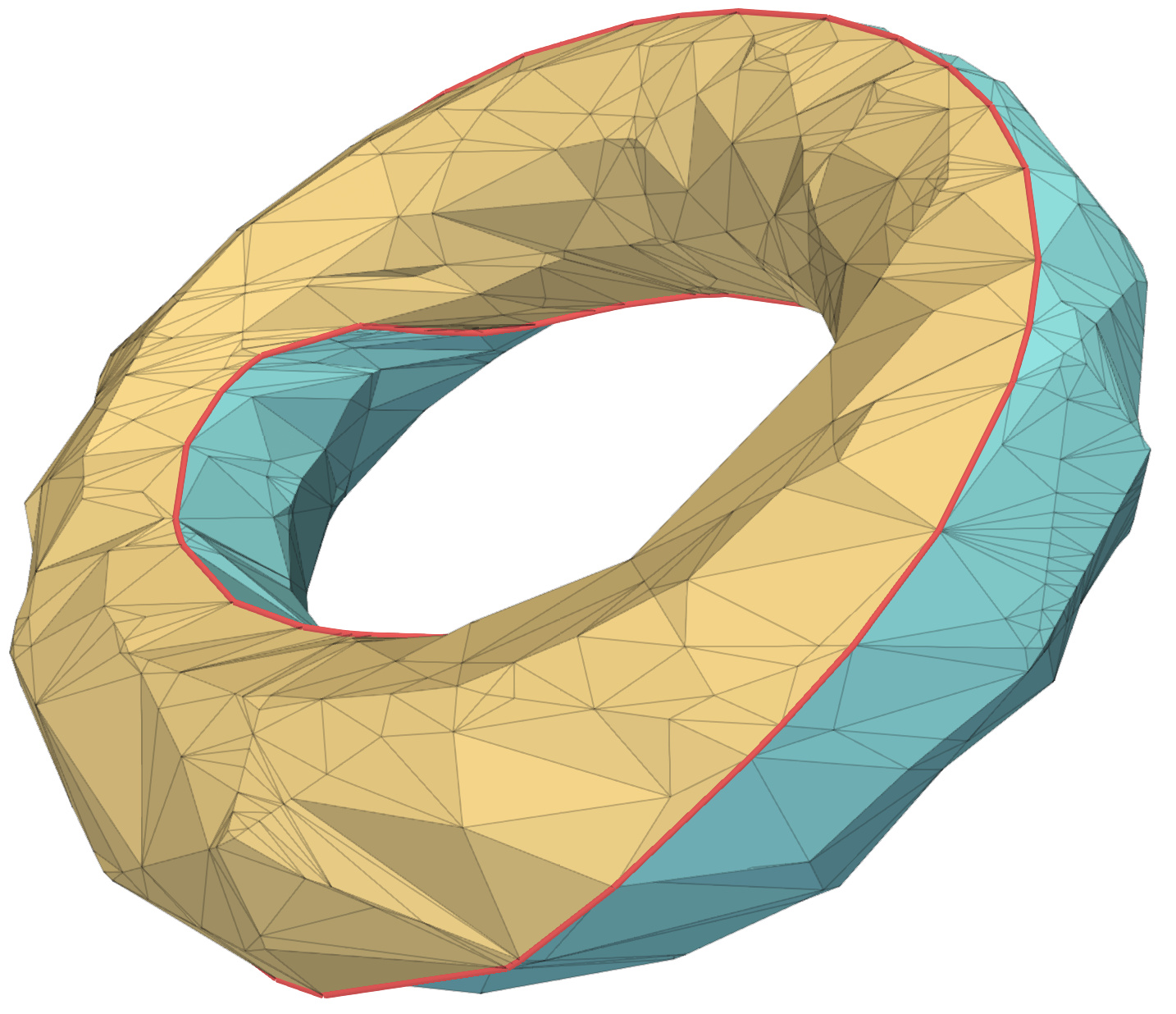} &
\includegraphics[width=\figwidth]{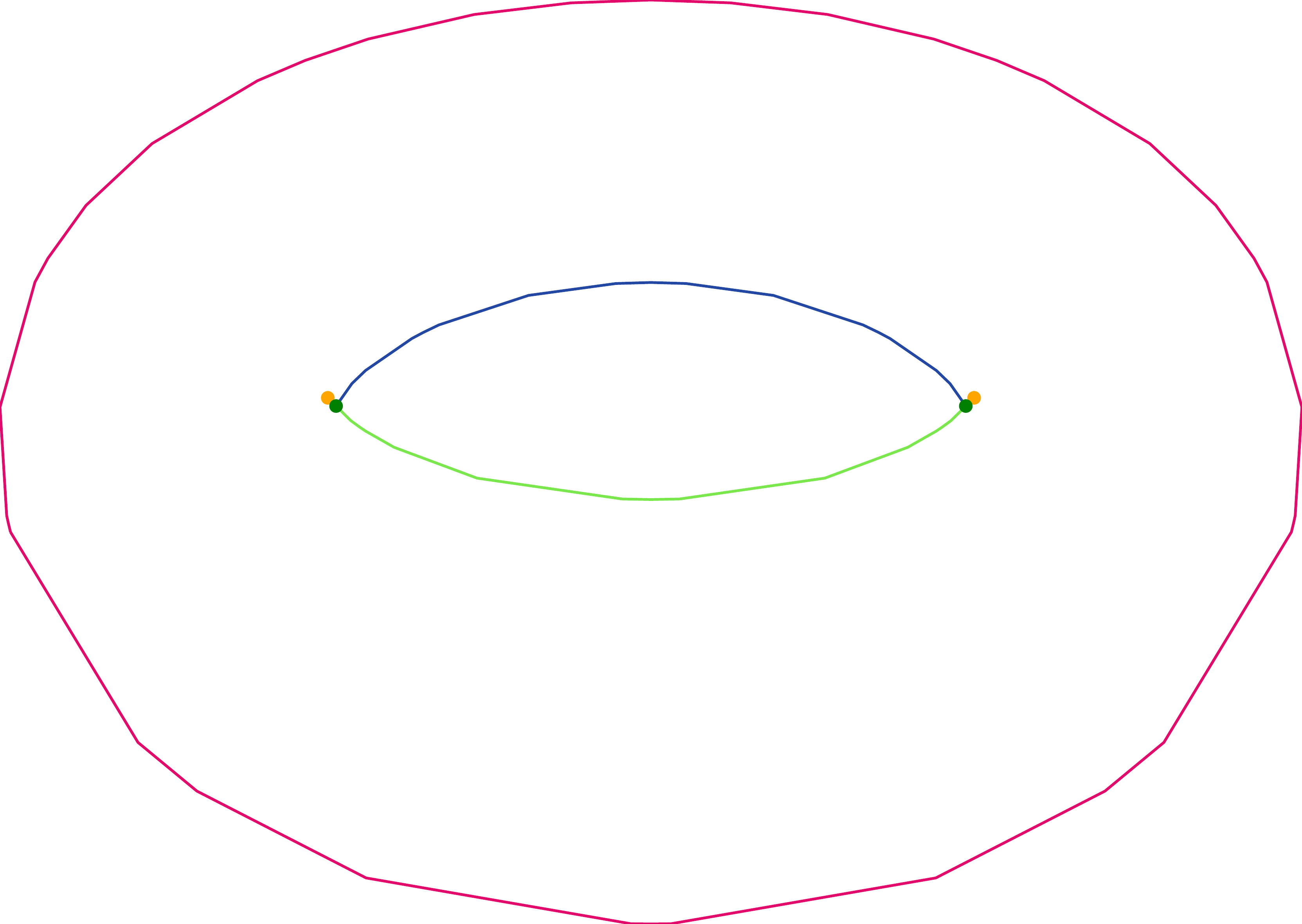} &
\includegraphics[width=\figwidth]{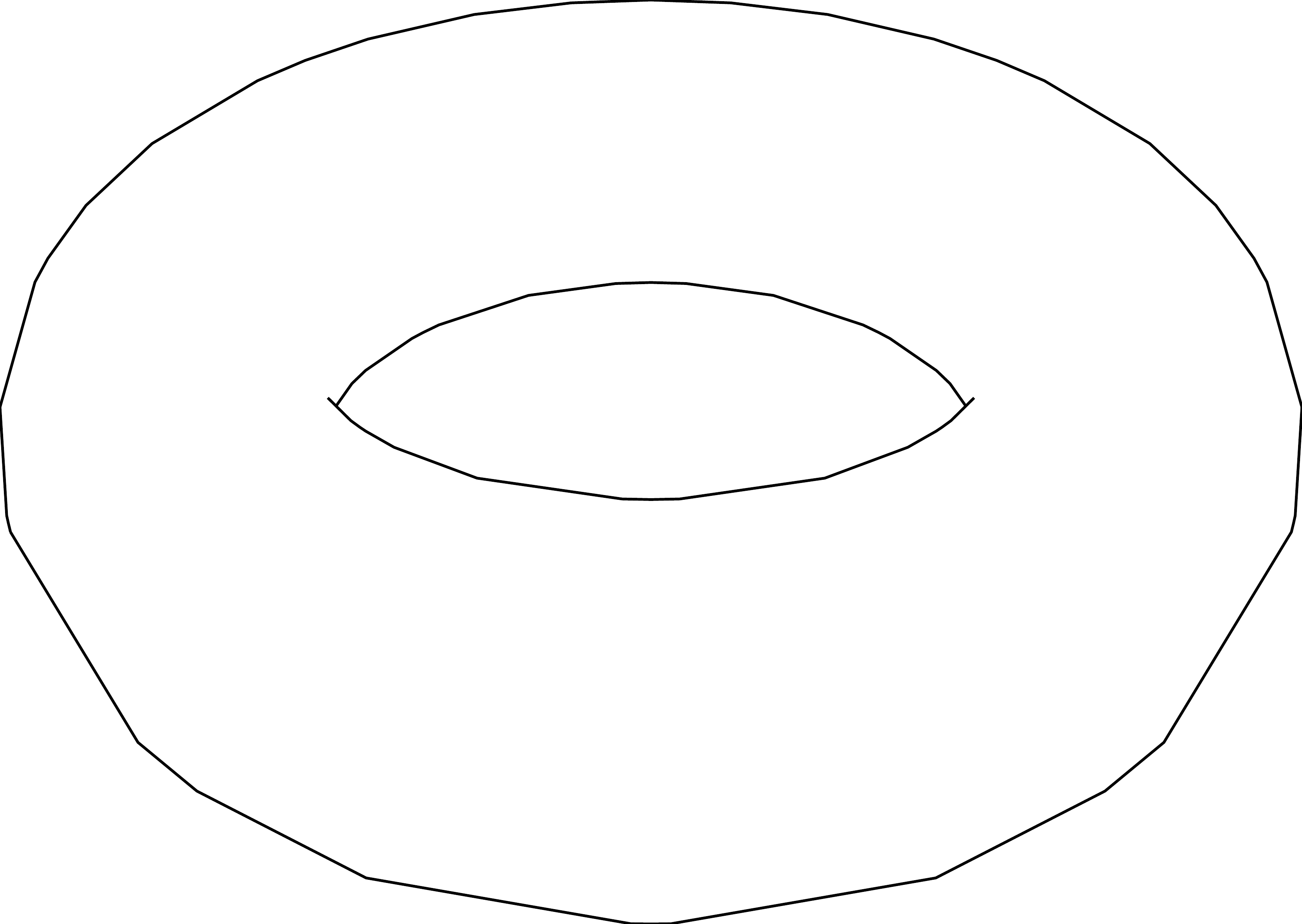} \\

\includegraphics[width=\figheight]{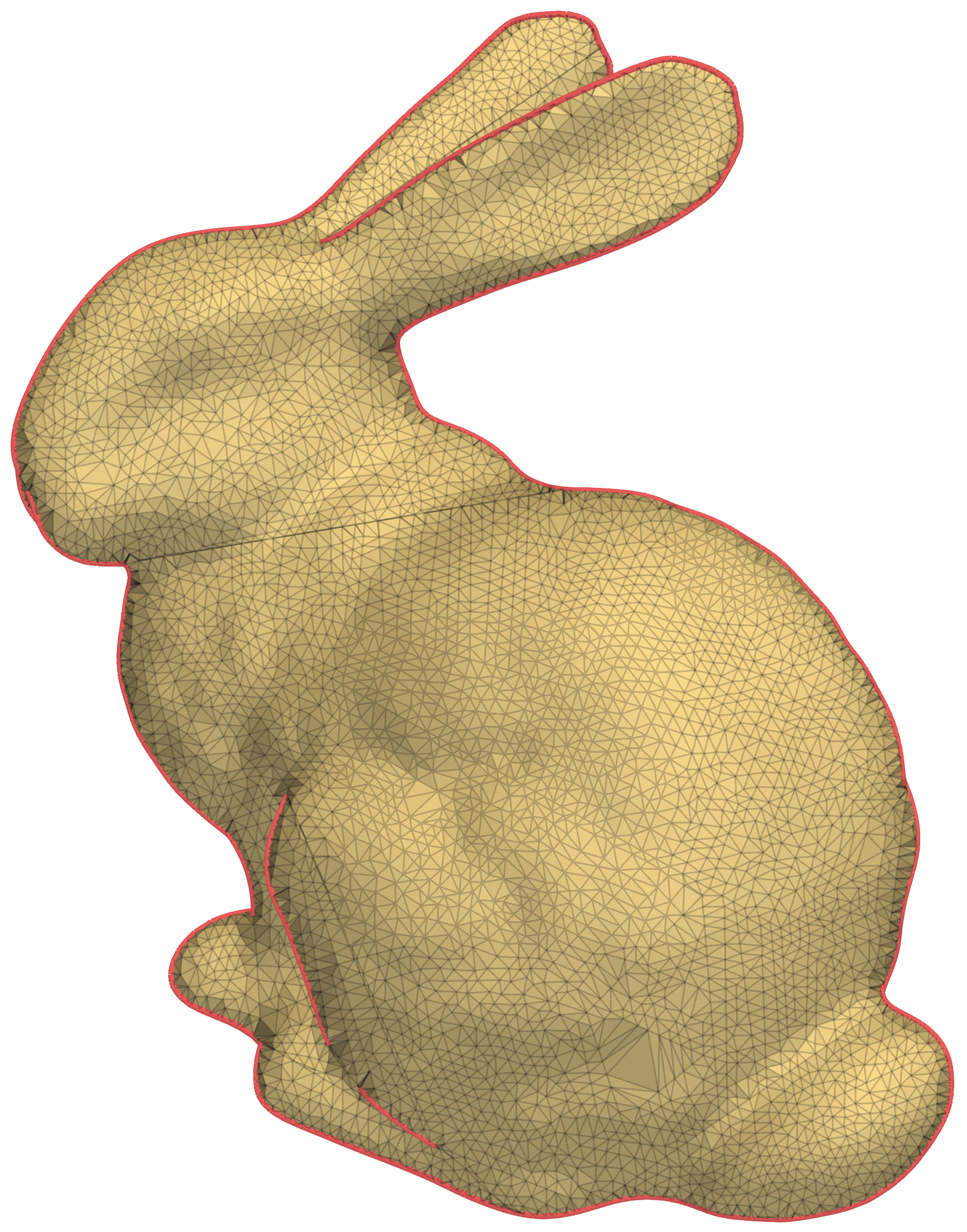} &
\includegraphics[width=\figheight]{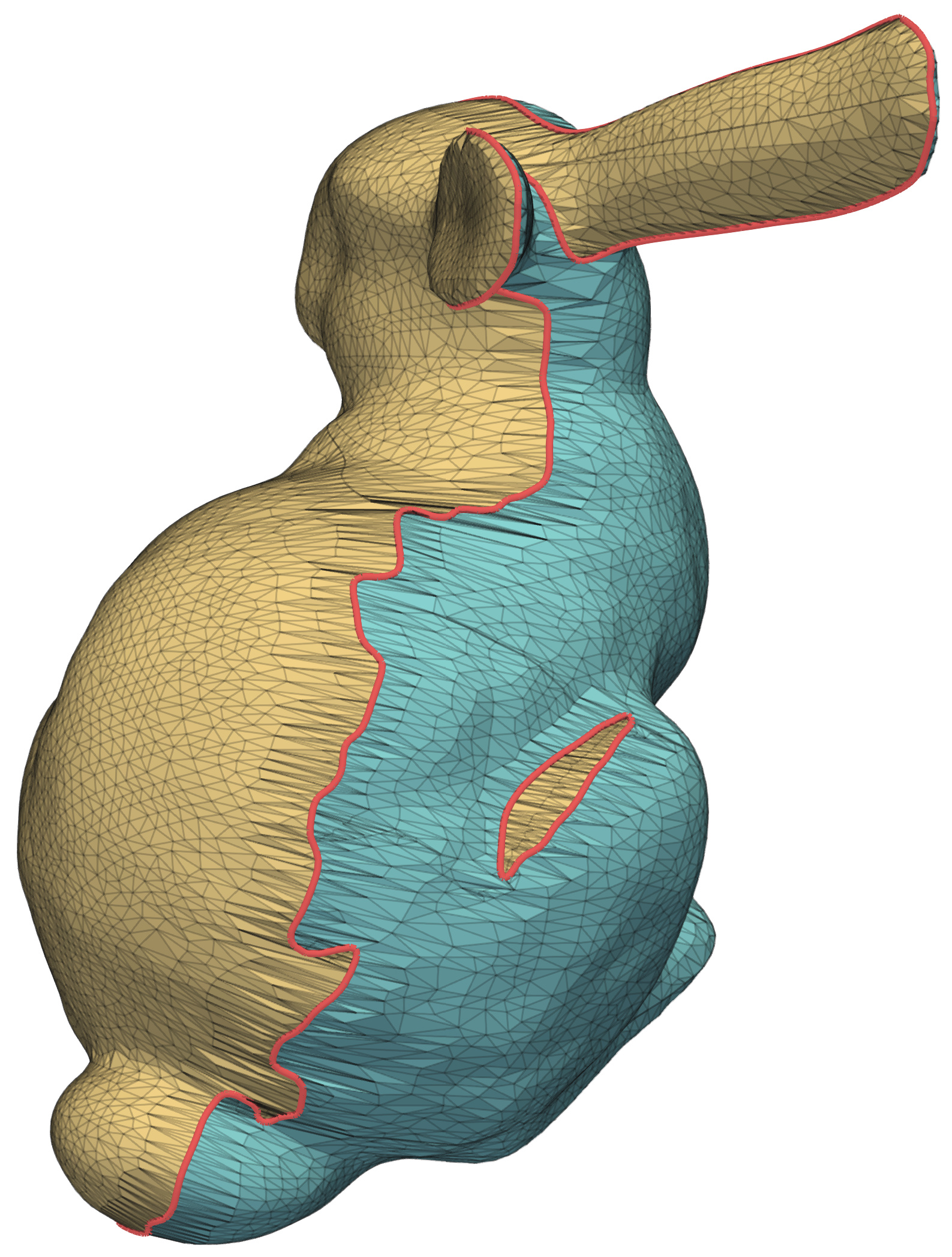} &
\includegraphics[width=\figheight]{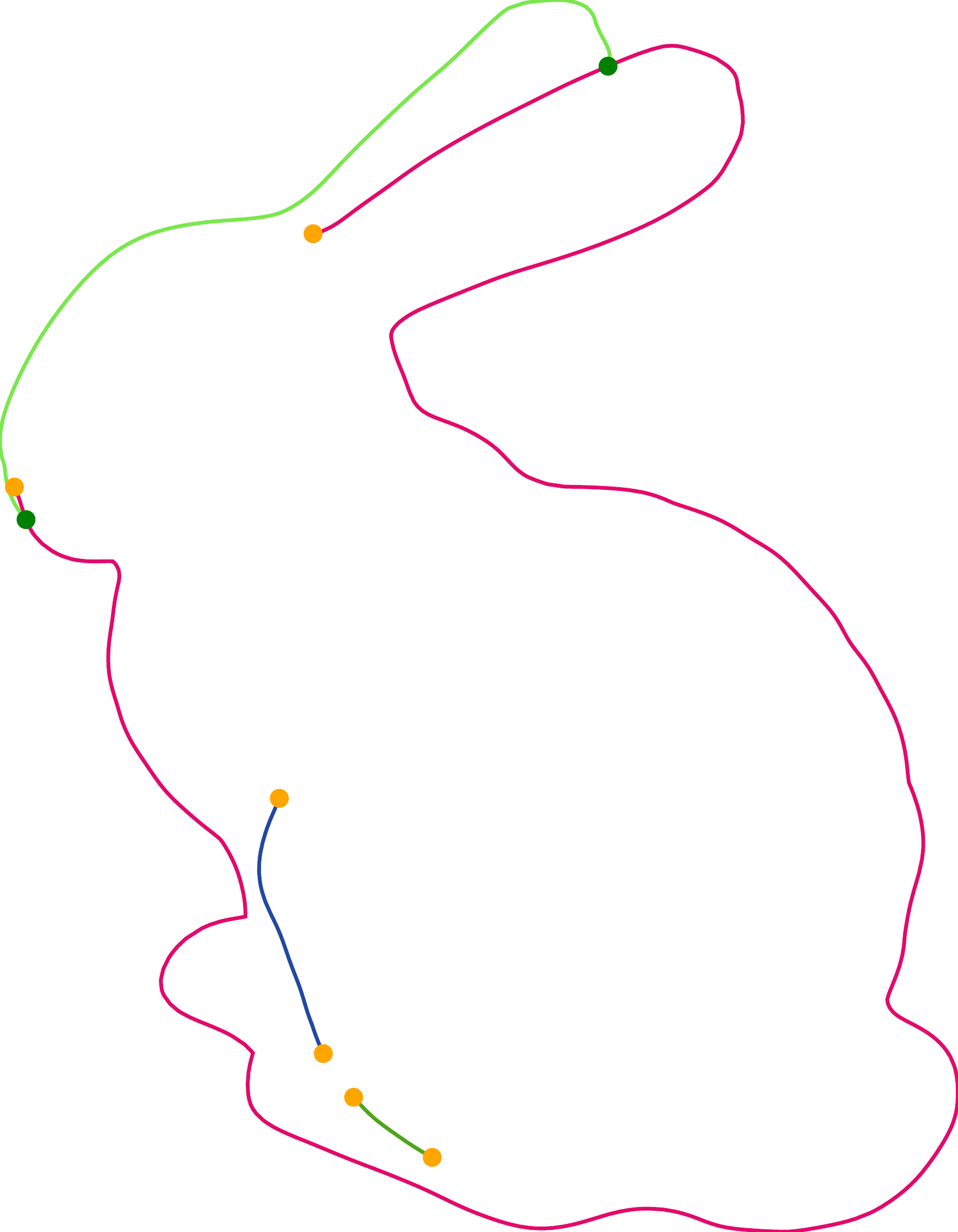} &
\includegraphics[width=\figheight]{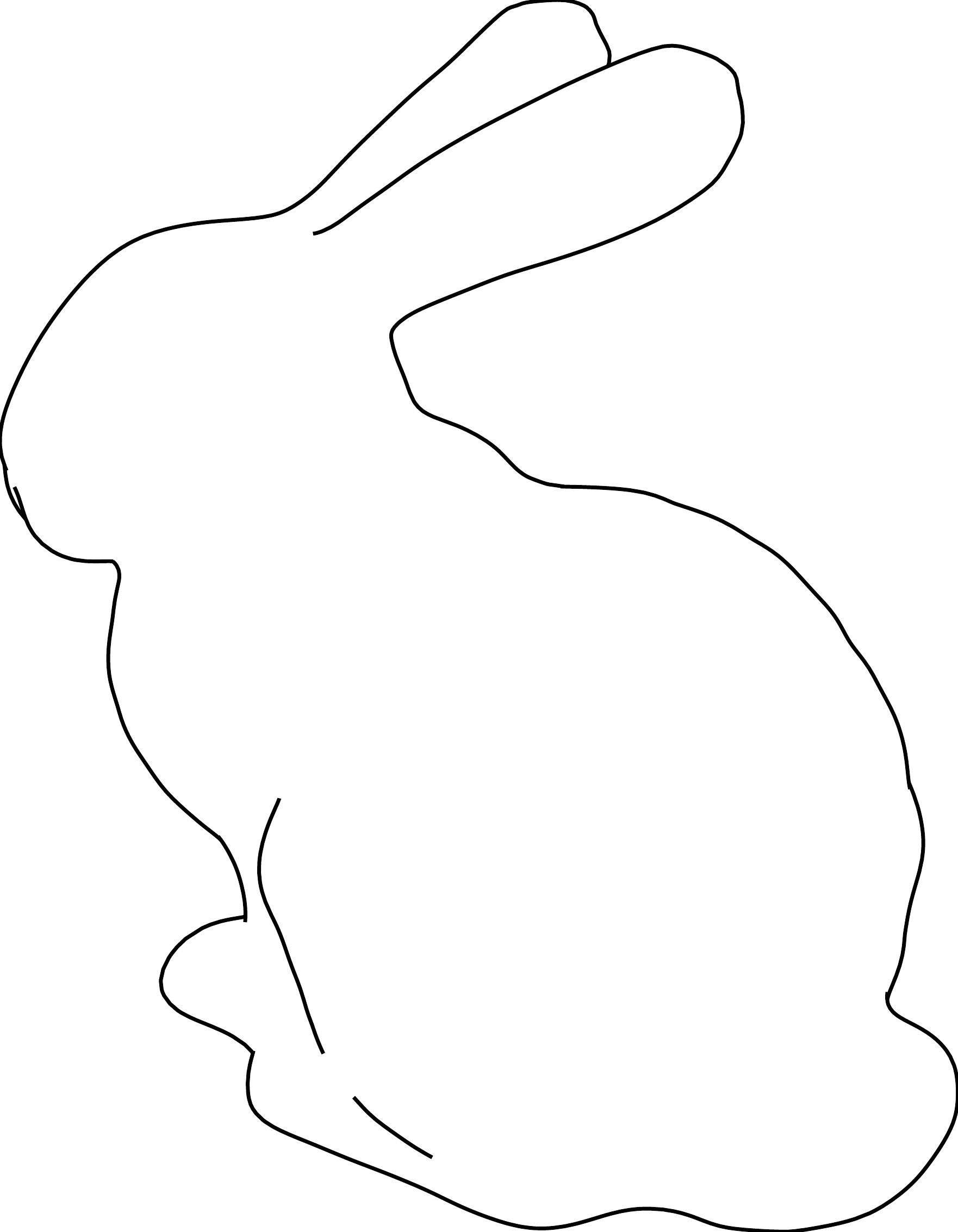} \\

\includegraphics[width=\figwidth]{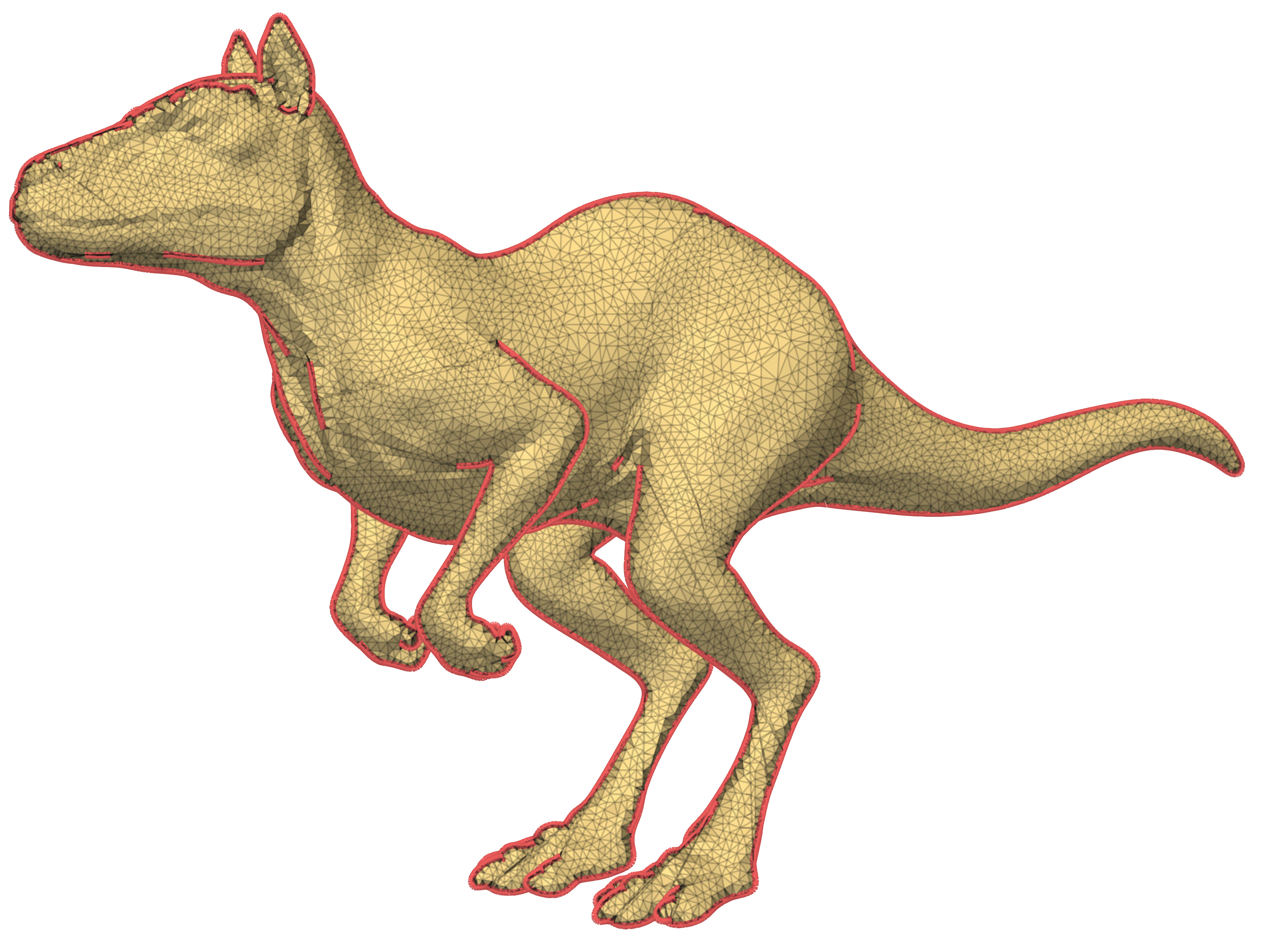} &
\includegraphics[width=\figheight]{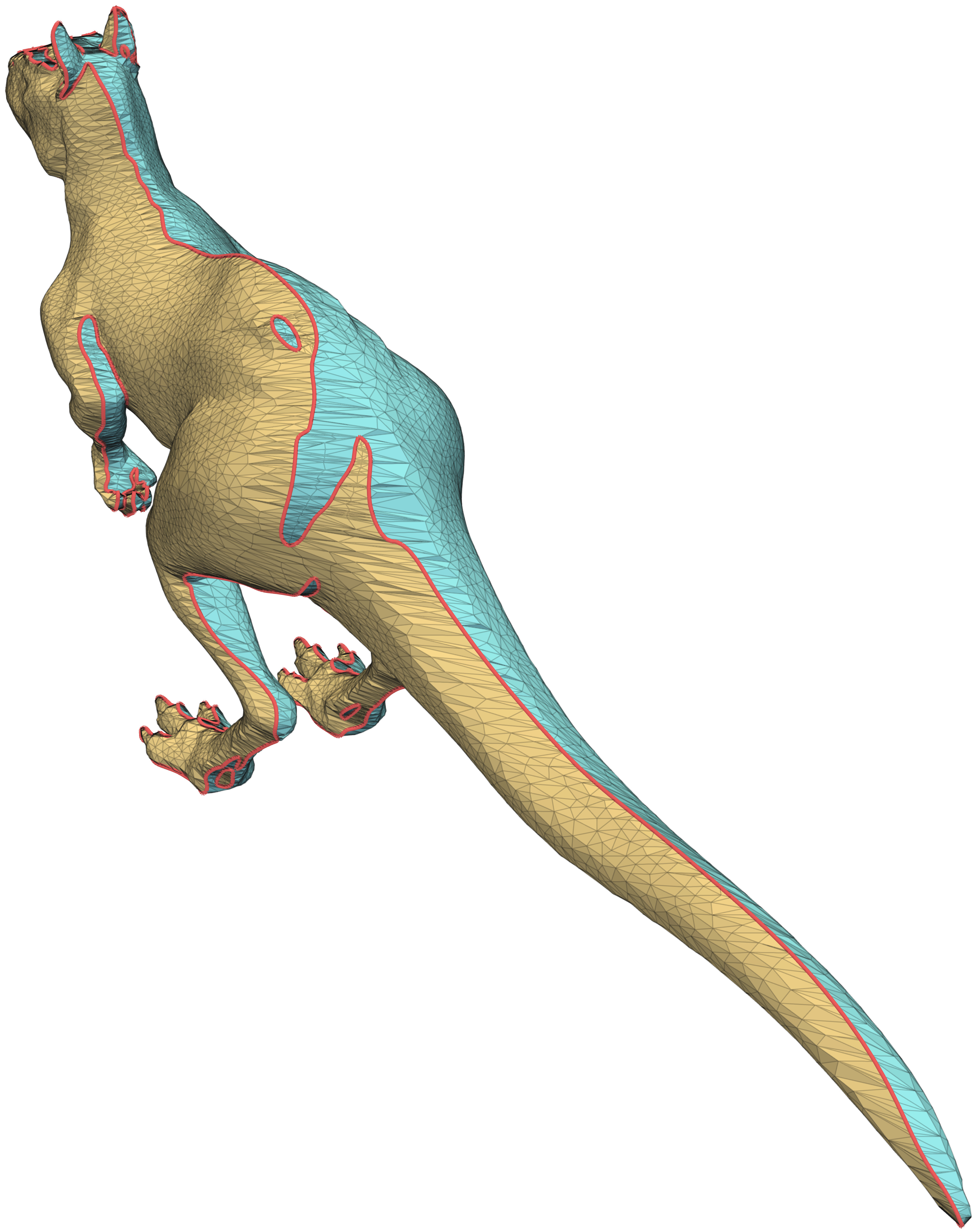} &
\includegraphics[width=\figwidth]{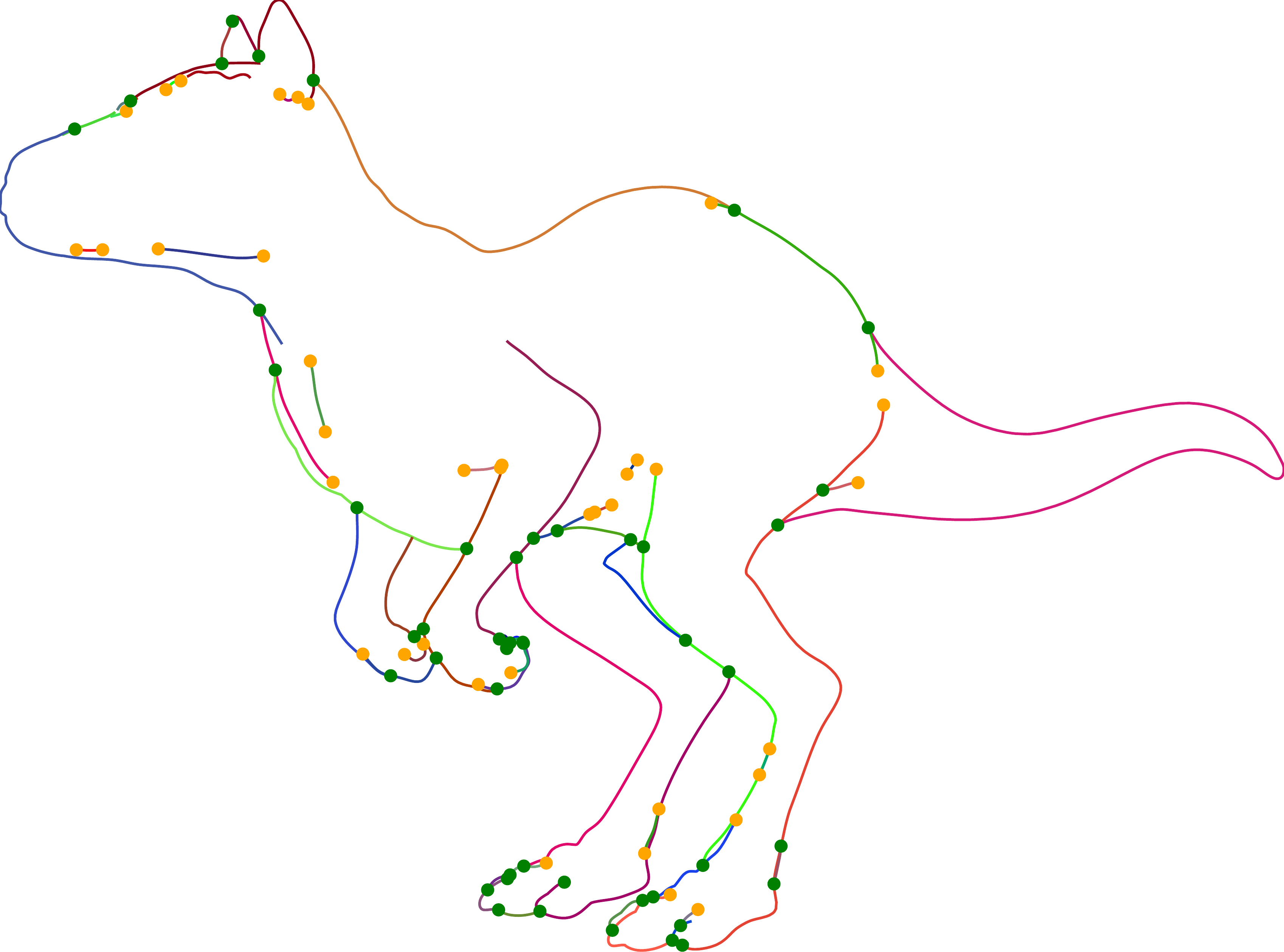} &
\includegraphics[width=\figwidth]{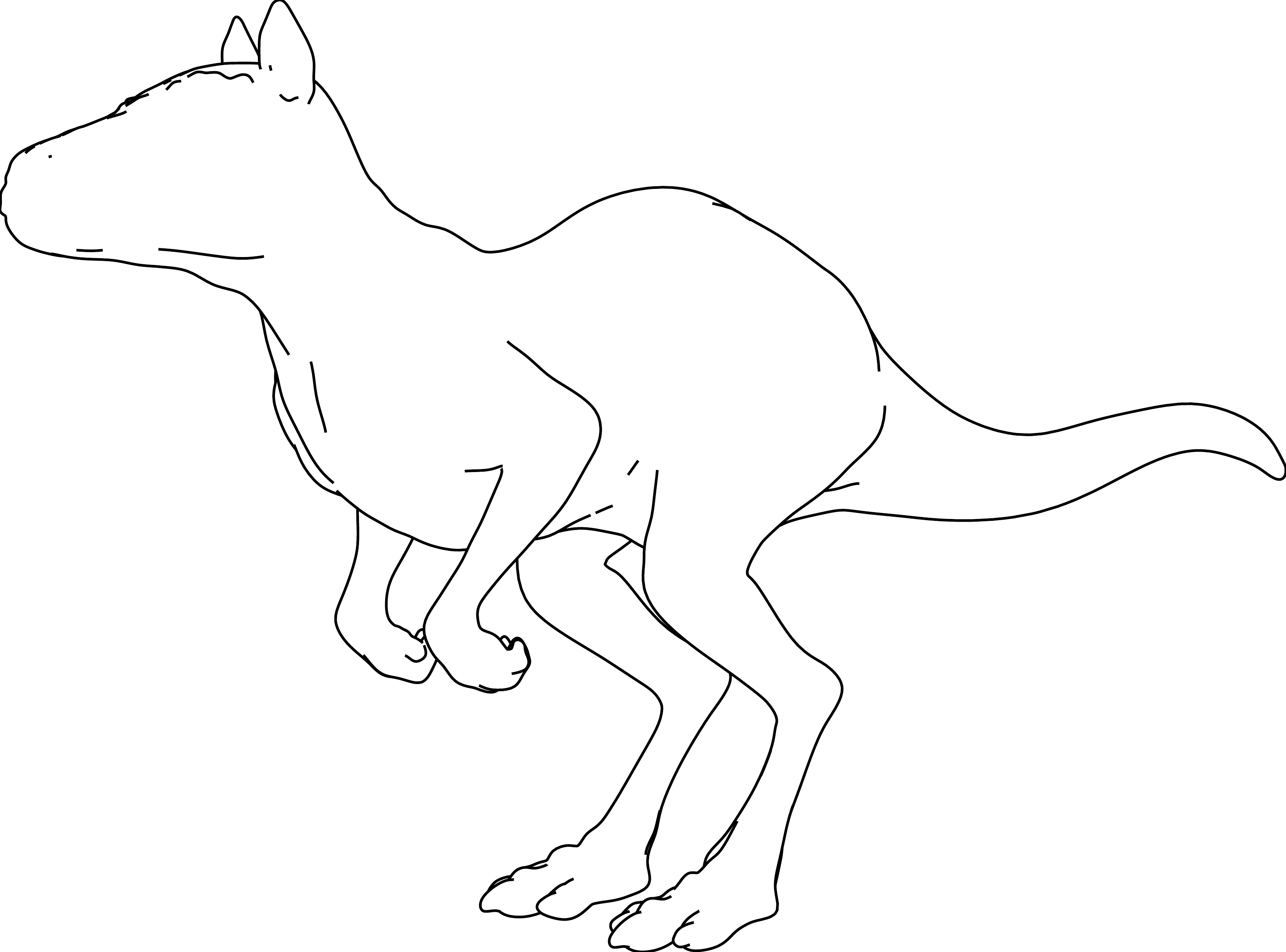} \\

\end{tabular}
\caption{
Examples of the ConTesse algorithm applied to various surfaces and camera views. Each example is a commonly-used mesh in geometry processing, treated as a Catmull-Clark base mesh. (a) Camera view of the output mesh, (b) side view of the output mesh, (c) view graph (curve network) of the visible occluding contours of the output mesh, with cusps marked in orange and 2D intersections in green, and
(d) occluding contours, after computing visibility.
(Public domain Spot model by Keenan Crane, Killeroo \copyright\ headus.com.au.)
\label{fig:results}
}
\end{figure*}
Figure \ref{fig:results} shows typical inputs and outputs of the method, illustrating the different complexities of meshes that can be handled correctly. 
In each example, the output mesh is completely consistent, and simple, clean view maps are produced, which can be simplified further for stylization (Figure \ref{fig:stylized_results}). 
\newcommand{\sfigwidth}{1.5in}
\newcommand{\sfigheight}{1.5in}

\begin{figure*}
\includegraphics[height=1.5in]{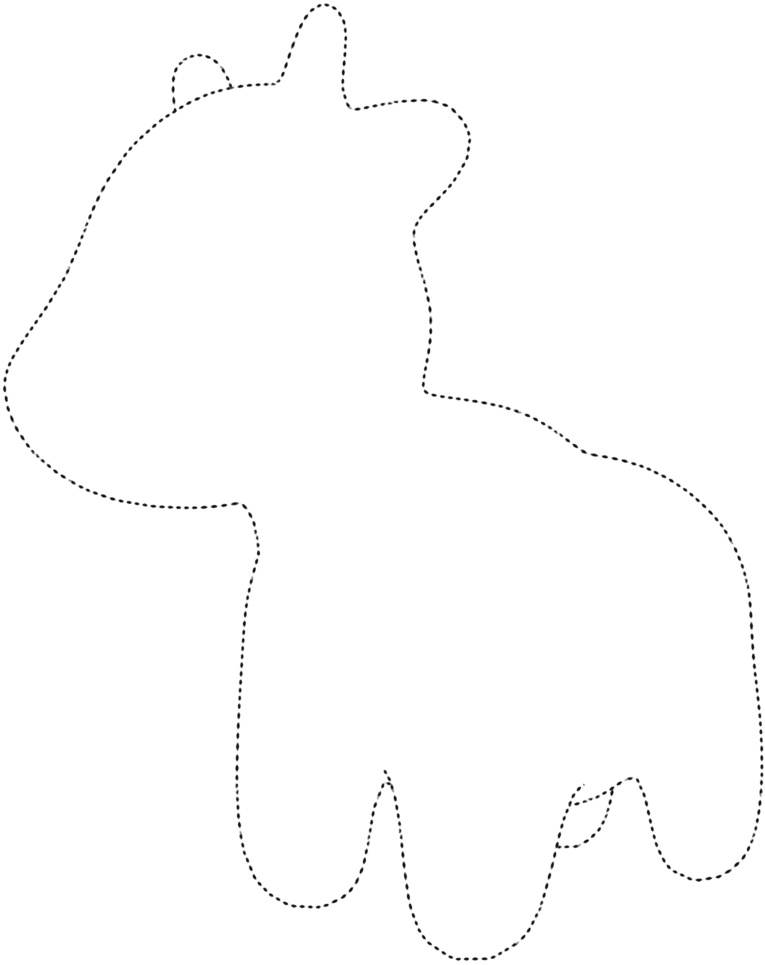}
\includegraphics[width=2in]{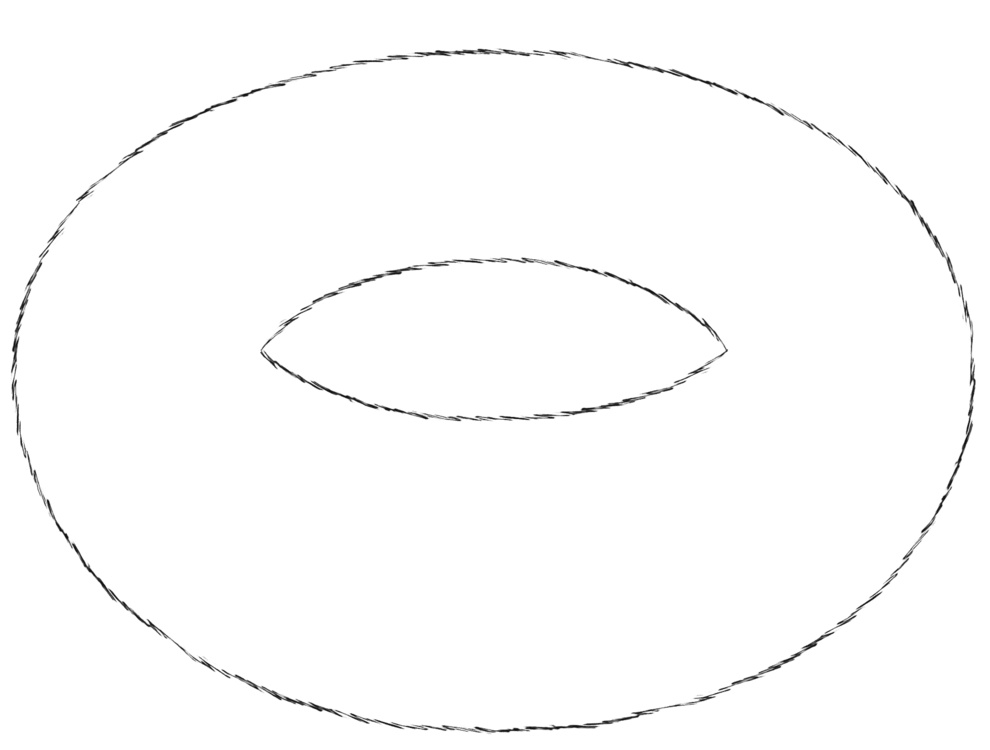} 
\includegraphics[height=1.5in]{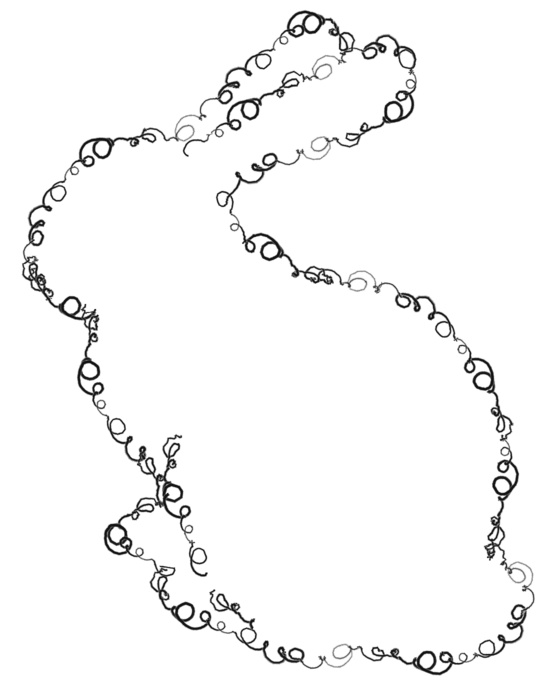} 
\includegraphics[height=1.5in]{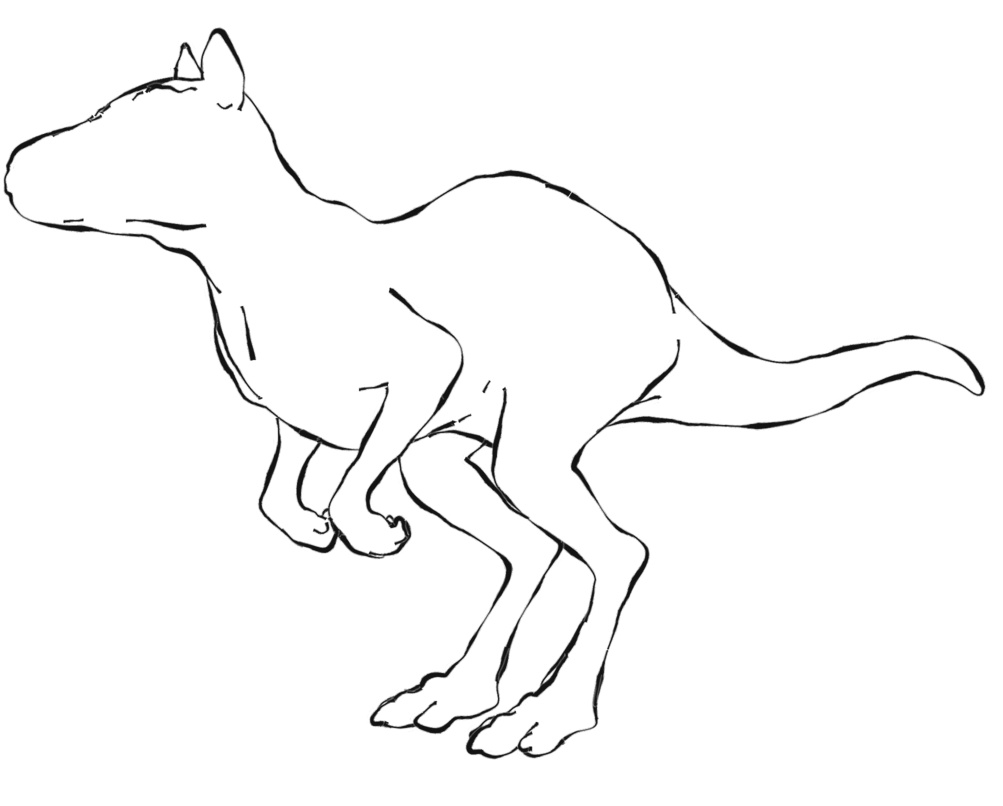}  \\
\includegraphics[height=\sfigheight]{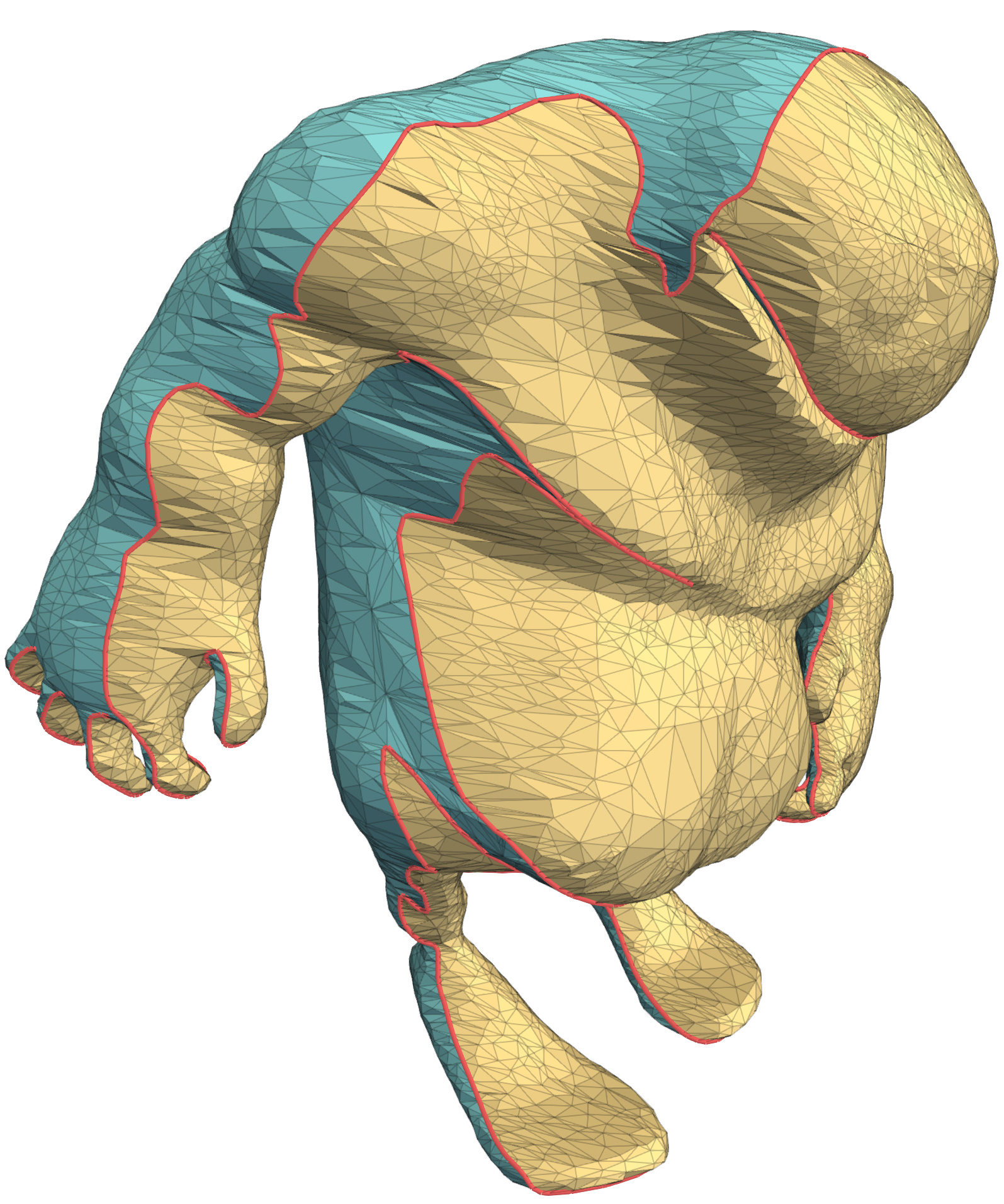}
\includegraphics[height=\sfigheight]{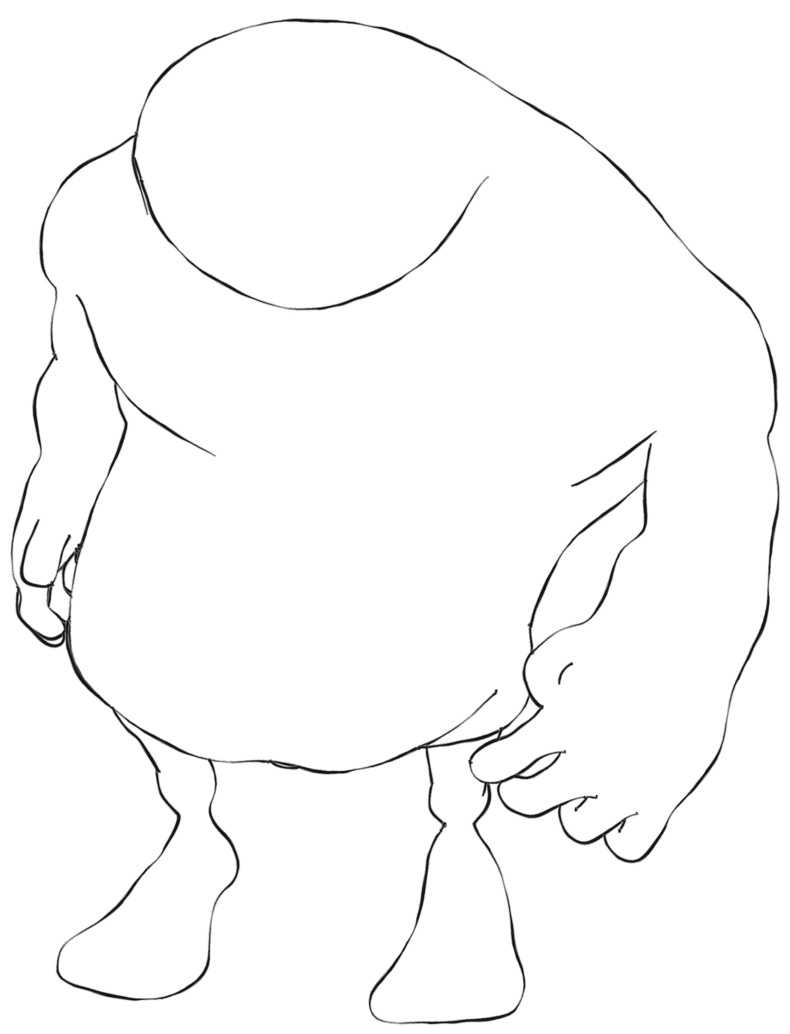}
\includegraphics[height=\sfigheight]{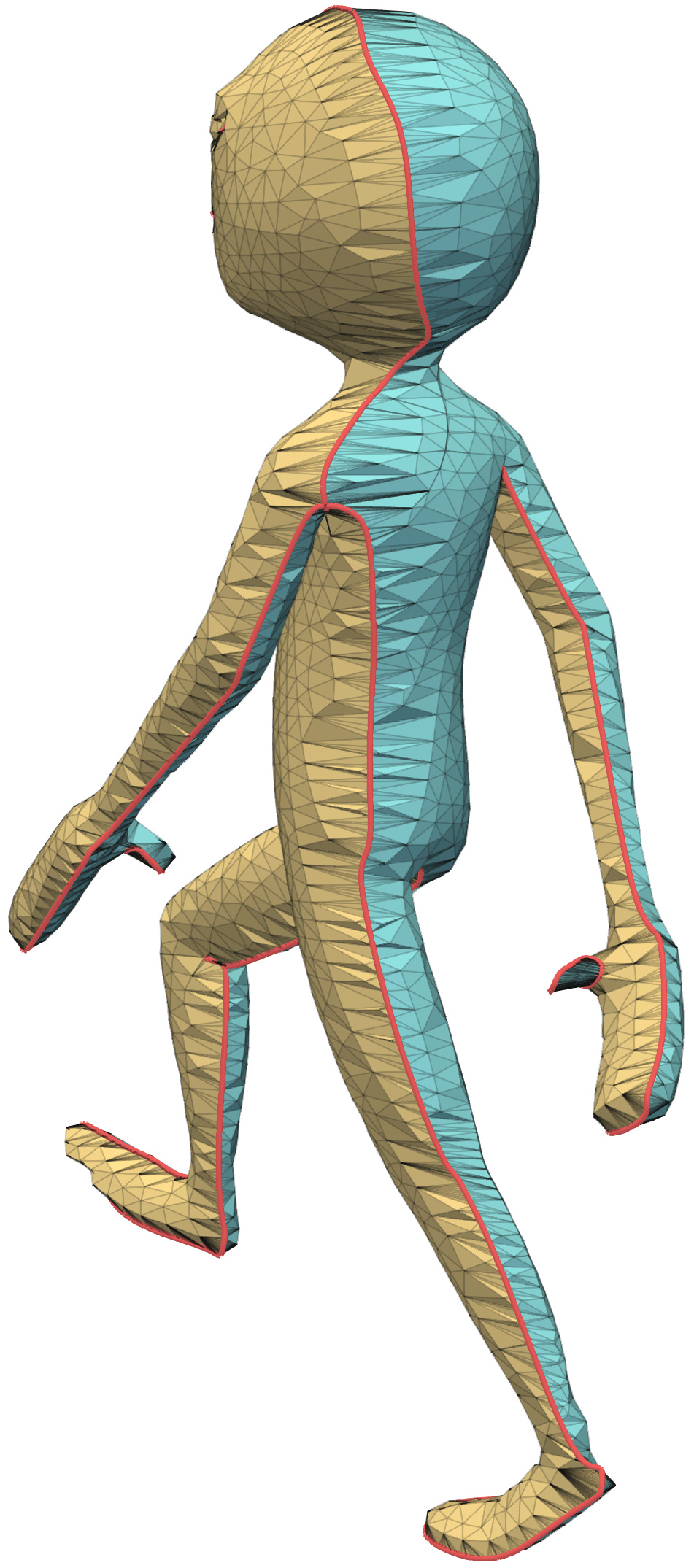}
\includegraphics[height=\sfigheight]{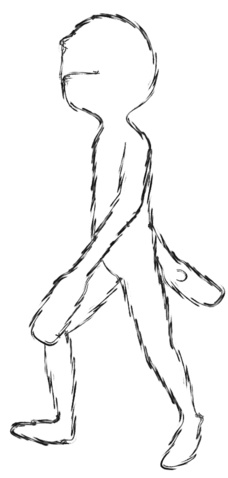} \includegraphics[height=\sfigheight]{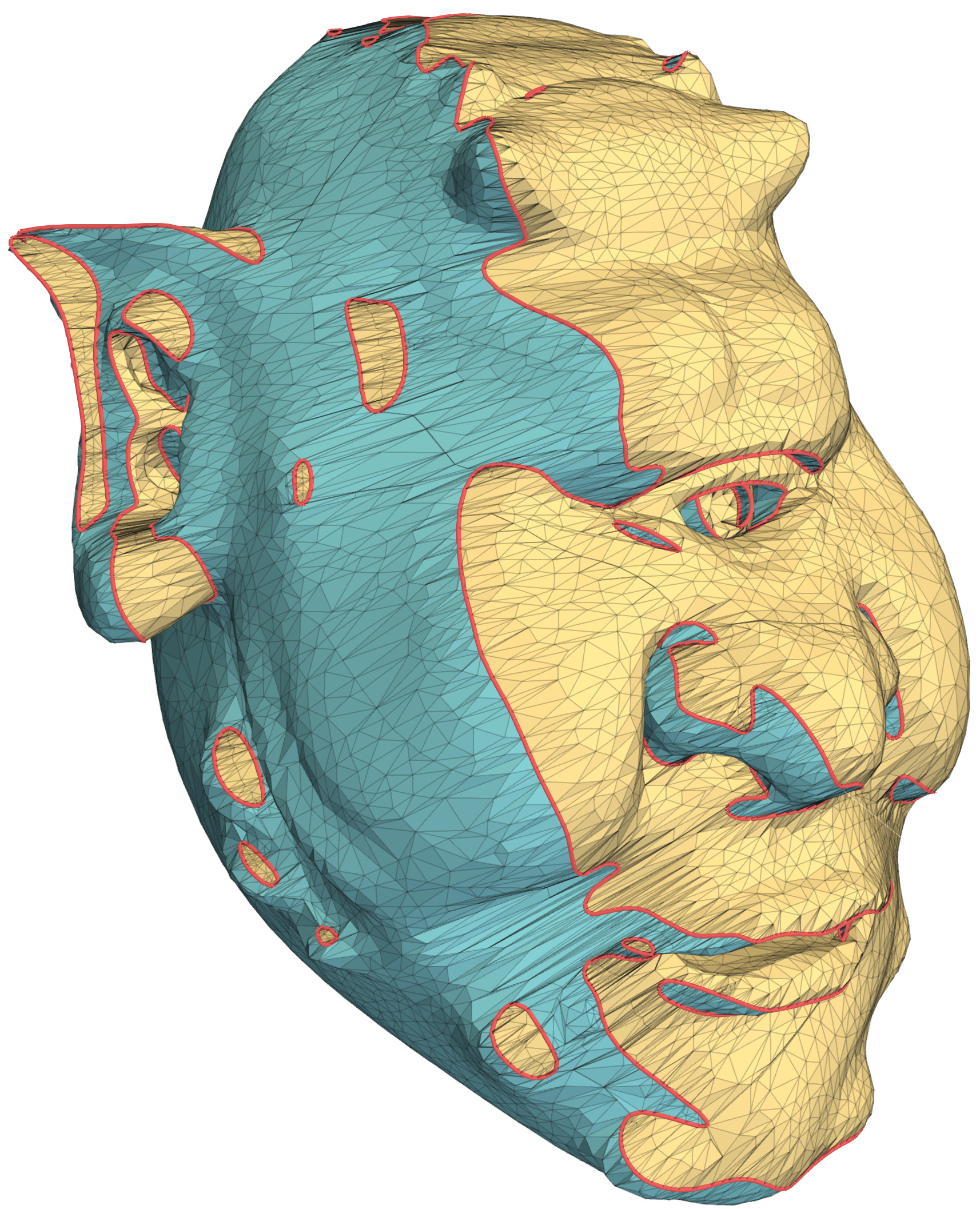}
\includegraphics[height=\sfigheight]{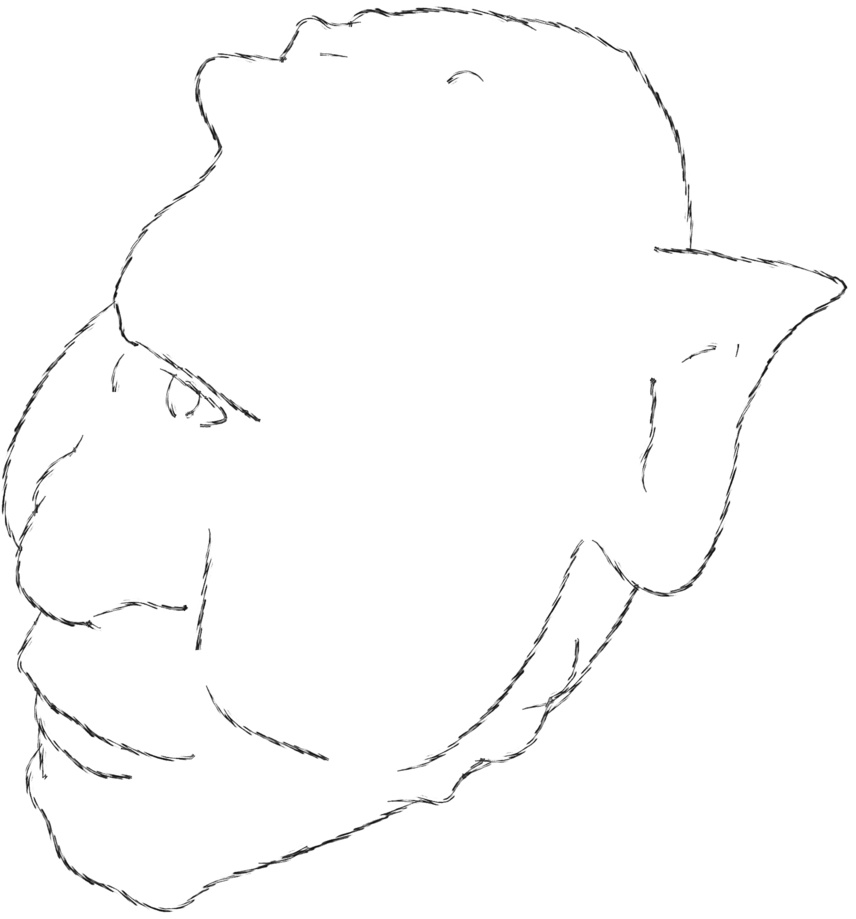} \\
\includegraphics[height=\sfigheight]{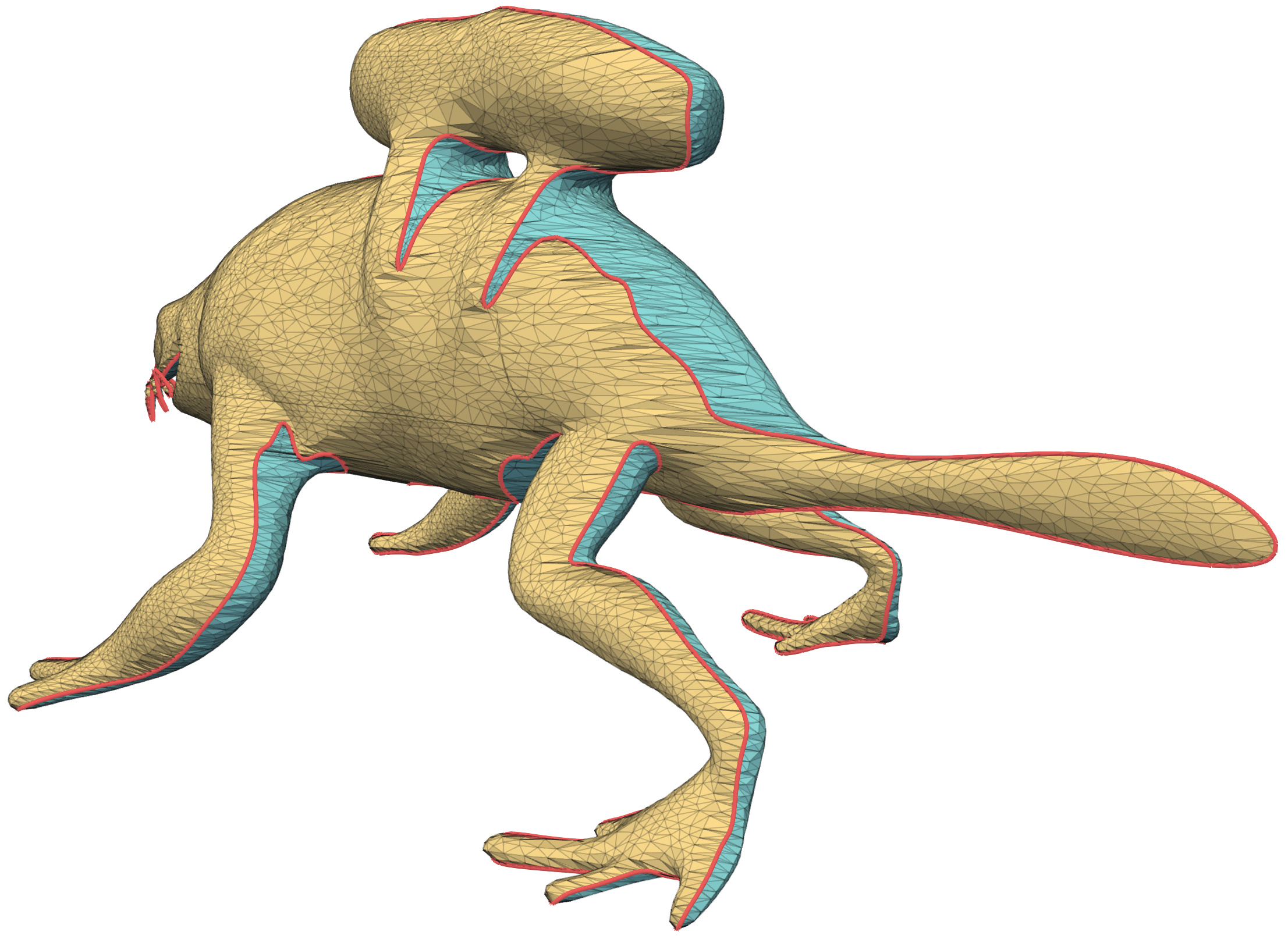}
\includegraphics[height=\sfigheight]{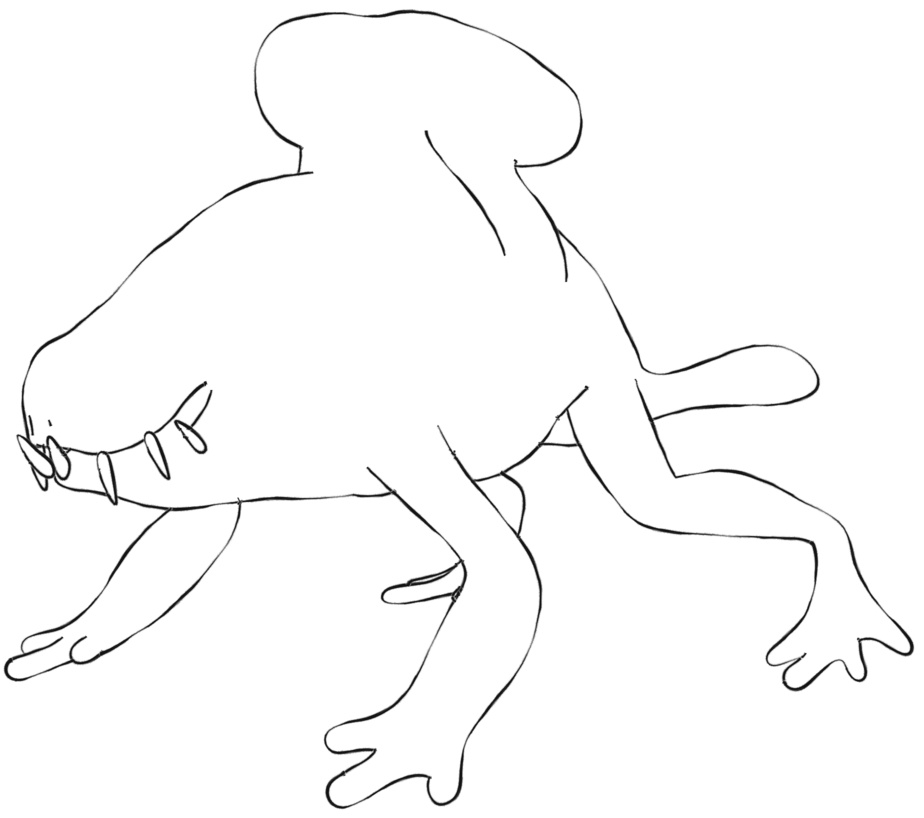} 
\includegraphics[height=\sfigheight]{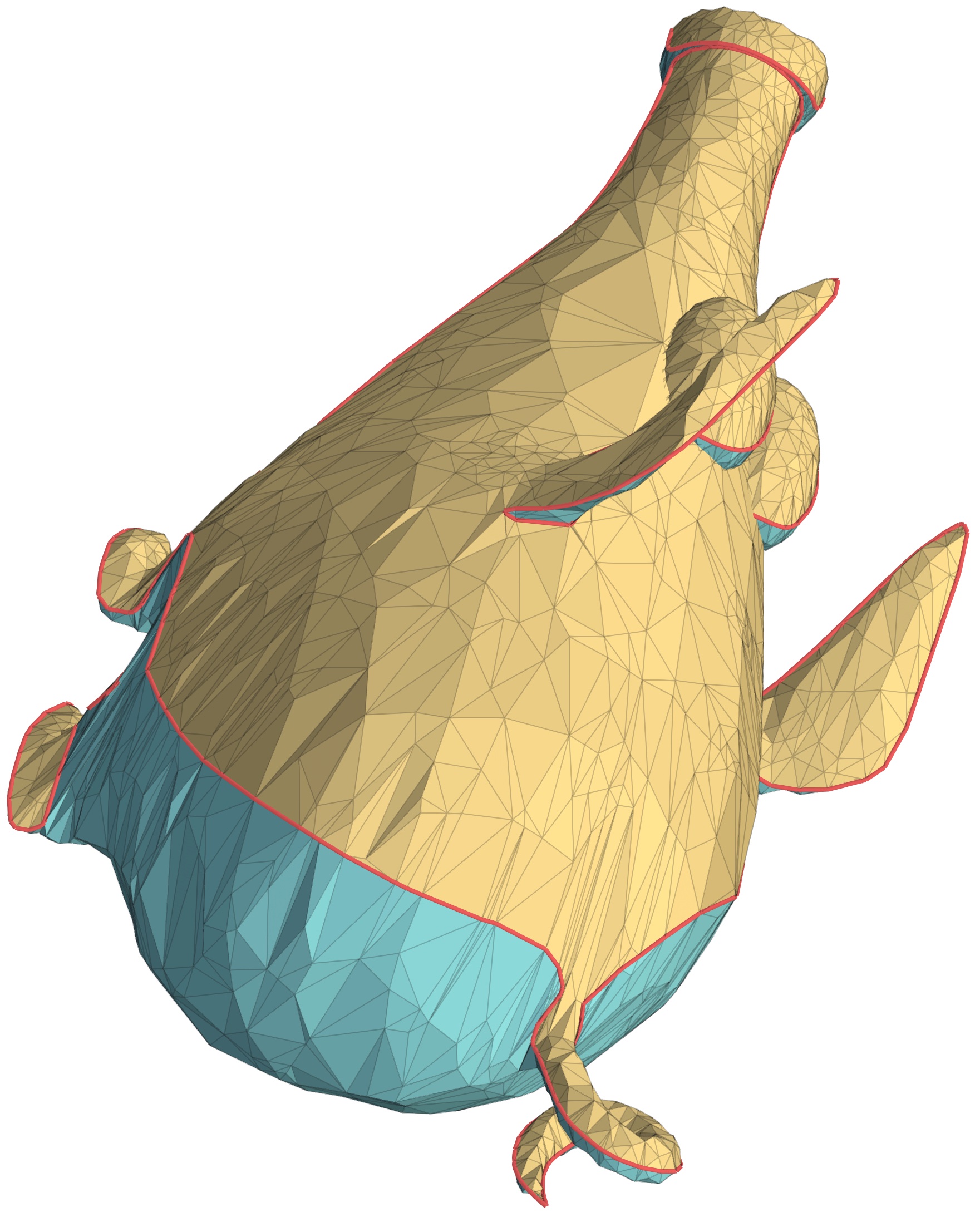}
\includegraphics[height=\sfigheight]{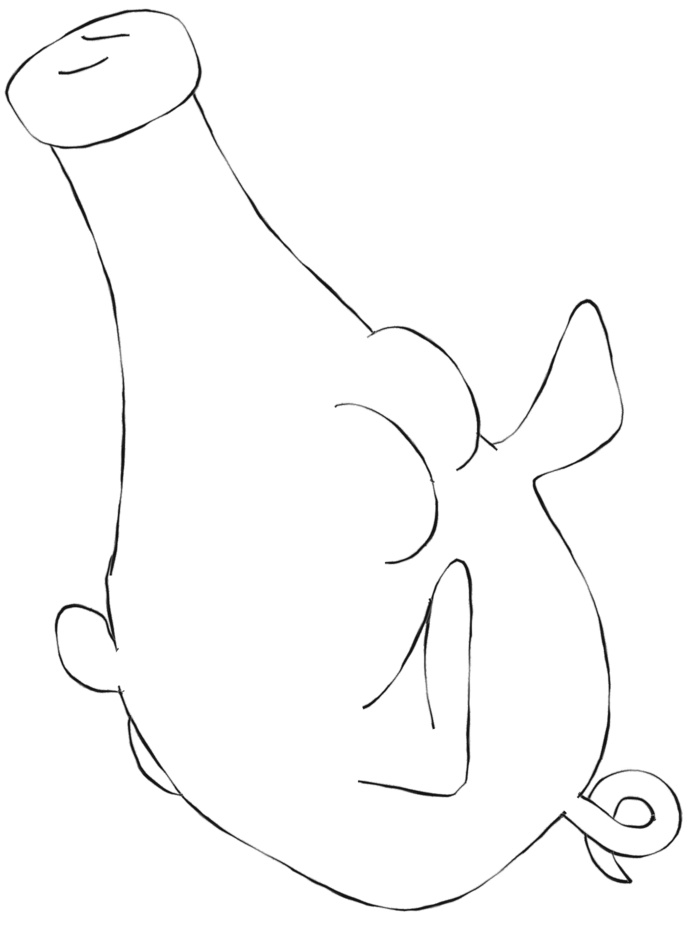}
\caption{Stylized versions of the contours from Figure \ref{fig:results}, and results from three more models in our dataset. Note that we do not render mesh self-intersections. 
(Public domain Pig, ogre, and Spot models by Keenan Crane,
Killeroo \copyright\ headus.com.au,
Bigguy and Monster Frog \copyright\ Bay Raitt, Walking Man \copyright\ Ryan Dale.)
\label{fig:stylized_results}} 
\end{figure*}

Finer-scale meshes may be obtained by increasing the initial subdivision level (Figure \ref{fig:subdivided_results}) or by increasing the sampling density in the triangulation and lifting step.
\begin{figure}
\centering
\includegraphics[width=1.5in]{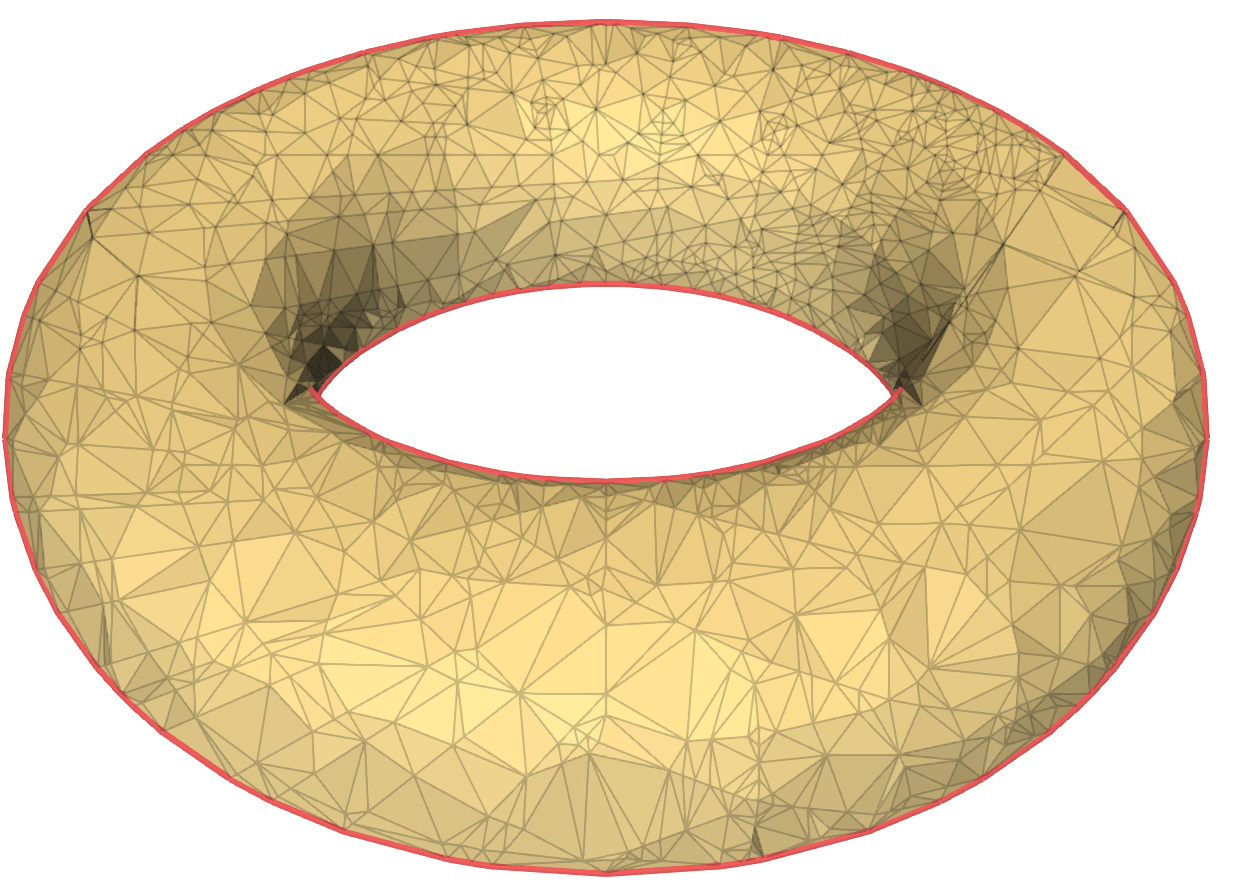}
\includegraphics[width=1.5in]{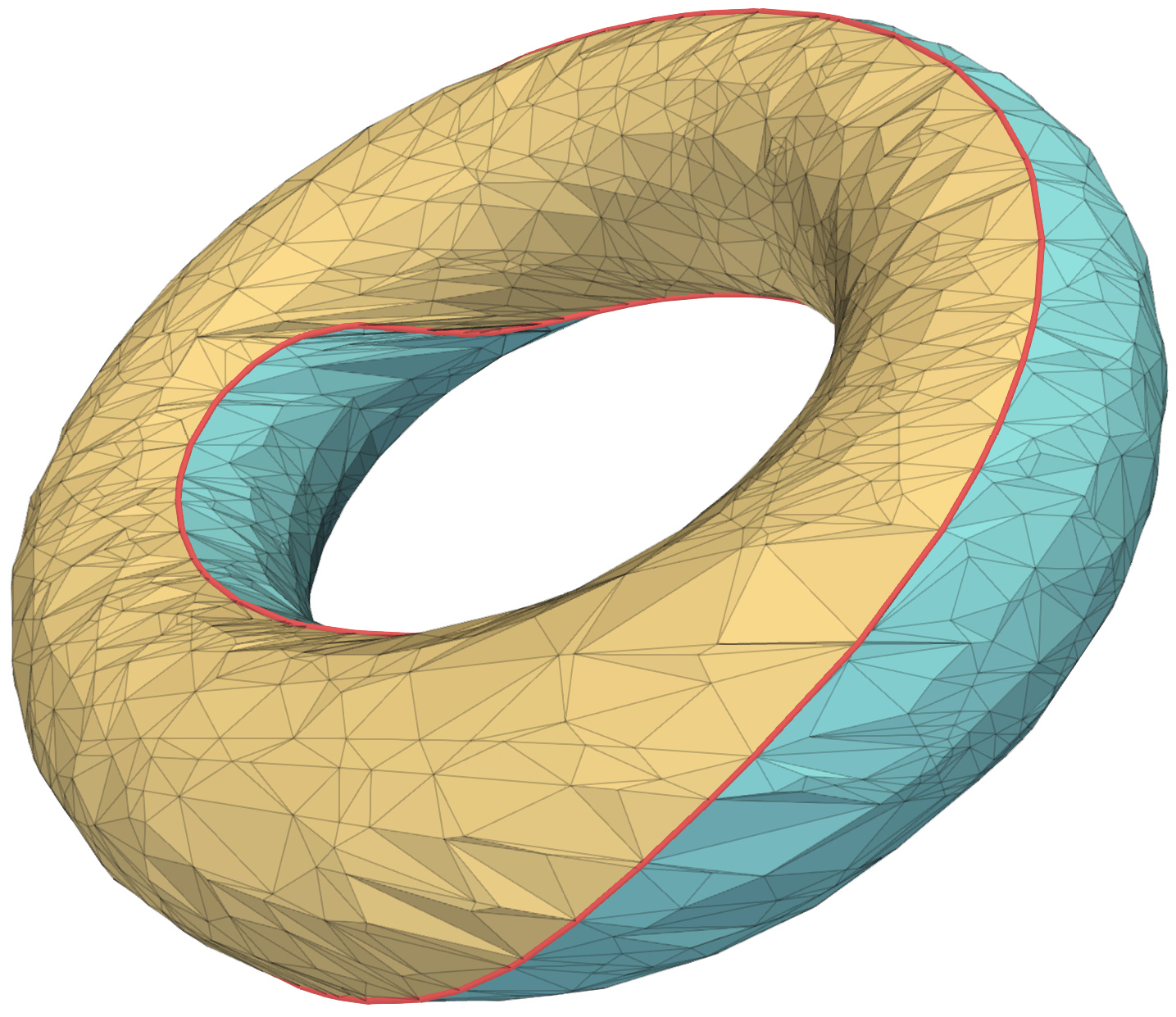}
\caption{Finer-resolution meshes may be obtained by increasing the number of initial subdivision levels. Here the torus was subdivided twice initially, as compared to once in Figure \ref{fig:results}.
\label{fig:subdivided_results}
}
\end{figure}

\paragraph{Dataset tests.}
In order to test the robustness of our method, we gathered 35 meshes from various sources. Most of the meshes are quad meshes, some including isolated triangles, and a few are purely triangle meshes.
For each model, we set up 26 camera views, equally spaced around the model in a turntable configuration. Additionally, we obtained the three non-proprietary animation sequences used by B\'enard et al.~\shortcite{Benard:2014} (Angela, bunny, walking man). 
Together, the turntable sequences and animations comprise 1580 distinct model/view combinations. Our implementation obtains correct WSOH results for each one, with at most four levels of subdivision.
Computation times and robustness are reported in Table \ref{fig:plots}. In nearly all cases, our algorithm requires less than 2 minutes to complete, often much less for smaller meshes. Computation times and output density depend significantly on the number of subdivision levels selected by the algorithm (Section \ref{sec:insertion}). In some cases, the number of output triangles is substantially lower than on the input mesh; additional vertices could easily be inserted if desired. 
\begin{figure*}
    \centering
    \includegraphics[width=3.3in]{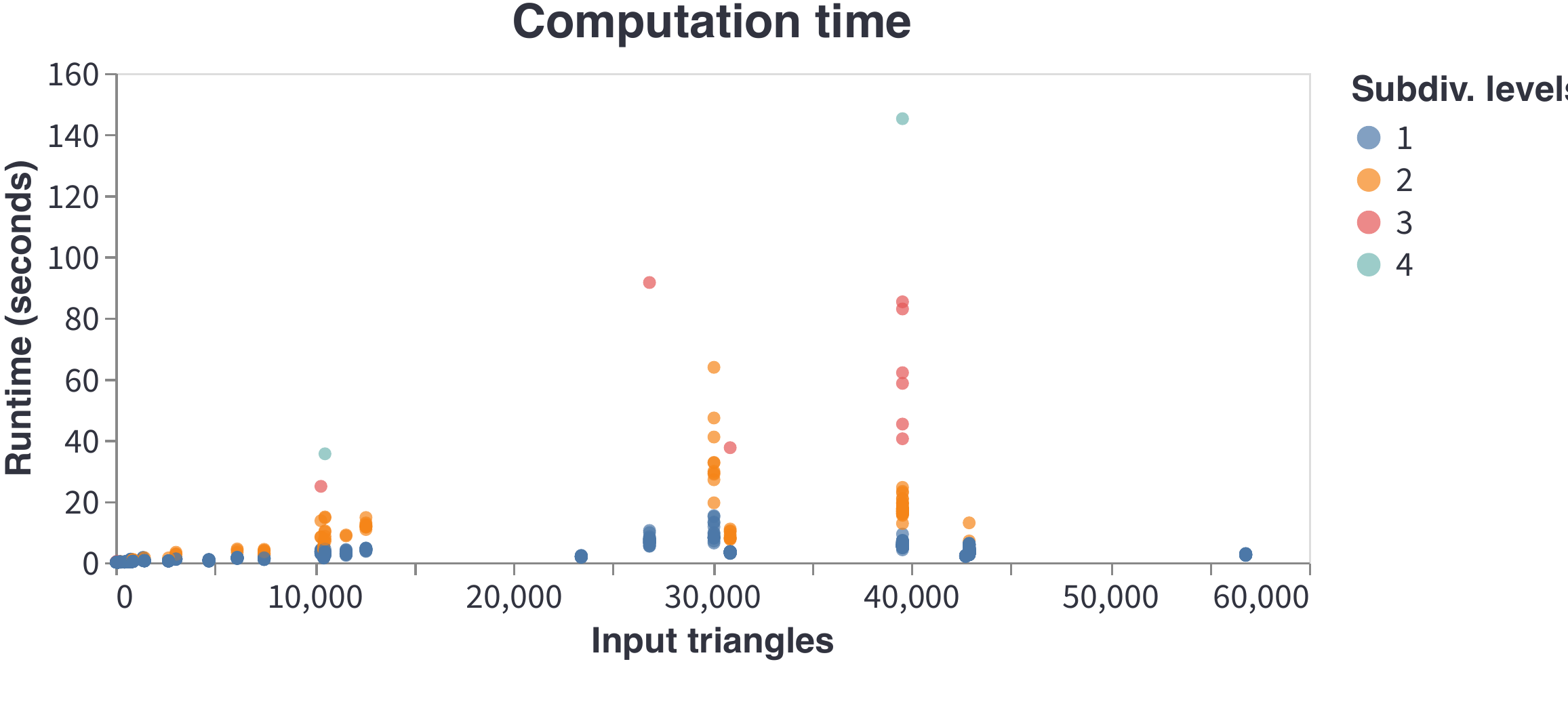}
    \includegraphics[width=3.3in]{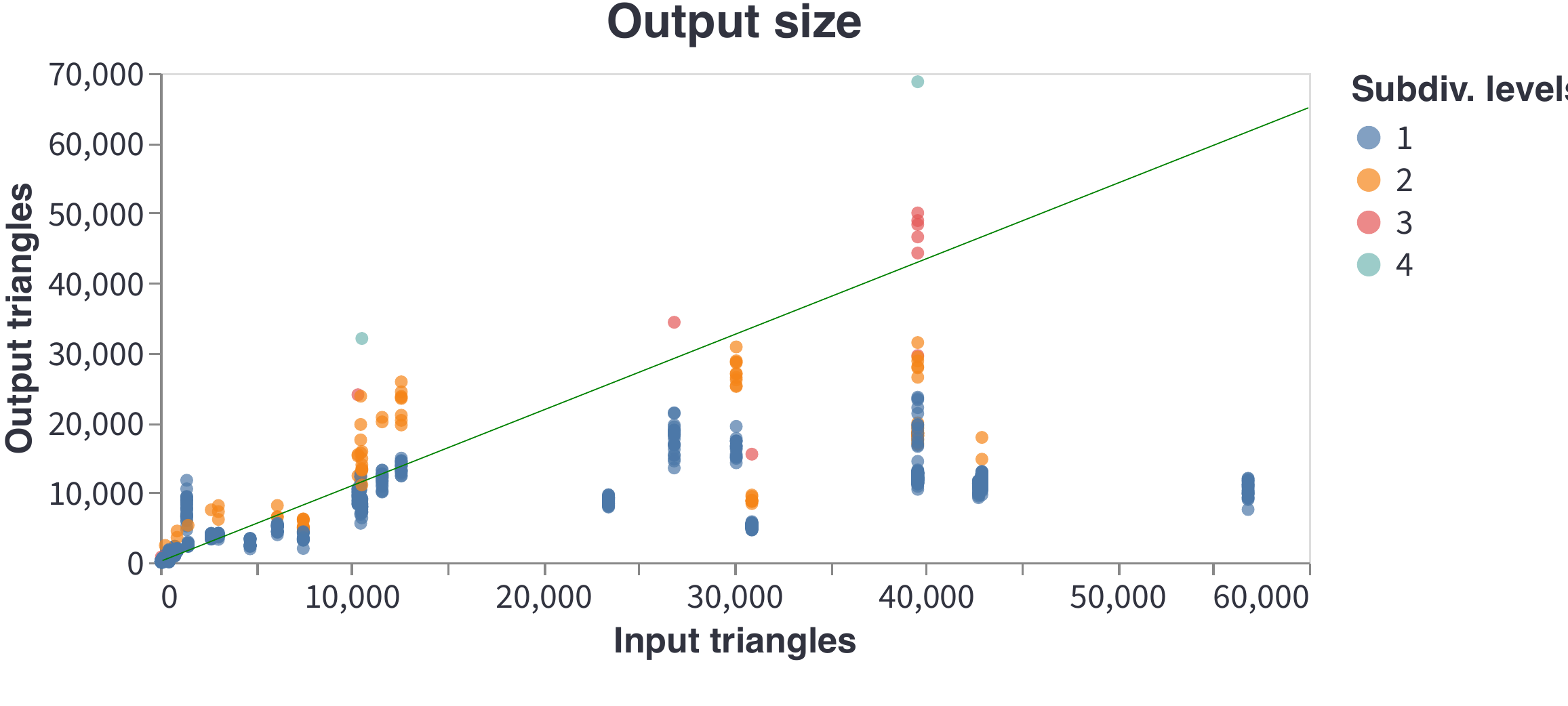}    
    \caption{Computation time and output size for each of the inputs. These results are for generating front-faces only; output triangles for back-facing regions are not generated in these plots, since they would almost never be used in practice.
    Each dot represents one or more of the 1580 test cases (mesh, camera), with the number of subdivision rounds required for that case color-coded. The vertical strip of dots at 39,576 input triangles are different tests on the Angela mesh.
    The green diagonal line shows where unity values would occur on the plot (i.e., one input for one output).
    These computations were performed on a MacBook Pro M1, 3.2Ghz, 16Gb memory.
    }
        \label{fig:plots}
\end{figure*}
Our heuristics were developed on this test set, and so more subdivision levels and time may be required for other models.

In order to test with a very challenging model, we separately tested with the genus-131 model ``Yeah, Right.'' Results for two viewpoints are shown in Figure \ref{fig:yeahright}. Due to the complexity of the model, we ran a maximum of 3 subdivision levels, and the method succeeded in 21 of 26 viewpoints; the average run-time for successful views was 62 minutes. 
In contrast, the method of B\'enard et al.~\shortcite{Benard:2014} failed to produce a fully-consistent mesh on any viewpoint, averaging 86 inconsistent triangles per frame.
As illustrated in the figure, our method produces valid visibility despite the exceedingly complex topology in both 2D and 3D.
\newcommand{\yrheight}{4in}

\begin{figure*}
    \centering
    \includegraphics[height=\yrheight]{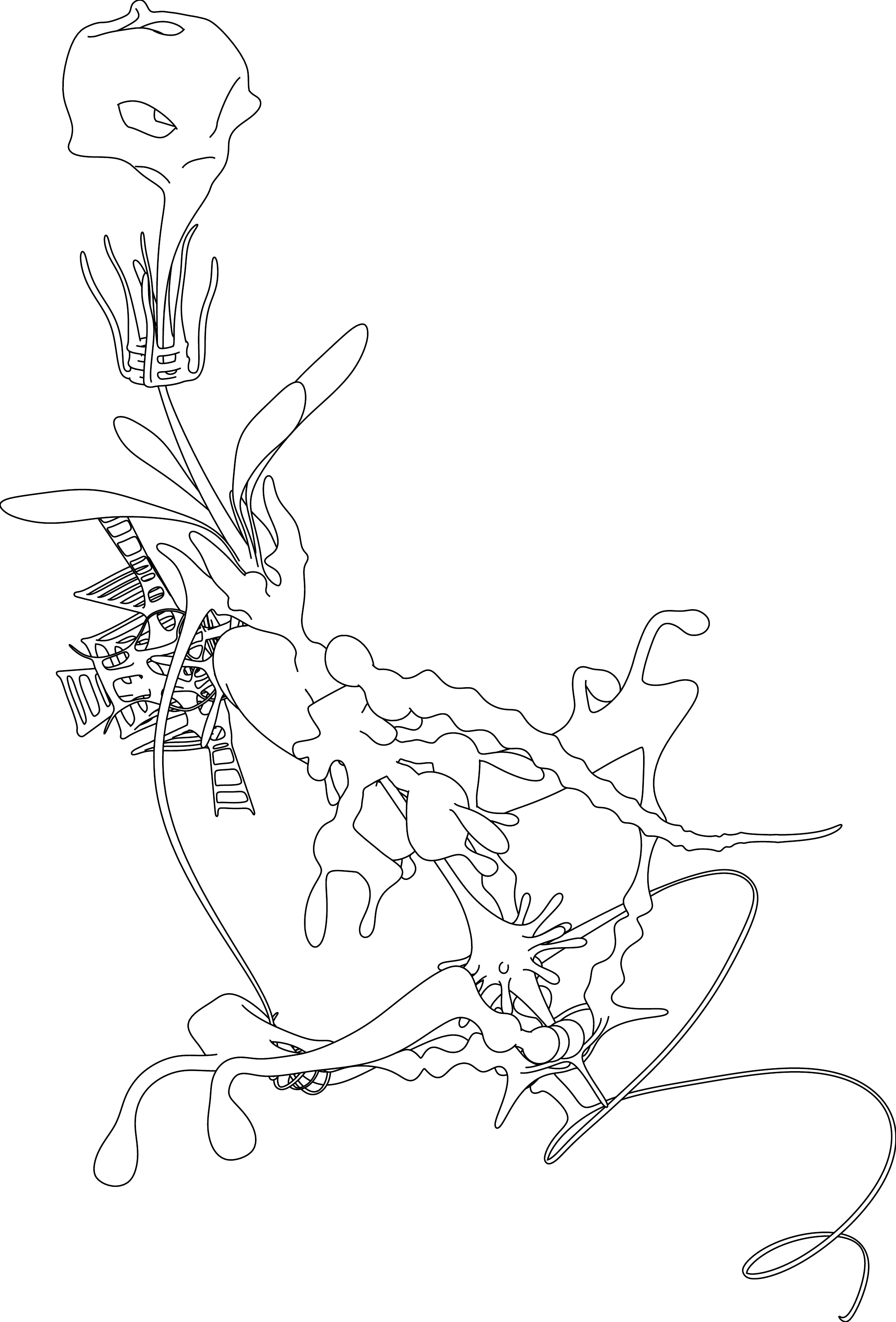}
        \includegraphics[height=\yrheight]{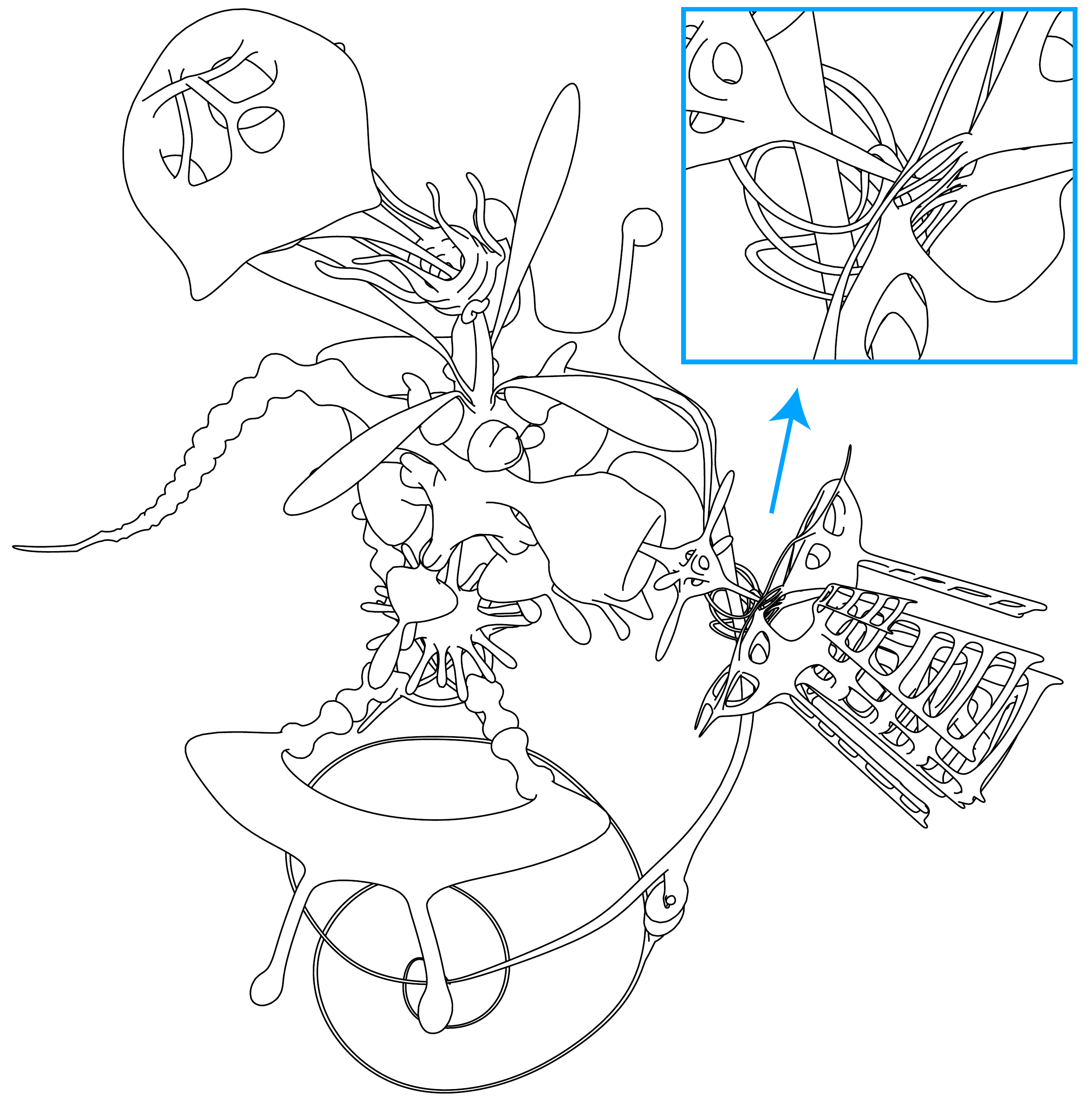}
\caption{Two views of the complex ``Yeah, Right'' model, which has genus 131. (Public domain model by Keenan Crane.)}
    \label{fig:yeahright}
\end{figure*}

\begin{figure}
\centering
\includegraphics[width=3in]{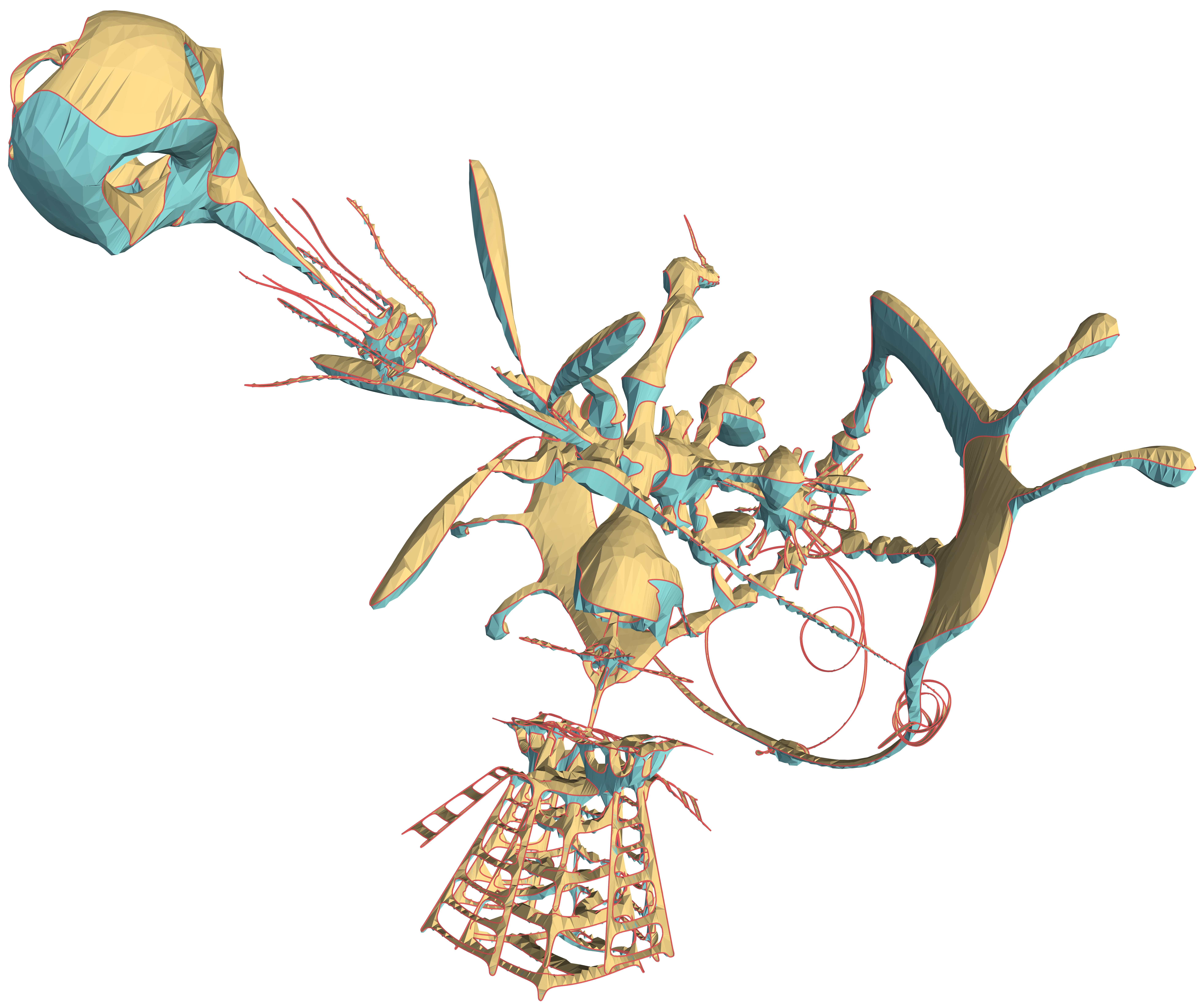}
\caption{Side view of the output triangulation of ``Yeah, Right,'' for the leftmost rendering in Figure \ref{fig:yeahright}}
\label{fig:yeahrightside}
\end{figure}

\paragraph{Disabling twist heuristics.}
We also experimented with running our method without the twist-removal heuristics, with a limit of 5 subdivision levels. The method successfully obtained WSOH results in 96\% of the cases, but with greater computation times, sometimes taking many hours. It is possible that the remaining cases would have succeeded at higher subdivision levels.

\paragraph{Comparison to state-of-the-art.}
We compare to B\'enard et al.~\shortcite{Benard:2014}'s statistics in Table \ref{table:benard}, using the three available animation sequences for which they reported numerical results; the fourth, ``Red'' was proprietary.
Whereas that method produced a handful of inconsistent faces for each mesh, our method produces perfectly consistent meshes. Moreover, our method operates an order-of-magnitude faster: $10\times$ on ``Stanford Bunny'', $6\times$ on ``Angela's face'', and $13\times$ on ``Walking Man''. It also produces roughly half as many output triangles, making our output more compact (more triangles may easily be added, if desired). 
We also believe our method is simpler to understand and simpler to implement, and the WSOH insights here will lead to more elegant algorithms in the future. 

\paragraph{How common are invalid polygons?}
As explained in Section \ref{sec:contour_regions}, prior methods can fail if they sample invalid polygons.  How common are invalid polygons? Our method uses extensive root-finding and careful sampling to compute polygons, and yet our method still frequently requires multiple subdivision levels and twist-removal heuristics in order to find valid WSOH regions. B\'enard et al.~\shortcite{Benard:2014} use similarly-careful root-finding procedures, but do no WSOH checks; their method always has at least a handful of invalid triangles.  Other methods sample the contour far less carefully, e.g., \cite{Hertzmann:2000}. Based on these experimental observations, \textbf{we believe that all previous methods frequently produce invalid polygons.} 

\begin{table*}[t]
\centering
    \begin{tabular}{|lcc|ccc|ccc|}
      \hline
      &&&\multicolumn{3}{c|}{\bf B\'{e}nard et al.~\shortcite{Benard:2014}} &
      \multicolumn{3}{c|}{\bf Our method} \\ \hline
      \multirow{2}{*}{Sequence} & Input & \multirow{2}{*}{Frames} &
      Output & Inconsistent & Time &
      Output & Inconsistent & Time 
      \\
      &faces&&faces&faces&(sec.)&faces&faces&(sec.)\\
      \hline
      Stanford Bunny & 42,928 & 400 &
      \quad $51,336 \pm 1,383$ & $7.6 \pm 4.4$ & $120 \pm35$ &
      $23,424 \pm518.1$ & 0 & $12 \pm2.2$\\
      Angela's face & 39,576 & 150 &
      \, $50,907 \pm693$ & $14 \pm4.4$& $170 \pm42$ &
      $26,642 \pm6,286$ & 0 & $28 \pm22$ \\
      Walking Man & 30,912 & 120 &
      $158,184 \pm663$ & $2.3 \pm2.1$ & $160 \pm15$ &
      $10,825 \pm2,620$ & 0  & $12 \pm10$\\
      \hline
    \end{tabular}
\caption{Statistics of our mesh generation algorithm on three of the animation sequences tested by \cite{Benard:2014}. The ``Red'' animation is omitted because it is proprietary.
The numbers are averaged over animation frames, which we list together with the standard deviation in each case. 
Our method produces no inconsistencies, far fewer output faces (roughly half), and runs much faster ($6\times$ to $13\times$ speedup). These values are for generating front and back faces, though back-faces would almost never be generated in practice.  Details: Both methods were run with a single thread on the same MacBook Pro (3.1GHz Intel Core i5 CPU, 8GB of memory).  The number of input faces is after one round of subdivision, and computation times are for mesh generation only, not stylization. We report fewer input faces for ``Walking Man'' because our method using 1 subdivision level as the minimum subdivision level, rather than 2 subdivisions used by B\'enard et al.}
\label{table:benard}
\end{table*}

\section{Discussion}
\label{sec:discussion}

The problem of computing visible occluding contours for smooth surfaces dates back to Weiss \shortcite{Weiss:1966:VPI:321328.321330}; we have shown, for the first time, how to characterize valid contours. Based on these insights, we have presented an algorithm that achieves state-of-the-art results on the problem.

Having a mathematical characterization of the space of valid solutions means that this problem is now in the domain of robust geometric computation. One important question is: how can we sample a contour curve in a way that guarantees a WSOH polygon? Simple strategies for refining the sampling seem like they ought to produce a WSOH curve eventually, but we do not have a proof.  Eliminating the need for twist and cut heuristics would also make the algorithm simpler. Some problems arise due to limitations of vertex insertion scheme and data structures, e.g., Fig 19 of \cite{Benard:2014}.

Our definition of valid triangulations does not, in itself, preserve depth ordering; instead, this is ensured by sampling all vertices from the smooth surface. We considered requiring preservation of contour convexity/concavity \bh{4.2,7.4}; while theoretically more elegant, it seemed unnecessarily complex in practice and potentially numerically sensitive. Likewise, we considered using Quantitative Invisibility (QI) \bh{4.7} to check WSO, following Eppstein and Mumford \shortcite{EppsteinMumford}. QI uses depth ordering constraints, making it potentially much more efficient than checking with triangulation. However, QI would need to be generalized to handle holes and convexity/concavity/cusps, which is potentially quite challenging in our case, but worth future study.

Our method can be applied to self-intersecting surfaces, although we have not tested this. That is, self-intersections do not need to be treated specially during triangulation; they can be detected on the output mesh during mesh contour extraction. However, the intersections produced by this method may be jagged; a smoothing step could be added, or the triangulations modified to accurately track the self-intersections of the subdivision surface. Our method also does not prevent spurious self-intersections, would could also be handled by an extra detection and mesh refinement step.

In our results, we do not make any effort to control mesh quality, as mesh quality is generally not important for line rendering, though it may be useful for other applications. Improving mesh quality would be straightforward, by adjusting the 2D sample points input to the CDT.

At present, our algorithm computes many quantities that may not be used in a final rendering, e.g., many of the triangles from the output mesh may not be needed. Lazy computations could improve efficiency. 

A more intriguing possibility is to compute an output planar map directly, rather than computing an intermediate mesh.    We chose to focus instead on mesh generation in the belief that it would give the most insight; our results show what the mesh $\cM$ looks like, and now future work can explore computing visibility without explicitly computing $\cM$. This will present several new challenges, such as accurately determining occlusion order without a mesh.  A mesh may still be needed in some regions, such as when there are self-intersections, and for some kinds of planar map rendering.

Finally, we wonder if it is possible to apply these ideas directly to a polygonal mesh, without any explicit smooth surface representation. For example, could we adjust interpolated contours \cite{Hertzmann:2000} to make them WSOH, thereby avoiding the complexity of smooth surface representations?

\section{Acknowledgements}

We are grateful to Denis Zorin and Qingnan Zhou for very helpful discussions, to Alec Jacobson and Danny Kaufman for comments on a draft, to Hanxiao Shen for providing a WSO implementation online, and to Alla Sheffer for support. P. B\'enard is supported in part by the ANR MoStyle project (ANR-20-CE33-0002). Thanks to Keenan Crane for sharing the pig, ogre, and Spot models, to Ryan Dale for the Walking Man, to Bay Raitt for Big Guy and Monster Frog, to AIM@Shape for Fertility, and 
 headus.com.au for Killeroo.

\bibliographystyle{ACM-Reference-Format}
\bibliography{contour_tutorial}


\begin{thebibliography}{26}

%%% the Web address.

\ifx \showCODEN    \undefined \def \showCODEN     #1{\unskip}     \fi
\ifx \showDOI      \undefined \def \showDOI       #1{#1}\fi
\ifx \showISBNx    \undefined \def \showISBNx     #1{\unskip}     \fi
\ifx \showISBNxiii \undefined \def \showISBNxiii  #1{\unskip}     \fi
\ifx \showISSN     \undefined \def \showISSN      #1{\unskip}     \fi
\ifx \showLCCN     \undefined \def \showLCCN      #1{\unskip}     \fi
\ifx \shownote     \undefined \def \shownote      #1{#1}          \fi
\ifx \showarticletitle \undefined \def \showarticletitle #1{#1}   \fi
\ifx \showURL      \undefined \def \showURL       {\relax}        \fi
% invisible to TeX
\providecommand\bibfield[2]{#2}
\providecommand\bibinfo[2]{#2}
\providecommand\natexlab[1]{#1}
\providecommand\showeprint[2][]{arXiv:#2}

\bibitem[Appel(1967)]%
        {Appel:1967}
\bibfield{author}{\bibinfo{person}{Arthur Appel}.}
  \bibinfo{year}{1967}\natexlab{}.
\newblock \showarticletitle{The Notion of Quantitative Invisibility and the
  Machine Rendering of Solids}. In \bibinfo{booktitle}{\emph{Proceedings of the
  1967 22nd National Conference}} \emph{(\bibinfo{series}{ACM '67})}.
  \bibinfo{publisher}{ACM}, \bibinfo{pages}{387--393}.
\newblock
\urldef\tempurl%
\url{https://doi.org/10.1145/800196.806007}
\showDOI{\tempurl}


\bibitem[B\'{e}nard and Hertzmann(2019)]%
        {BenardHertzmann}
\bibfield{author}{\bibinfo{person}{Pierre B\'{e}nard} {and}
  \bibinfo{person}{Aaron Hertzmann}.} \bibinfo{year}{2019}\natexlab{}.
\newblock \showarticletitle{Line Drawings from 3{D} Models}.
\newblock \bibinfo{journal}{\emph{Foundations and Trends in Computer Graphics
  and Vision}} \bibinfo{volume}{11}, \bibinfo{number}{1-2}
  (\bibinfo{year}{2019}), \bibinfo{pages}{1--159}.
\newblock
\urldef\tempurl%
\url{https://doi.org/10.1561/0600000075}
\showDOI{\tempurl}


\bibitem[B{\'e}nard et~al\mbox{.}(2014)]%
        {Benard:2014}
\bibfield{author}{\bibinfo{person}{Pierre B{\'e}nard}, \bibinfo{person}{Aaron
  Hertzmann}, {and} \bibinfo{person}{Michael Kass}.}
  \bibinfo{year}{2014}\natexlab{}.
\newblock \showarticletitle{Computing Smooth Surface Contours with Accurate
  Topology}.
\newblock \bibinfo{journal}{\emph{ACM Trans. Graph.}} \bibinfo{volume}{33},
  \bibinfo{number}{2}, Article \bibinfo{articleno}{19} (\bibinfo{year}{2014}),
  \bibinfo{numpages}{21}~pages.
\newblock
\urldef\tempurl%
\url{https://doi.org/10.1145/2558307}
\showDOI{\tempurl}


\bibitem[Chew(1989)]%
        {PaulChew1989}
\bibfield{author}{\bibinfo{person}{L.~Paul Chew}.}
  \bibinfo{year}{1989}\natexlab{}.
\newblock \showarticletitle{Constrained delaunay triangulations}.
\newblock \bibinfo{journal}{\emph{Algorithmica}} \bibinfo{volume}{4},
  \bibinfo{number}{1} (\bibinfo{date}{01 Jun} \bibinfo{year}{1989}),
  \bibinfo{pages}{97--108}.
\newblock
\showISSN{1432-0541}
\urldef\tempurl%
\url{https://doi.org/10.1007/BF01553881}
\showDOI{\tempurl}


\bibitem[Cole and Finkelstein(2010)]%
        {Cole:2010}
\bibfield{author}{\bibinfo{person}{Forrester Cole} {and} \bibinfo{person}{Adam
  Finkelstein}.} \bibinfo{year}{2010}\natexlab{}.
\newblock \showarticletitle{Two Fast Methods for High-Quality Line Visibility}.
\newblock \bibinfo{journal}{\emph{IEEE Transactions on Visualization and
  Computer Graphics}} \bibinfo{volume}{16}, \bibinfo{number}{5}
  (\bibinfo{year}{2010}), \bibinfo{pages}{707--717}.
\newblock
\urldef\tempurl%
\url{https://doi.org/10.1109/TVCG.2009.102}
\showDOI{\tempurl}


\bibitem[Cole et~al\mbox{.}(2008)]%
        {Cole:2008}
\bibfield{author}{\bibinfo{person}{Forrester Cole}, \bibinfo{person}{Aleksey
  Golovinskiy}, \bibinfo{person}{Alex Limpaecher},
  \bibinfo{person}{Heather~Stoddart Barros}, \bibinfo{person}{Adam
  Finkelstein}, \bibinfo{person}{Thomas Funkhouser}, {and}
  \bibinfo{person}{Szymon Rusinkiewicz}.} \bibinfo{year}{2008}\natexlab{}.
\newblock \showarticletitle{Where Do People Draw Lines?}
\newblock \bibinfo{journal}{\emph{ACM Trans. Graph.}} \bibinfo{volume}{27},
  \bibinfo{number}{3}, Article \bibinfo{articleno}{88} (\bibinfo{year}{2008}),
  \bibinfo{numpages}{11}~pages.
\newblock
\urldef\tempurl%
\url{https://doi.org/10.1145/1360612.1360687}
\showDOI{\tempurl}


\bibitem[Eisemann et~al\mbox{.}(2008)]%
        {Eisemann:2008}
\bibfield{author}{\bibinfo{person}{Elmar Eisemann}, \bibinfo{person}{Holger
  Winnem\"{o}ller}, \bibinfo{person}{John~C. Hart}, {and}
  \bibinfo{person}{David Salesin}.} \bibinfo{year}{2008}\natexlab{}.
\newblock \showarticletitle{Stylized Vector Art from 3D Models with Region
  Support}. In \bibinfo{booktitle}{\emph{Proceedings of the Nineteenth
  Eurographics Conference on Rendering}} \emph{(\bibinfo{series}{EGSR '08})}.
  \bibinfo{publisher}{Eurographics Association}, \bibinfo{pages}{1199--1207}.
\newblock
\urldef\tempurl%
\url{https://doi.org/10.1111/j.1467-8659.2008.01258.x}
\showDOI{\tempurl}


\bibitem[Elber and Cohen(1990)]%
        {Elber:1990}
\bibfield{author}{\bibinfo{person}{Gershon Elber} {and} \bibinfo{person}{Elaine
  Cohen}.} \bibinfo{year}{1990}\natexlab{}.
\newblock \showarticletitle{Hidden Curve Removal for Free Form Surfaces}. In
  \bibinfo{booktitle}{\emph{Proceedings of the 17th Annual Conference on
  Computer Graphics and Interactive Techniques}}
  \emph{(\bibinfo{series}{SIGGRAPH '90})}. \bibinfo{publisher}{ACM},
  \bibinfo{pages}{95--104}.
\newblock
\urldef\tempurl%
\url{https://doi.org/10.1145/97879.97890}
\showDOI{\tempurl}


\bibitem[Eppstein and Mumford(2009)]%
        {EppsteinMumford}
\bibfield{author}{\bibinfo{person}{David Eppstein} {and} \bibinfo{person}{Elena
  Mumford}.} \bibinfo{year}{2009}\natexlab{}.
\newblock \showarticletitle{Self-Overlapping Curves Revisited}. In
  \bibinfo{booktitle}{\emph{Proceedings 20th Annual ACM-SIAM Symposium on
  Discrete Algorithms (SODA'09)}}. \bibinfo{publisher}{SIAM},
  \bibinfo{pages}{160--169}.
\newblock


\bibitem[Grabli et~al\mbox{.}(2010)]%
        {Grabli:2010}
\bibfield{author}{\bibinfo{person}{St{\'e}phane Grabli},
  \bibinfo{person}{Emmanuel Turquin}, \bibinfo{person}{Fr{\'e}do Durand}, {and}
  \bibinfo{person}{Fran\c{c}ois~X. Sillion}.} \bibinfo{year}{2010}\natexlab{}.
\newblock \showarticletitle{Programmable Rendering of Line Drawing from 3D
  Scenes}.
\newblock \bibinfo{journal}{\emph{ACM Trans. Graph.}} \bibinfo{volume}{29},
  \bibinfo{number}{2}, Article \bibinfo{articleno}{18} (\bibinfo{year}{2010}),
  \bibinfo{numpages}{20}~pages.
\newblock
\urldef\tempurl%
\url{https://doi.org/10.1145/1731047.1731056}
\showDOI{\tempurl}


\bibitem[Gu et~al\mbox{.}(2002)]%
        {Gu:2002}
\bibfield{author}{\bibinfo{person}{Xianfeng Gu}, \bibinfo{person}{Steven
  Gortler}, {and} \bibinfo{person}{Hugues Hoppe}.}
  \bibinfo{year}{2002}\natexlab{}.
\newblock \showarticletitle{Geometry Images}.
\newblock \bibinfo{journal}{\emph{ACM Trans. Graphics}} \bibinfo{volume}{21},
  \bibinfo{number}{3} (\bibinfo{year}{2002}).
\newblock
\urldef\tempurl%
\url{https://doi.org/10.1145/566654.566589}
\showDOI{\tempurl}


\bibitem[Hertzmann and Zorin(2000)]%
        {Hertzmann:2000}
\bibfield{author}{\bibinfo{person}{Aaron Hertzmann} {and}
  \bibinfo{person}{Denis Zorin}.} \bibinfo{year}{2000}\natexlab{}.
\newblock \showarticletitle{Illustrating Smooth Surfaces}. In
  \bibinfo{booktitle}{\emph{Proceedings of the 27th Annual Conference on
  Computer Graphics and Interactive Techniques}}
  \emph{(\bibinfo{series}{SIGGRAPH '00})}. \bibinfo{publisher}{ACM
  Press/Addison-Wesley Publishing Co.}, \bibinfo{pages}{517--526}.
\newblock
\urldef\tempurl%
\url{https://doi.org/10.1145/344779.345074}
\showDOI{\tempurl}


\bibitem[Karpenko and Hughes(2006)]%
        {smoothsketch}
\bibfield{author}{\bibinfo{person}{Olga~A. Karpenko} {and}
  \bibinfo{person}{John~F. Hughes}.} \bibinfo{year}{2006}\natexlab{}.
\newblock \showarticletitle{SmoothSketch: 3D Free-Form Shapes from Complex
  Sketches}. In \bibinfo{booktitle}{\emph{ACM SIGGRAPH 2006 Papers}} (Boston,
  Massachusetts) \emph{(\bibinfo{series}{SIGGRAPH ’06})}.
  \bibinfo{publisher}{Association for Computing Machinery},
  \bibinfo{address}{New York, NY, USA}, \bibinfo{pages}{589–598}.
\newblock
\showISBNx{1595933646}
\urldef\tempurl%
\url{https://doi.org/10.1145/1179352.1141928}
\showDOI{\tempurl}


\bibitem[Karsch and Hart(2011)]%
        {Karsch:2011}
\bibfield{author}{\bibinfo{person}{Kevin Karsch} {and} \bibinfo{person}{John~C.
  Hart}.} \bibinfo{year}{2011}\natexlab{}.
\newblock \showarticletitle{Snaxels on a Plane}. In
  \bibinfo{booktitle}{\emph{Proceedings of the ACM SIGGRAPH/Eurographics
  Symposium on Non-Photorealistic Animation and Rendering}}
  \emph{(\bibinfo{series}{NPAR '11})}. \bibinfo{publisher}{ACM},
  \bibinfo{pages}{35--42}.
\newblock
\urldef\tempurl%
\url{https://doi.org/10.1145/2024676.2024683}
\showDOI{\tempurl}


\bibitem[Lacewell and Burley(2007)]%
        {lacewell}
\bibfield{author}{\bibinfo{person}{Dylan Lacewell} {and} \bibinfo{person}{Brent
  Burley}.} \bibinfo{year}{2007}\natexlab{}.
\newblock \showarticletitle{Exact Evaluation of Catmull-Clark Subdivision
  Surfaces Near B-Spline Boundaries}.
\newblock \bibinfo{journal}{\emph{Journal of Graphics Tools}}
  \bibinfo{volume}{12}, \bibinfo{number}{3} (\bibinfo{year}{2007}),
  \bibinfo{pages}{7--15}.
\newblock


\bibitem[Li and Barbi\v{c}(2018)]%
        {Li:2018:IOS}
\bibfield{author}{\bibinfo{person}{Yijing Li} {and} \bibinfo{person}{Jernej
  Barbi\v{c}}.} \bibinfo{year}{2018}\natexlab{}.
\newblock \showarticletitle{Immersion of Self-Intersecting Solids and
  Surfaces}.
\newblock \bibinfo{journal}{\emph{ACM Trans. Graph.}} \bibinfo{volume}{37},
  \bibinfo{number}{4}, Article \bibinfo{articleno}{45} (\bibinfo{year}{2018}),
  \bibinfo{numpages}{14}~pages.
\newblock
\urldef\tempurl%
\url{https://doi.org/10.1145/3197517.3201327}
\showDOI{\tempurl}


\bibitem[Markosian et~al\mbox{.}(1997)]%
        {Markosian:1997}
\bibfield{author}{\bibinfo{person}{Lee Markosian}, \bibinfo{person}{Michael~A.
  Kowalski}, \bibinfo{person}{Daniel Goldstein}, \bibinfo{person}{Samuel~J.
  Trychin}, \bibinfo{person}{John~F. Hughes}, {and} \bibinfo{person}{Lubomir~D.
  Bourdev}.} \bibinfo{year}{1997}\natexlab{}.
\newblock \showarticletitle{Real-time Nonphotorealistic Rendering}. In
  \bibinfo{booktitle}{\emph{Proceedings of the 24th Annual Conference on
  Computer Graphics and Interactive Techniques}}
  \emph{(\bibinfo{series}{SIGGRAPH '97})}. \bibinfo{publisher}{ACM
  Press/Addison-Wesley Publishing Co.}, \bibinfo{pages}{415--420}.
\newblock
\urldef\tempurl%
\url{https://doi.org/10.1145/258734.258894}
\showDOI{\tempurl}


\bibitem[Northrup and Markosian(2000)]%
        {Northrup:2000}
\bibfield{author}{\bibinfo{person}{J.~D. Northrup} {and} \bibinfo{person}{Lee
  Markosian}.} \bibinfo{year}{2000}\natexlab{}.
\newblock \showarticletitle{Artistic Silhouettes: A Hybrid Approach}. In
  \bibinfo{booktitle}{\emph{Proceedings of the 1st International Symposium on
  Non-photorealistic Animation and Rendering}} \emph{(\bibinfo{series}{NPAR
  '00})}. \bibinfo{publisher}{ACM}, \bibinfo{pages}{31--37}.
\newblock
\urldef\tempurl%
\url{https://doi.org/10.1145/340916.340920}
\showDOI{\tempurl}


\bibitem[Roberts(1963)]%
        {Roberts:1963}
\bibfield{author}{\bibinfo{person}{Lawrence Roberts}.}
  \bibinfo{year}{1963}\natexlab{}.
\newblock \emph{\bibinfo{title}{Machine Perception of Three-Dimensional
  Solids}}.
\newblock \bibinfo{thesistype}{Ph.\,D. Dissertation}.
  \bibinfo{school}{Massachusetts Institute of Technology. Dept. of Electrical
  Engineering}.
\newblock


\bibitem[Sacht et~al\mbox{.}(2013)]%
        {Sacht:SelfIntersectingVolumes:2013}
\bibfield{author}{\bibinfo{person}{Leonardo Sacht}, \bibinfo{person}{Alec
  Jacobson}, \bibinfo{person}{Daniele Panozzo}, \bibinfo{person}{Christian
  Sch{\"u}ller}, {and} \bibinfo{person}{Olga Sorkine-Hornung}.}
  \bibinfo{year}{2013}\natexlab{}.
\newblock \showarticletitle{Consistent Volumetric Discretizations Inside
  Self-Intersecting Surfaces}.
\newblock \bibinfo{journal}{\emph{Computer Graphics Forum (Proc.~SGP)}}
  \bibinfo{volume}{32}, \bibinfo{number}{5} (\bibinfo{year}{2013}),
  \bibinfo{pages}{147--156}.
\newblock


\bibitem[Saito and Takahashi(1990)]%
        {Saito:1990}
\bibfield{author}{\bibinfo{person}{Takafumi Saito} {and}
  \bibinfo{person}{Tokiichiro Takahashi}.} \bibinfo{year}{1990}\natexlab{}.
\newblock \showarticletitle{Comprehensible Rendering of 3-D Shapes}. In
  \bibinfo{booktitle}{\emph{Proceedings of the 17th Annual Conference on
  Computer Graphics and Interactive Techniques}}
  \emph{(\bibinfo{series}{SIGGRAPH '90})}. \bibinfo{publisher}{ACM},
  \bibinfo{pages}{197--206}.
\newblock
\urldef\tempurl%
\url{https://doi.org/10.1145/97879.97901}
\showDOI{\tempurl}


\bibitem[Shor and {Van Wyk}(1992)]%
        {ShorVanWyk}
\bibfield{author}{\bibinfo{person}{Peter~W. Shor} {and}
  \bibinfo{person}{Christopher~J. {Van Wyk}}.} \bibinfo{year}{1992}\natexlab{}.
\newblock \showarticletitle{Detecting and decomposing self-overlapping curves}.
\newblock \bibinfo{journal}{\emph{Computational Geometry}} \bibinfo{volume}{2},
  \bibinfo{number}{1} (\bibinfo{date}{Aug.} \bibinfo{year}{1992}),
  \bibinfo{pages}{31--50}.
\newblock
\urldef\tempurl%
\url{https://doi.org/10.1016/0925-7721(92)90019-O}
\showDOI{\tempurl}


\bibitem[Weber and Zorin(2014)]%
        {WeberZorin}
\bibfield{author}{\bibinfo{person}{Ofir Weber} {and} \bibinfo{person}{Denis
  Zorin}.} \bibinfo{year}{2014}\natexlab{}.
\newblock \showarticletitle{Locally injective parametrization with arbitrary
  fixed boundaries}.
\newblock \bibinfo{journal}{\emph{ACM Trans. Graph.}} \bibinfo{volume}{33},
  \bibinfo{number}{4} (\bibinfo{year}{2014}).
\newblock
\showISSN{0734-2071}
\urldef\tempurl%
\url{https://doi.org/10.1145/2601097.2601227}
\showDOI{\tempurl}


\bibitem[Weiss(1966)]%
        {Weiss:1966:VPI:321328.321330}
\bibfield{author}{\bibinfo{person}{Ruth~A. Weiss}.}
  \bibinfo{year}{1966}\natexlab{}.
\newblock \showarticletitle{BE VISION, A Package of IBM 7090 FORTRAN Programs
  to Draw Orthographic Views of Combinations of Plane and Quadric Surfaces}.
\newblock \bibinfo{journal}{\emph{J. ACM}} \bibinfo{volume}{13},
  \bibinfo{number}{2} (\bibinfo{year}{1966}), \bibinfo{pages}{194--204}.
\newblock
\showISSN{0004-5411}
\urldef\tempurl%
\url{https://doi.org/10.1145/321328.321330}
\showDOI{\tempurl}


\bibitem[Winkenbach and Salesin(1994)]%
        {Winkenbach:1994}
\bibfield{author}{\bibinfo{person}{Georges Winkenbach} {and}
  \bibinfo{person}{David~H. Salesin}.} \bibinfo{year}{1994}\natexlab{}.
\newblock \showarticletitle{Computer-generated Pen-and-ink Illustration}. In
  \bibinfo{booktitle}{\emph{Proceedings of the 21st Annual Conference on
  Computer Graphics and Interactive Techniques}}
  \emph{(\bibinfo{series}{SIGGRAPH '94})}. \bibinfo{publisher}{ACM},
  \bibinfo{pages}{91--100}.
\newblock
\urldef\tempurl%
\url{https://doi.org/10.1145/192161.192184}
\showDOI{\tempurl}


\bibitem[Winkenbach and Salesin(1996)]%
        {Winkenbach:1996}
\bibfield{author}{\bibinfo{person}{Georges Winkenbach} {and}
  \bibinfo{person}{David~H. Salesin}.} \bibinfo{year}{1996}\natexlab{}.
\newblock \showarticletitle{Rendering Parametric Surfaces in Pen and Ink}. In
  \bibinfo{booktitle}{\emph{Proceedings of the 23rd Annual Conference on
  Computer Graphics and Interactive Techniques}}
  \emph{(\bibinfo{series}{SIGGRAPH '96})}. \bibinfo{publisher}{ACM},
  \bibinfo{pages}{469--476}.
\newblock
\urldef\tempurl%
\url{https://doi.org/10.1145/237170.237287}
\showDOI{\tempurl}


\end{thebibliography}

\appendix
\section{Projection Preserves Triangle Orientation}
\label{app:orient}

We now show that the orientation of a triangle is preserved when projecting it from the image plane to 3D.   Suppose we define coordinates so that the camera is at the origin $\bc=(0,0,0)$, facing in the $+z$ direction. The vertices of a triangle $\Delta \bp\bq\br$ have image coordinates $\bp=(p_x,p_y,1)$, and so on. The orientation of the triangle is given by
$(\bc-\bp) \cdot ((\br-\bp)\times(\bq-\bp)) =
\det(-\bp, \br-\bp, \bq-\bp) = 
\det(\bp, \bq, \br)$.   For clockwise 2D triangles, the orientation is positive: $\det(\bp, \bq, \br) > 0$

We construct the 3D points by selecting depths $p_z,q_z,r_z >0$, and then projecting into 3D: $\bp'=p_z \bp, \bq' = q_z \bq_z, \br' = r_z \br$. The orientation of this new triangle is $(\bc-\bp') \cdot ((\br'-\bp')\times(\bq'-\bp')) = \det(\bp', \bq', \br') = \det(p_z \bp, q_z \bq_z, r_z \br) = p_z q_z r_z \det(\bp, \bq, \br) > 0$. Thus, the orientation of the 3D triangle is the same as the orientation of the 2d triangle: clockwise triangles become front-facing and counter-clockwise triangles become back-facing.

\section{Removing Twists}
\label{app:correcting}

We observe the following types of twists (Figure \ref{fig:twist_cases}), and use the following heuristics to correct them:
\begin{enumerate}
    \item A very small polygon comprises a simple loop with the wrong orientation (e.g., CW when it should be CCW). In this case, we simply remove the loop from the surface. These polygons are typically smaller than a pixel in image-space, and would be removed during stylization regardless.
    
    \item A pair of nearby cusps on a contour is on the wrong side of the contour forming an inverted fishtail. This configuration is detected twice; we check for these before and after testing the other four cases.

    \item The curve is twisted near a cusp/singularity. We correct this by inserting vertices on each curve at the intersection point, and moving the new vertices apart in image space. 
    
    \item A skinny portion of the polygon crosses over itself. We  insert vertices at the intersection points and shift them apart in image space.
    
    \item A cusp on a hole occurs outside the hole. We shift the cusp in image space. 
    
\end{enumerate}
Each of these cases applies only to the contours within a region, e.g., intersections are not detected between regions.

\begin{figure}
\centering
Case 1, Tiny hole:\\
\includegraphics[width=3in]{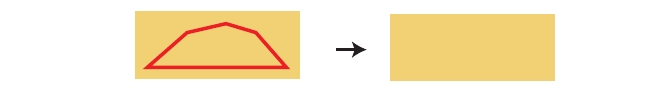}\\
Case 2, Inverted fishtail:\\
\includegraphics[width=3in]{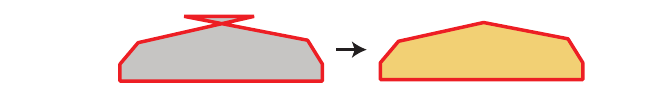}\\
Case 3, Twisted cusp: \\
\includegraphics[width=3in]{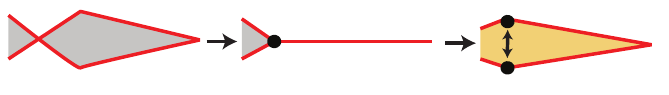}\\
Case 4, Twisted tube:\\
\includegraphics[width=3in]{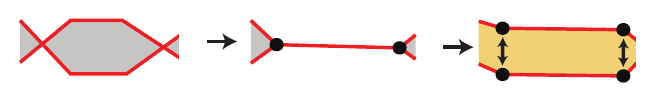}\\
Case 5, Outside cusp:\\
\includegraphics[width=3in]{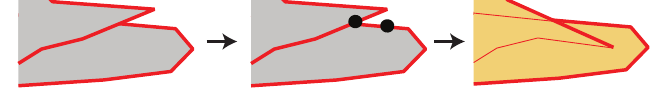}
\caption{The five cases of twists that we detect and resolve; see text for details.
\label{fig:twist_cases}
}
\end{figure}

\paragraph{Case 1.} 
If a patch is simple but fails the WSO triangulation, our algorithm removes it if the 2D contour length is below $100 \text{px}$. Otherwise, the patch will be left unchanged, and the algorithm will need to be rerun with a higher subdivision level. 
To remove a loop, the corresponding faces are assigned to the adjacent patch and the corresponding contour edges are unlabelled.

\paragraph{Case 2.} 
Case 2 is handled by collapsing any loop formed by a 2D intersection (valid or not) that attaches a 2D contour loop whose 2D arc-length is less than $10^{-3} \text{px}$.  
Note that this removal step may delete valid fishtails \bh{4.6}, but these are subpixel fishtails that would normally be removed during stylization \bh{9.3}.

\paragraph{Cases 3, 4, 5.} These three cases are detected as follows.
First, the algorithm finds all intersections where two contour edges intersect in image space. The following steps are then run on each intersection separately.

An image-space intersection comprises two distinct 3D points on the surface.  We first wish to determine if the intersection points are directly connected by geometry, which indicates a twist we may wish to remove. Specifically, we compute a 3D plane that contains the two intersection points and the bisector of the larger angle between the two intersecting edges in 2D. Then we intersect the plane with the region and check if this intersection creates a path between the two intersection points that stays within the region (i.e., does not cross any contours).  If the path is valid, then we conclude that there is a twist at this intersection. 

The next step is to determine which type of twist occurs. From each intersection point, we can trace around the contour loop until returning to the intersection point. The tracing direction is the one that makes the adjacent front facing patch lie on the left side of the contour.
If the sub-loop traced in this process has the wrong orientation (CW or CCW), then that sub-loop has an error. The case depends on how this tracing returns to the intersection:
(a) Case 3: return back to the intersection point via a crossing edge.
(b) Case 4: meet another invalid intersection point.

While case 3 and 4 are distinguished by whether the tracing returns or reaches another invalid intersection point, case 5 is separated from case 3 and 4 by  heuristics.
We have the observation that the twisting is caused by under-sampling and thus is supposed to have a small scale.

Let the tracing distance in 2D from an invalid intersection point to either another invalid intersection or itself be $D_1$ and $D_2$ where $D_1 = D_2$ in case 5.
Let the total lengths of the 2D contour chains containing $D_1$ and $D_2$ be $c_1$ and $c_2$ respectively.
If both tracing paths belong to the same chain, we have $c_1 = c_2$.
Since the tracing directions determined by the orientation rule could be incorrect if the under-sampling causes a fishtail to untwist into a Z-shape structure, we also consider the flipped tracing directions.
Let the tracing distance of the flipped tracing directions be $D_1'$ and $D_2'$ respectively.
Let the tracing distance in 2D from the invalid intersection point to the corresponding cusp and its projection be $d_1$ and $d_2$.
We categorize the twisting as case 3 if all the following conditions are satisfied:
\begin{itemize}
    \item Scale violation in case 3 and 4: $\frac{D_1}{c_1} > 0.9 \text{ OR } D_1 > 100 \text{px}$ OR $\frac{D_2}{c_2} > 0.9 \text{ OR } D_2 > 100 \text{px}$.
    \item Contrast between the flipped tracing directions and case 5: $\frac{\max{(D_1', D_2')}}{\min{(d_1, d_2)}} > 4$ OR $\min{(d_1, d_2)} < 2 \text{px}$.
    \item Safety check in case 5: $\frac{d_1}{c_1} < 0.5 \text{ AND } d_1 < 40 \text{px}$ AND $\frac{d_2}{c_2} < 0.5 \text{ AND } d_2 < 40 \text{px}$.
\end{itemize}
If the first condition is satisfied yet either of the latter two is not, we choose to trace in the flipped directions.
If the first condition is not satisfied, we trace in the directions initially determined by the orientation rule.

\section{Cut-to-disk algorithm}
\label{app:cuts}

We now describe our modification to Gu et al.~\shortcite{Gu:2002}'s cut-to-disk algorithm to avoid cuts that pass outside the region boundary in image space. 

We first mark the following sets of edges that the cut should not pass through: (a) any edge adjacent to any inconsistent face, (b) any edge adjacent to a cusp, and (c)  any edge that intersects a nearby contour in image-space.
For this last case, we test for intersections between each edge, and an contour edge within its 5-ring, as well as all other contour edges with 5 edges of the contour edge.

For each contour loop in the region, we check if the contour loop has no valid edges coming out of it. In this case, we find a consistent triangle that touches the contour at a single vertex, and split this triangle to produce a usable edge. 

We then modify the cut-to-disk algorithm to keep the above edges out of the cut. Specifically, the initial seeds are all faces adjacent to any of the above edges. We keep track of the connected components of faces during region growing. If two regions become adjacent, then they are merged and the edge between them is removed rather than becoming part of the cut.
Finally, we apply the same shrinking step from Gu et al.

The initial seeds are all faces adjacent to any invalid edge.
Each seed is assigned a unique label and non-seed faces is viewed as unassigned.
To indicate the cut result, each edge has a flag showing whether it is in the cut.

In the first phase, the method grows and merges labeled face regions by processing edges adjacent to any assigned face.
These edges are stored in a priority queue with three priority levels: $0$ to $2$, where $2$ is the highest one.
The priority of an edges is determined based on its two neighboring faces and its validity: (a) priority 2: if the two faces are assigned and the edge is invalid; (b) priority 1: if one of the two faces is unassigned; (c) priority 0: if the two faces are assigned and the edge is valid.
The method iteratively draws from the priority queue until the queue is empty.
An edge is processed as follows,
\begin{itemize}
    \item If the edge has priority 0 or 2, and the two faces have the same label, the edge is in the cut; if the two faces have difference labels, the two labels are merged and the edge is not in the cut.
    \item If the edge has priority 1, the unassigned face is set to the label of the other face and the edge is set to be not in the cut. The other two edges of the newly assigned face, if have not been processed, are added to the queue with priorities determined as above.
\end{itemize}
The special case is the boundary edge.
They are always set to be in the cut and the region growing stops once reaches these edges.

In the second phase, the method repeatedly removes cut edges adjacent to valence-1 vertex in the same way as the second phase of Gu et al. 2002.
Intuitively, this phase removes tree structures from the cut.

This algorithm finds a cut fully consisting of valid edges if it exists, or finds a cut containing invalid edges otherwise.
This is ensured by the initialization and the priority order.
Note that the initialization enqueues all invalid edges with top priorities and thus they are processed before any other types of edges.
An invalid edge would be in the cut if and only if its two neighboring faces have the same label.
Since we assign each such face a unique label, this situation would only happen when the corresponding labeled region has a non-disk topology and only has invalid edges in its interior.
This means there exists no cut that only consists of valid edges.
Once the top priority edges are fully processed, the method continues in a fashion similar to the original Gu et al. 2002.
Finally, we apply the same shrinking step from Gu et al.

It is possible that there is a contour loop with no adjacent consistent triangles, and so there is no feasible cut. In this case, the above algorithm will find a cut that touches the loop via an invalid edge.
We use three priority levels in the cut-to-disk priority queue, in order to ensure that some cut is found.

\end{document}